%% file: main.tex
\documentclass[acmsmall]{acmart}
\AtBeginDocument{%
  \providecommand\BibTeX{{%
    \normalfont B\kern-0.5em{\scshape i\kern-0.25em b}\kern-0.8em\TeX}}}

\setcopyright{none}
\renewcommand\footnotetextcopyrightpermission[1]{}

\acmJournal{JACM}
\acmVolume{37}
\acmNumber{4}
\acmArticle{111}
\acmMonth{8}

\input{macros}

\begin{document}

\title{Comparative Synthesis: Learning Near-Optimal Network Designs by Query}


\author{Yanjun Wang}
\affiliation{
  \institution{Purdue University}
  \city{West Lafayette}
  \state{IN}
  \postcode{47907}
  \country{USA}
}
\email{wang3204@purdue.edu}
\author{Zixuan Li}
\affiliation{
  \institution{Purdue University}
  \city{West Lafayette}
  \state{IN}
  \postcode{47907}
  \country{USA}
}
\email{li3566@purdue.edu}
\author{Chuan Jiang}
\affiliation{
  \institution{Purdue University}
  \city{West Lafayette}
  \state{IN}
  \postcode{47907}
  \country{USA}
}
\email{jiang486@purdue.edu}
\author{Xiaokang Qiu}
\affiliation{
  \institution{Purdue University}
  \city{West Lafayette}
  \state{IN}
  \postcode{47907}
  \country{USA}
}
\email{xkqiu@purdue.edu}
\author{Sanjay G. Rao}
\affiliation{
  \institution{Purdue University}
  \city{West Lafayette}
  \state{IN}
  \postcode{47907}
  \country{USA}
}
\email{sanjay@ecn.purdue.edu}

%
\input{aAbstract}



\maketitle

\input{Aintro}

\input{Amotivation}
\input{A-Sanjay-formalDefinition}



\section{Voting-Guided Learning Algorithm}
\label{sec:voting}

In this section, we focus on the learner side of the framework and propose a voting-guided learning algorithm that can play the role of the comparative learner and solve the comparative synthesis problem.
Below we propose a novel search space combining the program search and objective learning, then present an estimation of query informativeness, based on which our voting-guided algorithm is designed. We discuss the convergence of the algorithm at last.

\subsection{A Unified Search Space}

A syntactical and natural means to describing quantitative specification is \emph{target functions} (in contrast to the semantically defined metric ranking in Def~\ref{def:ranking}).
Now to solve a comparative synthesis problem efficiently, an explicit task of the learner is \emph{program search}: the goal is to minimize human interaction (i.e., the number of queries) and maximize the quality of the solution (see Def~\ref{def:quality}) proposed through \Validate~queries. Another implicit task of the learner is \emph{objective learning}: to steer program search faster to the optimal and minimize the query count, the learner should
conjecture target functions that fit the teacher-provided preferences, and use them to determine which programs are more likely to be optimal. Note that the conjectured target function need not (and sometimes cannot) be perfect --- the goal is just to \emph{approximates} the teacher's metric ranking $\lesssim_{\mathcal{T}}$. 

Our key insight is that
the two tasks 
are inherently tangled and better be done together. On one hand, the quantitative synthesis task needs to be guided by an appropriate objective; otherwise the search is blind and unlikely to steer to those candidates satisfying the user. On the other hand, learning a perfect target function can be extremely expensive (if not impossible --- see the ``why metric ranking'' discussion in \S\ref{sec:foundation}) and unnecessary --- even an inaccurate target function may guide the program search. We first define the target function space.

\begin{definition}[Target function space]
\label{def:objective-space}
A target function space $\mathcal{O}$ is a set of target functions with respect to a $d$-dimensional metric group $M$ such that for any metric ranking $\lesssim_M \subset \mathbb{R}^{d} \times \mathbb{R}^{d}$ and any finite subset $S \subset_{\textrm{fin}} \mathbb{R}^{d}$, there exists a target function $O \in \mathcal{O}$ such that for any $u, v \in S$, $u \lesssim_M v$ if and only if $O(u) \leq O(v)$.
\end{definition}

\begin{example}
\label{ex:clia}
The class of \emph{conditional linear integer arithmetic} (CLIA) functions forms a target function space. 
A CLIA target function, intuitively, uses linear conditions over metrics to divide the domain into multiple regions, and defines in each region as a linear combination of metric values. Formally, for any $d$-dimensional metric group $M$, a target function space $\mathcal{O}_{CLIA}^d$ can be defined as the class of expressions derived from the nonterminal $T$ of the following grammar:
\begin{align*}
\small
T & ~::=~ E ~\mid~ \code{if}~ B~ \code{then}~ T~ \code{else}~ T~~~~~~~ \\
B & ~::=~ E \geq 0 ~\mid~ B \land B ~\mid~ B \lor B~~~~ \\
E & ~::=~ v_1 ~\mid~ \dots ~\mid~ v_d ~\mid~ c ~\mid~ E + E ~\mid~ E - E
\end{align*}
where $c \in \mathbb{Z}$ is a constant integer, and $v_i$ is the $i$-th value of the metric vector. It is not hard to see that $\mathcal{O}_{CLIA}^d$ is indeed a target function space, because with arbitrarily many conditionals, one function can be constructed to fit \emph{any finite subset} of any metric ranking.
\end{example}

\begin{example}
\label{ex:mcf-space}
 While the CLIA space is general enough for arbitrary metric group, it can be too large to be efficiently searched. For many concrete metric groups, more target function spaces usually exist.
 For our running example, a commonly used function to quantify this trade-off is the multi-commodity flow functions used in software-driven WAN~\cite{swan}. The $O_{real}$ function (Equation~\ref{eqn:real}) in our running example is an instance of the generalized, two-segment MCF function space, which can be described in the following form:
 \vspace{-.05in}
\begin{multline*} \label{eq:tmplt}
{\small
O(\throughput, \latency) \defeq \throughput~*~?? ~-~ \max(\throughput - ??,~ 0) ~*~ ??} \\
{\small ~-~ \latency ~*~ ?? ~-~ \max(\latency~-~??,~ 0) ~*~ ??}
\end{multline*}
where $??$ can be arbitrary weights or thresholds. Note that the two-segment template is insufficient to characterize arbitrary finite metric ranking. In that case, the template may be extended to more segments. We call the whole target function space $\mathcal{O}_{MCF}$.
\end{example}

Now as the learner's task is to search two spaces --- one for programs and one for target functions --- we merge the two tasks into a single one, searching over a unified search space which we call \emph{Pareto candidate set}: 

\begin{definition}[Pareto candidate set]
\label{def:pareto}
Let $\mathscr{C} = (\mathcal{P}, C, \Phi, M, \mathcal{T})$ be a comparative synthesis problem and $\mathcal{O}$ be a target function space w.r.t. $M$. A \emph{Pareto candidate set} (PCS) with respect to $\mathscr{C}$ and $\mathcal{O}$ is a finite partial mapping $\mathscr{G} : \mathcal{O} \nrightarrow C$ from a space of target functions $\mathcal{O}$ to a space of program parameters $C$, such that for any $O \in \mathcal{O}$, $\mathscr{G}(O)$ is the effectively optimal solution under target function $O$, i.e., a solution $c \in C$ such that $\mathscr{G}(O) \lesssim_O c$, if exists, cannot be effectively found. Specifically, for any other $O' \in \mathcal{O}$, $\mathscr{G}(O') \lesssim_O \mathscr{G}(O)$.

\end{definition}
Intuitively, a Pareto candidate set (PCS) $\mathscr{G}$ maintains a set of candidate target functions, a set of candidate programs, and a mapping between the two sets, and guarantees that every candidate target function $O$ is mapped to the best candidate program under $O$.\footnote{Note that in general, the best candidate program under $O$ is not necessarily unique. To break the tie and make $\mathscr{G}$ uniquely determined by the component sets of target functions and programs, when two candidate programs $c_1$ and $c_2$ both get the highest reward under $O$, we assume $\mathscr{G}(O)$ is $c_1$ if $M(c_1)$ is smaller than $M(c_2)$ in lexicographical order, or $c_2$ otherwise.}

\subsection{Query Informativeness}

Now with the unified search space --- PCS in Def~\ref{def:pareto} --- comparative synthesis becomes a game between the learner and the teacher: a PCS $\mathscr{G}$ is maintained as the current search space; and in each iteration, the learner makes a query and the teacher gives her response, based on which $\mathscr{G}$ is shrunken. The learner's goal should be, in each iteration, to pick the most informative query in the sense that it can reduce the size of $\mathscr{G}$ as fast as possible. 
The key question is how to evaluate the informativeness of a query.

In this paper, we develop a greedy strategy which evaluates the informativeness by computing \emph{how many candidate programs from $\mathscr{G}$ can be removed immediately with the teacher's response}. As the teacher's response can be arbitrary, our evaluation considers all possible responses and take the minimum number among all cases.
The formulation is shown in Fig~\ref{fig:estimate} and explained below.

\begin{figure}
\small
\begin{align*}
\Quality \big(\Compare(c_1, c_2)\big) & ~\defeq~ \min \Big(
\RemNEQ{c_1}{c_2},~  \RemNEQQ{c_2}{c_1},~
 \RemEQ{c_1}{c_2}
\Big) \\
\Quality \big(\Validate(c)\big) & ~\defeq~ \min \Big(
\RemNewBest{c},~
 \RemOldBest{c}
\Big)
\end{align*}
\begin{tabular}{rlrl}
$\big| \varphi \big|$ & $\!\!\!\defeq~ \Big| \big\{ x \mid \forall O \in \textrm{dom}(\mathscr{G}): \mathscr{G}(O) = x ~\Rightarrow~ \varphi \big\} \Big|$
&
$\RemNewBest{c}$ & $\!\!\!\defeq~ r\_best <_O c ~\lor~ \forall y \in \textrm{image}(\mathscr{G}). y \lesssim_O c$


\end{tabular}
\vspace{-.1in}
\caption{Informativeness of queries (with $\mathscr{G}$ the current PCS, and $r\_best$ the running best program).}
\label{fig:estimate}
\end{figure}

$\bullet~ \Compare$ {\bf query:} For $\Compare(c_1, c_2)$, recall the teacher may prefer $\mathcal{P}[c_1]$ to $\mathcal{P}[c_2]$, or vice versa, or consider the two programs equally good (corresponding to the three responses: $<$, $>$ and $=$). In each case, we can remove all the candidate target functions that have a different relative ordering
of $\mathcal{P}[c_1]$ and $\mathcal{P}[c_2]$ than the teacher's preference. We denote the number of candidates
that can be removed when the teacher prefers $\mathcal{P}[c_1]$ (resp. $\mathcal{P}[c_2]$) as $\RemNEQ{c_1}{c_2}$ (resp.  $\RemNEQQ{c_2}{c_1}$). Further let $\RemEQ{c_2}{c_1}$ denote the candidates that can be 
removed if the user indicates both programs are equally good. The overall informativeness is just the minimum of the 
three cases.


$\bullet~ \Validate$ {\bf query:} For $\Validate(c)$, recall the teacher may confirm that $\mathcal{P}[c]$ is indeed better than the running best $\mathcal{P}[r\_best]$ (response $\top$), or keep the current running best (response $\bot$). 
Like the compare query, in each case, we can eliminate all candidate target functions that do not satisfy this relative
preference between $\mathcal{P}[c]$ and $\mathcal{P}[r\_best]$ expressed by the user. However, in the case that the user
prefers $\mathcal{P}[c]$, we can additionally remove all candidate target functions for which $\mathcal{P}[c]$ is already the best choice, i.e., more queries are not needed for further improvement --- we use $\RemNewBest{c}$ in Fig~\ref{fig:estimate} to denote the total number of eliminated candidate programs in this case.
The overall informativeness is just the minimum of the two cases.


\begin{table*}
\begin{minipage}[t]{0.3\textwidth}
\centering
\caption{Example PCS $\mathscr{G}_{ex}$.
}
\vspace{-.1in}
\resizebox{\textwidth}{!}{
\begin{tabular}{|c|c|c|}
    \hline
    Trgt & Optimal & Ranking of all \\
    func & solution & candidate programs \\
    $O$ & $\mathscr{G}_{ex}(O)$ & under $O$ \\
    \toprule
    $O_0$ & $\boldsymbol{P_2}$ & $P_0 < r\_best < P_1 < \boldsymbol{P_2}$  \\
    \hline
    $O_1$ & $\boldsymbol{P_2}$ & $P_1 < P_0 < r\_best < \boldsymbol{P_2}$ \\
    \hline
    $O_2$ & $\boldsymbol{P_0}$ & $P_1 = P_2 < r\_best < \boldsymbol{P_0}$ \\
    \hline
    $O_3$ & $\boldsymbol{P_1}$ & $r\_best < P_0 < P_2 < \boldsymbol{P_1}$ \\
    \hline
    $O_4$ & $\boldsymbol{P_2}$ & $r\_best < P_0 < P_1 < \boldsymbol{P_2}$ \\
    \hline
    \end{tabular}
    }
    \label{tab:ranking}
\end{minipage}
\hfill%
\begin{minipage}[t]{0.38\textwidth}
\caption{Informativeness of queries \Compare($c_1, c_2$).}
\vspace{-.1in}
\resizebox{\textwidth}{!}{
\begin{tabular}{|c|c|r|r|r|r|}
    \hline
    $c_1$ & $c_2$ & $\RemNEQ{c_1}{c_2}$ & $\RemNEQQ{c_2}{c_1}$ & $\RemEQ{c_1}{c_2}$ & \Quality \\
    \toprule
    $P_0$ & $P_1$ & 1 & 1 & 3 & 1 \\
    \hline
    $P_1$ & $P_2$ & 2 & 2 & 2 & 2 \\
    \hline
    $P_0$ & $P_2$ & 2 & 1 & 3 & 1 \\
    \hline
    $P_0$ & $r\_best$ & 0 & 2 & 3 & 0 \\
    \hline
    $P_1$ & $r\_best$ & 1 & 1 & 3 & 1 \\
    \hline
    $P_2$ & $r\_best$ & 1 & 2 & 3 & 1 \\
    \hline
    \end{tabular}
    }
    \label{tab:compare-query-example}
\end{minipage}
\hfill%
\begin{minipage}[t]{0.3\textwidth}
\centering
\caption{Informativeness of queries \Validate($c$).}
\vspace{-.1in}
\resizebox{\textwidth}{!}{
\begin{tabular}{|c|r|r|r|}
    \hline
    $c$ & $\RemNewBest{c}$ & $\RemOldBest{c}$ & \Quality \\
    \toprule
    $P_0$ & 1 & 2 & 1 \\
    \hline
    $P_1$ & 2 & 1 & 1 \\
    \hline
    $P_2$ & 2 & 2 & 2 \\
    \hline
    \end{tabular}
    }
    \label{tab:propose-query-example}
\end{minipage}
\end{table*}

\begin{example}
Table~\ref{tab:ranking} shows a PCS $\mathscr{G}_{ex}$ that consists of 5 candidate target functions, namely $O_i$ for $0 \leq i \leq 4$, and 3 candidate programs, namely $P_0, P_1$ and $P_2$. The rankings of all candidate programs and the current running best $r\_best$ under target functions are also shown in Table~\ref{tab:ranking}. The informativeness of \Compare~and \Validate~queries for PCS $\mathscr{G}_{ex}$ is presented in Tables~\ref{tab:compare-query-example} and~\ref{tab:propose-query-example}, respectively. Take \Compare($P_1$, $P_2$) as an example, $\RemNEQ{P_1}{P_2}$ is 2 since all candidate target function except $O_3$ believes $P_1 \lesssim_O P_2$ and removing these four target functions essentially removes $P_0$ and $P_2$ from PCS $\mathscr{G}_{ex}$. As per Fig~\ref{fig:estimate}, $\RemNEQQ{P_2}{P_1}$ and $\RemEQ{P_1}{P_2}$ are also 2.
The informativeness of \Compare($P_1$, $P_2$) is therefore 2, the minimum of the three cases above, which means at least 2 candidate programs will be removed from $\mathscr{G}_{ex}$ no matter which the user prefers. 
Consider \Validate($P_2$) as another example, $\RemNewBest{P_2}$ is 2, since $O_2$ prefers $r\_best$ over $P_2$ and $P_2$ is the best choice under $O_0$, $O_1$ and $O_4$. Candidate program $P_1$ and $P_2$ will be removed from $\mathscr{G}_{ex}$ as their target functions either do not satisfy user preference or can not improve the running best further. 
Given PCS $\mathscr{G}_{ex}$, both \Validate($P_2$) and \Compare($P_1$, $P_2$) share highest informativeness, which is 2. In this case, \Validate($P_2$) will be presented to the user.
\end{example}


\subsection{The Algorithm}

Our learning algorithm is almost straightforward: for each iteration, compute the informativeness of every possible query and make the most informative query. The remaining issue is that it is not realistic to keep a PCS that contains all possible candidates, because the number of candidates is usually very large, if not infinite. For example, $\mathscr{L}_{MCF}$ in Example~\ref{ex:mcf-qual} has infinitely many Pareto optimal solutions, ranging in the continuous spectrum from maximizing throughput to minimizing latency. To this end, our voting-guided algorithm maintains a moderate-sized PCS, from which queries are generated and selected based on their informativeness.

Algorithm~\ref{alg:voting} illustrates the voting-guided algorithm. The algorithm takes as input a comparative synthesis problem $\mathscr{C}$ and a target function space $\mathcal{O}$, and maintains a PCS $\mathscr{G}$ w.r.t. $\mathscr{C}$ and $\mathcal{O}$ and set of preferences $R$, both empty initially. 
In each iteration, the algorithm computes the informativeness of all possible queries that can be made about the current candidates $\textrm{image}(\mathscr{G})$, and picks the highest-informativeness query according to the computation presented in Fig~\ref{fig:estimate} (line~\ref{line:quality}). 
After the query is made and the response is received, an \compareupdate~subroutine is invoked to update $\mathscr{G}$ and remove all candidates violating the preference (lines~\ref{line:prune-begin}--\ref{line:prune-end}).  
Moreover, the algorithm also checks at the beginning of every iteration the size of $\mathscr{G}$; if $\textrm{image}(\mathscr{G})$ is below a fixed threshold {\sc Thresh}, the algorithm attempts to extend $\mathscr{G}$ using a \generatemore~subroutine. The algorithm terminates and returns the current running best when $\mathscr{G}$ becomes 0 or {\sc NQuery} queries have been made, where {\sc NQuery} is the number of queries that the teacher promises to answer (line~\ref{line:loopend}). Table~\ref{tab:example-run} shows an example run of this algorithm.

\begin{algorithm}[t]
\SetKwComment{Comment}{/* }{ */}
\SetKwInOut{Input}{input}
\SetKwInOut{Output}{output}
\SetKwProg{Fn}{def}{\string:}{}
\SetKwFunction{Conc}{\color{blue} \bf concolic-synth}
\SetKwFunction{Grad}{\color{blue} \bf cooperative-synth}
\SetKwFunction{Fixed}{\color{blue} \bf fixed-height}
\SetKwFunction{Bounded}{\color{blue} \bf bounded-height-enumerate}
\SetKwFunction{BuildTree}{buildGraph}
\SetKwFunction{CheckSat}{\color{blue} \bf verify}
\SetKwFunction{Synth}{\color{blue} \bf ind-synth}
\SetKwFunction{Tree}{\color{blue} \bf checkTreeness}
\SetKwFunction{ValidSk}{ \sc ValidSkeleton}
\SetKwFunction{CandSk}{ \sc CandSk}
\SetKwFunction{ObviousSol}{ \sc ObviousSolution}
\SetKwFunction{Simplify}{\sc SubproblemsB}
\SetKwFunction{Deduct}{\color{blue} \bf deduct}
\SetKwFunction{Enumerate}{\color{blue} \bf enum-n-val}
\SetKwFunction{Objectivefirst}{\color{blue} \bf objective-first}
\SetKwFunction{Objectiveaided}{\color{blue} \bf objective-aided}
\SetKwFunction{Objectivevoting}{\color{blue} \bf voting-guided}
\SetKwFunction{Learn}{\color{blue} \bf voting-guided-learn}
\SetKwFunction{GenProg}{{\sc SynProg}}
\SetKwFunction{Generate}{\color{blue} \bf generate-more}
\SetKwFunction{Generateobj}{\color{blue} \bf gen-objective}
\SetKwFunction{Generatebetter}{\color{blue} \bf gen-better-program}
\SetKwFunction{ValidateCompare}{\color{blue} \bf validate-then-compare}
\SetKwFunction{QueryType}{{\sc BestQuery}}
\SetKwFunction{MakeQuery}{{\sc MakeQuery}}
\SetKwFunction{BestValQuery}{{\sc BestValidateQuery}}
\SetKwFunction{RemoveCommonWorst}{\sc RemoveCommonWorst}
\SetKwFunction{Needupdateobj}{\color{blue} \bf need-update-objective}
\SetKwFunction{Needobj}{\color{blue} \bf need-more-objectives}
\SetKwFunction{Needprog}{\color{blue} \bf need-more-programs}
\SetKwFunction{Needcand}{\color{blue} \bf need-more-candidates}
\SetKwFunction{CompareUpdate}{\color{blue} \bf update}
\SetKwFunction{Cegis}{\color{blue} \bf cegis}
\SetKwFunction{Simple}{ \sc Simple}
\SetKwFunction{Solved}{\sc Solved}
\SetKwFunction{Unsolved}{\sc Unchecked}
\SetKwFunction{AddChild}{\sc AddChild}
\SetKwFunction{CreateNode}{\sc CreateNode}
\SetKwFunction{Children}{{\sc Children}}
\SetKwFunction{HighestLeaf}{{\sc HighestUncheckedLeaf}}
\SetKwFunction{WeakerSpec}{{\sc WeakerSpecs}}
\SetKwFunction{FixedTerm}{{\sc FixedTerms}}
\SetKwFunction{Sub}{\sc SubproblemsA}
\SetKwFunction{Dequeue}{dequeue}
\SetKwFunction{Enqueue}{enqueue}
\SetKwFunction{EmptyQueue}{emptyQueue()}
\SetKwFunction{EmptyPriorityQueue}{emptyPriorityQueue()}
\SetKwFunction{root}{source}
\SetKwFunction{Empty}{empty}
\SetKwFunction{spec}{spec}
\SetKwFunction{target}{target}
\SetKwFunction{grammar}{grammar}
\SetKwFunction{parent}{pred}
\SetKwFunction{children}{succ}
\SetKwFunction{solution}{solution}
\SetKwFunction{IsSolution}{\sc IsSolution}
\SetKwFunction{Simplify}{{\sc Simplify}}
\SetKwFor{For}{for}{:}{}%
\SetKwFor{ForEach}{foreach}{:}{}
\SetKwIF{If}{ElseIf}{Else}{if}{:}{elif}{else:}{}%
\SetKwFor{While}{while}{:}{fintq}%
\SetKw{Assert}{assert}
\SetKw{Break}{break}
\SetKw{Continue}{continue}
\SetKw{True}{true}
\SetKw{False}{false}
\AlgoDontDisplayBlockMarkers
\SetAlgoNoEnd
\LinesNumbered
\Input{A comparative synthesis problem $\mathscr{C} = (\mathcal{P}, C, \Phi, M, \mathcal{T})$ and a target function space $\mathcal{O}$ with respect to $M$}
\Output{A quasi-optimal solution to $\mathscr{C}$}
\Fn{\Learn{$\mathcal{P}, C, \Phi, M, \mathcal{T}, \mathcal{O}$}}{
    $R \leftarrow \top$, $count \leftarrow 0$, $\mathscr{G} \leftarrow \emptyset$ \tcp{collected preferences, query count and the PCS}
    $r\_best \leftarrow$ \GenProg{$\mathcal{P}, C, \phi_{P}$} \tcp{get the first solution and initialize the running best}
    \Repeat{$\mathscr{G} = \emptyset \lor count \geq$ {\sc NQuery} \label{line:loopend}}{
        \If{$|{\textrm{image}}(\mathscr{G})| <$ {\sc Thresh}}{
            \Generate($\mathcal{P}$, $M$, $R$, $r\_best$, $\mathscr{G}$) \tcp{generate more candidates}
        }
        \Comment{pick and make the most informative query}
        $q\_type, c_1, c_2 \leftarrow$ \QueryType($\mathscr{G}$)
        \label{line:quality} \\
        $response \leftarrow$ \MakeQuery($q\_type$, $c_1$, $c_2$) \\
        \If{$q\_type = \Validate$ \tcp*[h]{if \Validate{($c_1$)}}}{
            \CompareUpdate($M(\mathcal{P}[c_1])$, $M(\mathcal{P}[r\_best])$, $response$) \tcp{update $R$ and $\mathscr{G}$}
            \If{$response = (\top)$  \label{line:prune-begin}}{
                \tcp{if $\mathcal{P}[c_1]$ is better than running best}
                $\mathscr{G} \leftarrow \mathscr{G}\mid_{\{O \mid \mathscr{G}(O) >_O \mathcal{P}[c_1] \}}$ \\
                $r\_best \leftarrow c_1$
            }
        }
        \ElseIf{$q\_type = \Compare$ \tcp*[h]{if \Compare{($c_1, c_2$)}}}{
            \CompareUpdate($M(\mathcal{P}[c_1])$, $M(\mathcal{P}[c_2])$, $response$) \tcp{update $R$ and $\mathscr{G}$} \label{line:prune-end}
        }
        $count++$ \\
    }
    \Return{$r\_best$} \label{line:voting-fail}
}
\caption{The voting-guided learning algorithm.}
\label{alg:voting}
\end{algorithm}

The subroutines involved in the algorithm are shown as Algorithm~\ref{alg:subroutines}. The \compareupdate~subroutine is straightforward, taking a new preference pair and shrinking $\mathscr{G}$ accordingly.
The \generatemore~subroutine is tasked to expand $\mathscr{G}$ as much as possible within a time limit. Each time, it tries to find a pair $(O, c)$ such that $O$ satisfies all existing preferences and prefers $\mathcal{P}[c]$ to $\mathcal{P}[r\_best]$, and $\mathcal{P}[c]$ is effectively optimal under $O$. Note that this subroutine delegates several heavy-lifting tasks to off-the-shelf, domain-specific procedures: {\sc SynProg} for qualitative synthesis, {\sc SynObj} for objective synthesis, and {\sc Improve} for optimization under a known objective.
For example, the qualitative synthesis and objective synthesis problem of our running example can be encoded to a logical query and discharged by any SMT solvers, such as Z3~\cite{Z3}. The optimization problem under known objectives can also be solved by a linear programming solver, such as Gurobi~\cite{gurobi}.

\begin{algorithm}[t]
\SetKwInOut{Input}{input}
\SetKwInOut{Output}{output}
\SetKwInOut{Modifies}{modifies}
\SetKwProg{Fn}{def}{\string:}{}
\SetKwFunction{CheckSat}{\color{blue} \bf verify}
\SetKwFunction{Synth}{\color{blue} \bf ind-synth}
\SetKwFunction{Simplify}{\sc SubproblemsB}
\SetKwFunction{Deduct}{\color{blue} \bf deduct}
\SetKwFunction{Enumerate}{\color{blue} \bf enum-n-val}
\SetKwFunction{Objectivefirst}{\color{blue} \bf objective-first}
\SetKwFunction{Objectiveaided}{\color{blue} \bf objective-aided}
\SetKwFunction{ValidateCompare}{\color{blue} \bf validate-then-compare}
\SetKwFunction{CompareUpdate}{\color{blue} \bf update}
\SetKwFunction{Generate}{{\color{blue} \bf generate-more}}
\SetKwFunction{GenObj}{{\sc SynObj}}
\SetKwFunction{GenBetterProg}{{\sc Improve}}
\SetKwFunction{GenProg}{{\sc SynProg}}
\SetKwFunction{Generatek}{\color{blue} \bf gen-k}
\SetKwFunction{Cegis}{\color{blue} \bf cegis}
\SetKwFunction{Simple}{ \sc Simple}
\SetKwFunction{Solved}{\sc Solved}
\SetKwFunction{Unsolved}{\sc Unchecked}
\SetKwIF{If}{ElseIf}{Else}{if}{:}{elif}{else:}{}%
\SetKw{Assert}{assert}
\SetKw{Break}{break}
\SetKw{Continue}{continue}
\SetKw{True}{true}
\SetKw{False}{false}
\AlgoDontDisplayBlockMarkers
\SetAlgoNoEnd
\SetAlgoNoLine%
\Input{Two program metric vectors $m, n$ and their comparison result $response$}
\Modifies{The current metric vector preferences $R$ and the current PCS $\mathscr{G}$}
\Fn{\CompareUpdate($m$, $n$, $response$)}{
     \If{$response = (>)$}{
        $R \leftarrow R \land m > n$
     }
     \ElseIf{$response = (<)$}{
        $R \leftarrow R \land n > m$
     }
     \Else{
        $R \leftarrow R \land m = n$
     }
     $\mathscr{G} \leftarrow \mathscr{G}\mid_{\{O \mid O \models R\}}$ \\
     \Return \\
}
\Input{A parameterized program $\mathcal{P}$, a metric group $M$, current metric vector preferences $R$ and current running best $r\_best$}
\Modifies{The Pareto candidate set $\mathscr{G}$}
\Fn{\Generate($\mathcal{P}$, $M$, $R$, $r\_best$, $\mathscr{G}$)}{
    \Repeat{timeout}{
        $c \leftarrow $ \GenProg{$\mathcal{P}$, $M$} \tcp*[l]{synthesize an arbitrary (Pareto optimal) program}
        $O \leftarrow$ \GenObj{$R \land M(\mathcal{P}[c]) > M(\mathcal{P}[r\_best])$} \tcp{synthesize an objective that prefers the new $c$ over $r\_best$}
        \If{$O \neq \bot$}{
            $c \leftarrow$ \GenBetterProg{$O$, $\mathcal{P}$, $M$, $c$} \tcp*[h]{this is optional: try to improve $\mathcal{P}[c]$}
        }
        \Else{
            $O \leftarrow$ \GenObj{$R$} \tcp{synthesize an arbitrary objective satisfying $R$}
            $c \leftarrow$ \GenBetterProg{$O$, $\mathcal{P}$, $M$, $r\_best$}
            \tcp*[h]{synthesize a best possible program under $O$, but at least better than $r\_best$}
        }
        $\mathscr{G} \leftarrow \mathscr{G} \uplus (O, c)$
    }
}
\caption{The subroutines involved in the voting-guided learning algorithm.}
\label{alg:subroutines}
\end{algorithm}

\subsection{Convergence}

In the rest of the section, we discuss the convergence of the algorithm. Recall that our algorithm only produces quasi-optimal programs as the ground-truth target function is not present. Therefore, the algorithm should be evaluated on the rate of convergence~\cite{convergence}, i.e., how fast the median quality~\footnote{The algorithm involves random sampling and results we prove below are for the median quality of output solutions; the proofs can be easily adapted to get similar results for the mean quality of solutions.} of solutions (see Def~\ref{def:quality}) approaches 1 as more queries are made. Our first result is that the algorithm guarantees a logarithmic rate of convergence.

\begin{theorem}
\label{thm:bound}
Given a comparative synthesis problem $\mathscr{C}$ and a target function space $\mathcal{O}$ as input, if Algorithm~\ref{alg:voting} terminates after $n$ queries, the median quality of the output solutions is at least $\displaystyle 2^{\frac{-1}{n+1}}$.
\end{theorem}
\begin{proof}
Note that every query will discard at least one candidate program from the PCS, regardless of the query type. In other words, the final output $c$ must be the optimal among at least $(n+1)$ randomly selected candidates from the uniform distribution. Therefore, the quality of $c$ is at least the $(n+1)$-th order statistic of the uniform distribution, which is a beta distribution $\textrm{Beta}(n+1, 1)$, whose median is $\displaystyle 2^{\frac{-1}{n+1}}$.
\end{proof}


The proved lower-bound in the theorem above is tight only when each query only removes one candidate from the PCS $\mathscr{G}$. Unfortunately, the following lemma shows that in general, this scenario is always realizable:
\begin{theorem}
\label{thm:tight}
The bound in Theorem~\ref{thm:bound} is tight.
\end{theorem}

The PCS constructed for the following lemma serves as a witness of the bound tightness:
\begin{lemma}
\label{thm:worst}
Let $\mathscr{S} = (\mathcal{P}, C, \Phi)$ be a qualitative synthesis problem with infinitely many solutions and $\mathcal{O}$ be a target function space. For any integer $n > 0$, there exist a PCS $\mathscr{G}: \mathcal{O} \nrightarrow C$ and a parameter $r\_best \in C$ such that: 
{\bf (1)} $|\textrm{image}(\mathscr{G})| = n$; 
{\bf (2)} for any $c_1, c_2 \in \textrm{image}(\mathscr{G})$, $\Quality(\Compare(c_1, c_2)) = 1$; 
{\bf (3)} for any $c \in \textrm{image}(\mathscr{G})$, $\Quality(\Validate(c)) = 1$.
\end{lemma}
\begin{proof}
As $\mathscr{C}$ has infinitely many solution, we can pick arbitrary $n$ solutions, say $c_1, \dots, c_n$. For each $1 \leq i \leq n$, one can construct a total order $\lesssim_i$ such that $c_n \lesssim_i \dots c_{i+1} \lesssim_i c_{i-1} \dots c_{1} \lesssim_i c_i$. According to the definition of target function space (Def~\ref{def:objective-space}), there exists a target function $O_i$ that fits $\lesssim_i$. Now we can construct $\mathscr{G}$ such that $\textrm{dom}(\mathscr{G}) = \{O_1, \dots, O_n\}$, and $\mathscr{G}(O_i) = c_i$ for each $i$. It can be verified $\mathscr{G}$ is a Pareto candidate set satisfying the required conditions.
\end{proof}

\subsection{Better Convergence Rate with Sortability}
We have shown that our voting-guided algorithm guarantees a logarithmic rate of convergence in general, but are there scenarios for which the algorithm guarantees faster convergence?
We next show that when the comparative synthesis problem is convex and the target function space is concave with two metrics, our algorithm guarantees a faster, \emph{linear} convergence. 
The conditions are commonly seen in practice --- satisfied by half of optimization scenarios studied in \S\ref{sec:evaluation} --- and intuitively, capture the assumption that there are two competing metrics (e.g., throughput and latency) such that for each metric continued improvement leads to diminishing marginal utility (e.g., increasing throughput from 1Gbps to 2Gbps is more favorable than increasing throughput from 2Gbps to 3Gbps).

The idea of the proof bears a similarity to the convergence guarantee for many algorithms in traditional convex optimization~\cite{convex-optimization}; but the key difference is that the objective is indeterminate for our algorithm. We first introduce a key enabling notion for the proof called \emph{sortability},
which makes sure that the candidates in the PCS can be ordered appropriately such that every target function with corresponding candidate $c$ always prefers its nearer neighbors to farther neighbors.

\begin{definition}[Sortability]
\label{def:sortability}
A PCS $\mathscr{G}$ is sortable if there exists a total order $\prec$ over $\textrm{image}(\mathscr{G})$ such that for any target functions $O, P, Q \in \textrm{dom}(\mathscr{G})$ such that $\mathscr{G}(O) \prec \mathscr{G}(P) \prec \mathscr{G}(Q)$, the following two conditions hold:
$\mathscr{G}(P) >_{O} \mathscr{G}(Q)$, and $\mathscr{G}(P) >_{Q} \mathscr{G}(O)$
A target function space $\mathcal{O}$ is sortable with respect to a comparative synthesis problem $\mathscr{C}$ if any PCS $\mathscr{G}$ w.r.t. $\mathscr{C}$ and $\mathcal{O}$ is sortable.
\end{definition}


The following lemma shows that if a PCS is sortable, one can make a query to cut at least half of the candidates, no matter what the teacher's response is.

\begin{lemma}
\label{lemma:half}
If a Pareto candidate set $\mathscr{G}$ is finite and sortable, then there exists a query whose quality for $\mathscr{G}$ as computed in Fig~\ref{fig:estimate} is $\displaystyle \lfloor \frac{|\textrm{image}(\mathscr{G})|}{2} \rfloor$.
\end{lemma}
\begin{proof}
Let $n = |\textrm{image}(\mathscr{G})|$ and $m = \displaystyle \lfloor \frac{|\textrm{image}(\mathscr{G})|}{2} \rfloor$. As $\mathscr{G}$ is sortable, by Def~\ref{def:sortability}, there exists a total order $\mathscr{G}(O_1) \prec \dots \prec \mathscr{G}(O_{n})$. Now we claim that $\Quality \Big(\Compare \big(\mathscr{G}(O_m), \mathscr{G}(O_{m+1})\big) \Big) = m$.
By Def~\ref{def:sortability}, for any $1 \leq i \leq m$, $\mathscr{G}(O_m) >_{O_i} \mathscr{G}(O_{m+1})$, and for any $m+1 \leq j \leq n$, $\mathscr{G}(O_m) <_{O_j} \mathscr{G}(O_{m+1})$. Then according to the query quality estimation described in Fig~\ref{fig:estimate}, both $\#\textrm{RemNEQ}\big(\mathscr{G}(O_m), \mathscr{G}(O_{m+1})\big)$ and $\#\textrm{RemNEQ}\big(\mathscr{G}(O_{m+1}), \mathscr{G}(O_{m})\big)$ are at least $m$. Therefore, $\Quality \big(\Compare(\mathscr{G}(O_m), \mathscr{G}(O_{m+1}))\big) = m$, which is $\displaystyle \lfloor \frac{|\textrm{image}(\mathscr{G})|}{2} \rfloor$.
\end{proof}

With the lower bound of removed candidates guaranteed by Lemma~\ref{lemma:half}, 
our voting-guided synthesis algorithm guarantees to produce a unique best candidate after a logarithmic amount of queries:
\begin{theorem} \label{thm:bound2}
Given a comparative synthesis problem $\mathscr{C}$ with metric group $M$ and a sortable target function space $\mathcal{O}$ w.r.t. $M$ as input, if Algorithm~\ref{alg:voting} terminates after $n$ queries, the median quality of the output solutions is at least $\displaystyle \big(1 - \frac{1}{\Omega(1.5^{n})}\big)$.
\end{theorem}
\begin{proof}
Note that Algorithm~\ref{alg:voting} generates candidates for the Pareto candidate set $\mathscr{G}$ (through the \generatemore~ subroutine) through random sampling. Therefore, if a query cuts the size of current candidate pool ($\mathscr{G}$ and the running best) by a ratio of $r$, the search space (those candidates satisfying all preferences in $R$) is cut by an equal or higher ratio in that iteration (extra candidates may be discarded by \generatemore, before the query). Now as $\mathscr{G}$ is sortable, by Lemma~\ref{lemma:half}, after the highest-informativeness query, the number of candidates remaining in $\mathscr{G}$ is at most $\lceil \frac{|\textrm{image}(\mathscr{G})|}{2} \rceil$. In other words, the query reduces the size of $\mathscr{G}$ by a ratio of at least $\displaystyle \frac{2}{3}$ (when $|\textrm{image}(\mathscr{G})| = 2$, the total number of candidates including the running best, reduces from 3 to 2), except for the last query. Therefore, the output is the best among $\mathcal{O}(1.5^n)$ randomly selected candidates, which is $\textrm{Beta}(1.5^n, 1)$-distributed. Hence by Def~\ref{def:quality}, the median of the quality of the output is $\displaystyle \frac{1}{2^{\big(\frac{1}{1.5^n}\big)}}$, which is asymptotically equivalent to $\displaystyle \big(1 - \frac{1}{\Omega(1.5^{n})}\big)$.
\end{proof}


\vspace{.1in}

We now formally define the convexity of the comparative synthesis problem and the concavity of the target function space, and build the main convergence result by proving the sortability.

\begin{definition}[Convexity of comparative synthesis problem]
\label{def:convexity}
A comparative synthesis problem $\mathscr{C}$ with metric group $M$ is convex if for any two solutions $c_1, c_2$ to $\mathscr{C}^{\textrm{qual}}$ and any $\alpha \in [0,1]$, a solution $c_3$ to $\mathscr{C}^{\textrm{qual}}$ can be effectively found such that $M(c_3) \gtrsim \alpha \cdot M(c_1) + (1-\alpha) \cdot M(c_2)$.
\end{definition}

\begin{definition}[Concavity of target function space]
\label{def:concavity}
Let $\mathcal{O}$ be a target function space w.r.t. a $d$-dimensional metric group $M$. $\mathcal{O}$ is concave if for any $O \in \mathcal{O}$, for any $v_1, v_2 \in \mathbb{R}^{d}$ and any $\alpha \in [0,1]$, $O(\alpha \cdot v_1 + (1-\alpha) \cdot v_2) \geq \max(O(v_1), O(v_2))$.
\end{definition}

\begin{example}
Our running example falls in this subclass. The comparative synthesis problem $\mathscr{C}_{MCF}$ in Example~\ref{ex:mcf-compare} is convex. As shown in Fig~\ref{fig:mcf}, both \throughput~and \latency~are weighted sum of allocations to every link. Therefore given any two solutions $c_1$ and $c_2$, their convex combination is still feasible, and the metric vector is also the corresponding convex combination of $M(c_1)$ and $M(c_2)$. Moreover, it is not hard to verify that the target function space $\mathcal{O}_{MCF}$ in Example~\ref{ex:mcf-space} is concave, as both the weights of \throughput~and \latency~decrease when their values are good enough and exceed a threshold.
\end{example}


\begin{theorem} \label{thm:sortable}
 Let $\mathscr{C}$ be a convex comparative synthesis problem with a 2-dimensional metric group $M$ and $\mathcal{O}$ be a concave target function space w.r.t. $M$, then $\mathcal{O}$ is sortable w.r.t. $\mathscr{C}^{\textrm{qual}}$.
\end{theorem}
\begin{proof}
We shall show the sortability of any Pareto candidate set $\mathscr{G}$ w.r.t. $\mathscr{C}^{\textrm{qual}}$ and $\mathcal{O}$. We claim that the lexicographic order $\prec_{lex}$ over $\mathbb{R}^2$ (i.e., $(a_1, a_2) \prec_{lex} (b_1, b_2)$ if and only if $a_1 < b_1$ or $a_1 = b_1 \land a_2 < b_2$) witnesses the sortability. 
Per Def~\ref{def:sortability}, for any target functions $O, P, Q \in \textrm{dom}(\mathscr{G})$ such that $\mathscr{G}(O) \prec_{lex} \mathscr{G}(P) \prec_{lex} \mathscr{G}(Q)$, we shall show $\mathscr{G}(P) >_O \mathscr{G}(Q)$ below. It can be similarly proved that $\mathscr{G}(P) >_Q \mathscr{G}(O)$.

Let $M(\mathscr{G}(O)) = (o_1, o_2)$, $M(\mathscr{G}(P)) = (p_1, p_2)$, and $M(\mathscr{G}(Q)) = (q_1, q_2)$. Note that by Def~\ref{def:pareto}, each of $\mathscr{G}(O)$, $\mathscr{G}(P)$ and $\mathscr{G}(Q)$ is optimal under a distinct target function, therefore $M(\mathscr{G}(O))$, $M(\mathscr{G}(P))$, and $M(\mathscr{G}(Q))$ are pairwise incomparable, i.e., $\{o_1, p_1, q_1\}$ and $\{o_2, p_2, q_2\}$ are all distinct values. Due to the lexicographic order $\prec_{lex}$, we have $o_1 < p_1 < q_1$ and $o_2 > p_2 > q_2$. Now by Def~\ref{def:convexity}, one can effectively find a solution $c$ such that 
\[
\small
\begin{split}
M(c) \geq \big(\frac{(q_1 - p_1) \cdot o_1 + (p_1 - o_1) \cdot q_1}{q_1 - o_1}, \frac{(q_1 - p_1) \cdot o_2 + (p_1 - o_1) \cdot q_2}{q_1 - o_1}\big) 
= \big(p_1, \frac{(q_1 - p_1) \cdot o_2 + (p_1 - o_1) \cdot q_2}{q_1 - o_1}\big)
\end{split}
\]
Then by Def~\ref{def:pareto}, $\mathscr{G}(P)$ is at least as good as $c$ and $p_2 \geq \displaystyle\frac{(q_1 - p_1) \cdot o_2 + (p_1 - o_1) \cdot q_2}{q_1 - o_1}$. Finally, by Def~\ref{def:concavity}, we have
$
\displaystyle O((p_1, p_2)) \geq O((p_1, \frac{(q_1 - p_1) \cdot o_2 + (p_1 - o_1) \cdot q_2}{q_1 - o_1})) \geq O((q_1, q_2)) 
$.
In other words, $\mathscr{G}(P) >_O \mathscr{G}(Q)$.
\end{proof}



\section{Evaluation}
\label{sec:evaluation}
\input{Aevaluation}

\input{Auser_study}

\section{Related Work}

\textbf{Network verification/synthesis.}
As we discussed in~\S\ref{sec:motivation}, the na\"{i}ve approach to comparative synthesis proposed in \cite{comparative-synthesis} is preliminary and may involve prohibitively many queries. In contrast, we generalize and formalize the framework, design the first synthesis algorithm with the explicit goal of minimizing queries, present formal convergence results, and conduct extensive evaluations including a user study.

Much recent work applies program languages techniques to networking. Several works focus on synthesizing forwarding tables or router configurations based on predefined rules~\cite{genesis,netgen,merlin,cocoon,netegg,synet,netcomplete}, or synthesizing provably-correct network updates~\cite{McClurgPLDI15, McClurgPLDI16}. Much research focuses on verifying network configurations and dataplanes~\cite{Beckett2020, qarc, netdice}, and does not consider synthesis.
Recent works mine network specification from
configurations~\cite{config2spec},
generate code for programmable switches from program sketches~\cite{domino,chipmunk},
or focus on generating network classification programs from raw network traces~\cite{shi:tacas}. In contrast to these works, we focus on synthesizing network designs to meet quantitative objectives, with the objectives themselves not fully specified.


\textbf{Optimal synthesis.}
There is a rich literature on synthesizing optimal programs with respect to a fixed or user-provided quantitative objective.
Some of these techniques aims to solve optimal syntax-guided synthesis problems by minimizing given cost functions~\cite{synapse,qsygus}. Other approaches either generate optimal parameter holes in a program through probabilistic reasoning~\cite{Chaudhuri14} or solve SMT-based optimization problems~\cite{symba}, under specific target functions.
In example-based synthesis, the examples as a specification can be insufficient or incompatible. Hence quantitative objectives can be used to determine to which extent a program satisfies the specification or whether some extra properties hold.
\citet{gulwani2019quantitative} and \citet{wrex} defined the problem of quantitative PBE (qPBE) for synthesizing string-manipulating programs that satisfy input-output examples as well as minimizing a given quantitative cost function.
%
Our work is different from all optimal synthesis work mentioned above as in our setting, the objective is unknown and automatically learnt/approximated from queries. 

\textbf{Human interaction.}
Many novel human interaction techniques have been developed for synthesizing string-manipu\-lating programs.
A line of work focuses on proposing user interaction models to help resolve ambiguity in the examples~\cite{mayer15} and/or accelerate program synthesis~\cite{Drachsler-Cohen17,peleg18}.
%
Using interactive approaches to solve multiobjective optimization problems has been studied by the optimization community for decades (as surveyed by \citet{multiobjective}).
Morpheus~\cite{metricComparison} is a routing control platform that allows users to flexibly specify their policy preferences. Morpheus requires pairwise comparisons on relative weights of metrics as input, which can be viewed as a special form of target functions.
Our novelty on the interaction method is to proactively ask comparative queries on concrete network designs, with the aim of minimizing the number of queries and maximizing the desirability of the found solution. The comparison of concrete candidates is easier than asking the user to provide rank scale, marginal rates of substitution or aspiration level, which is done by most existing approaches. The objectives we target to learn for network design also involve guard conditions, which is beyond what most existing methods can handle.

\textbf{Oracle-guided synthesis.}
The learner-teacher interaction paradigm we use in the paper
has been studied in the context of programming-by-example (PBE), aiming at minimizing the sequence of queries. \citet{jha10}~presented an oracle-based learning approach to automatic synthesis of bit-manipulating programs and malware deobfuscation over a given set of components. Their synthesizer generates inputs that distinguishes between candidate programs and then queries to the oracle for corresponding outputs. \citet{Ji2020} followed up and studied how to minimize the sequence of queries. This line of work allows input-output queries only (``what is the output for this input?'') to distinguish different programs. If two programs are distinguishable, they consider them equivalent or a ranking function is given.
In invariant synthesis, \citet{ice-cs} followed this paradigm and synthesized inductive invariants by checking hypotheses (equivalent to \Validate~queries in our setting) with the teacher. \citet{ogis} proposed a theoretical framework, called oracle-guided inductive synthesis (OGIS) for inductive synthesis. The framework OGIS captures a class of synthesis techniques that operate via a set of queries to an oracle. Our comparative synthesis can be viewed as a new instantiation of the OGIS framework.

\textbf{Active learning.} 
Our algorithm for comparative synthesis has parallels to active learning~\cite{active-learning,angluin2004} in the machine learning community, which interactively queries a user to label data in settings where labeling is expensive. 
Query-by-committee (QBC)~\cite{query-by-committee} is a general query strategy framework that chooses the most informative query based on the disagreements among a committee of models. How to construct the committee space and how to measure the disagreements among committee members are questions must be answered when instantiating the QBC framework. In contrast, we interactively query users to learn objectives using a carefully designed search space PCS and propose ways to estimate query informativeness specific to our setting. 



\input{Aconclusion}


\bibliographystyle{ACM-Reference-Format}
\bibliography{refs,refs2,refs3,ref-networkdesign}

\clearpage

\appendix






\section{Additional Experimental Results}

\subsection{Evaluation on Perfect Oracle}
\label{app:fc}
The range of solution quality achieved by \name{} sometimes overlaps with \baseline{} in Fig~\ref{fig:fc-grouped} -- however \name{} outperformed \baseline{} in every scenario-topology combination. To demonstrate this for the \textbf{NF} scenario which saw the most overlap in Fig~\ref{fig:fc-grouped}, consider
Fig~\ref{fig:fc_individual} which presents a detailed breakdown of results by topology, and clearly illustrates \name{}'s out-performance. 

\begin{figure*}
\begin{subfigure}[c]{0.24\textwidth}
\centering
\includegraphics[width=\textwidth]{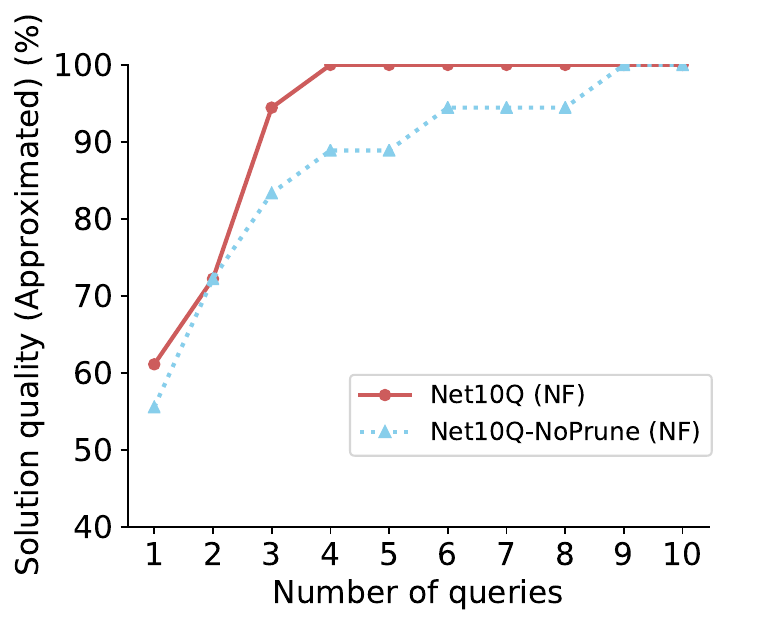}
\caption{Abilene.}
\end{subfigure}
\begin{subfigure}[c]{0.24\textwidth}
\centering
\includegraphics[width=\textwidth]{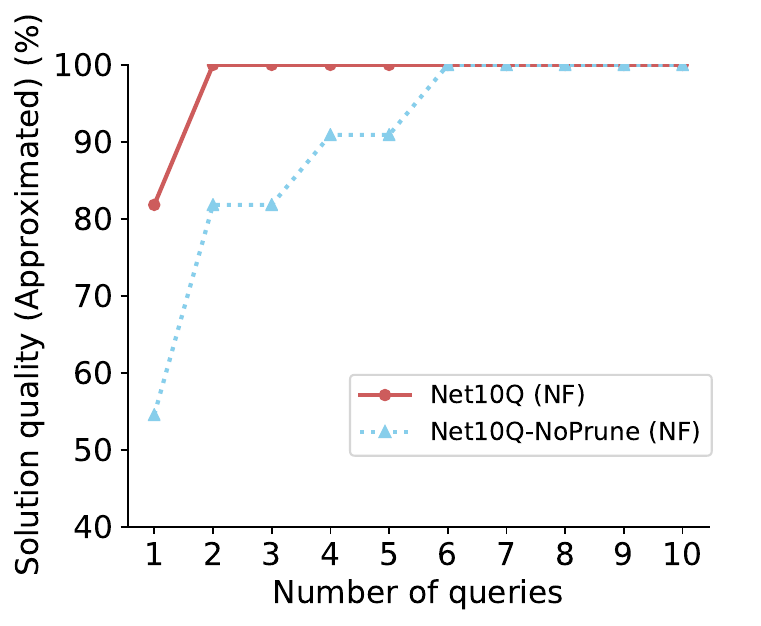}
\caption{B4.}
\end{subfigure}
\begin{subfigure}[c]{0.24\textwidth}
\centering
\includegraphics[width=\textwidth]{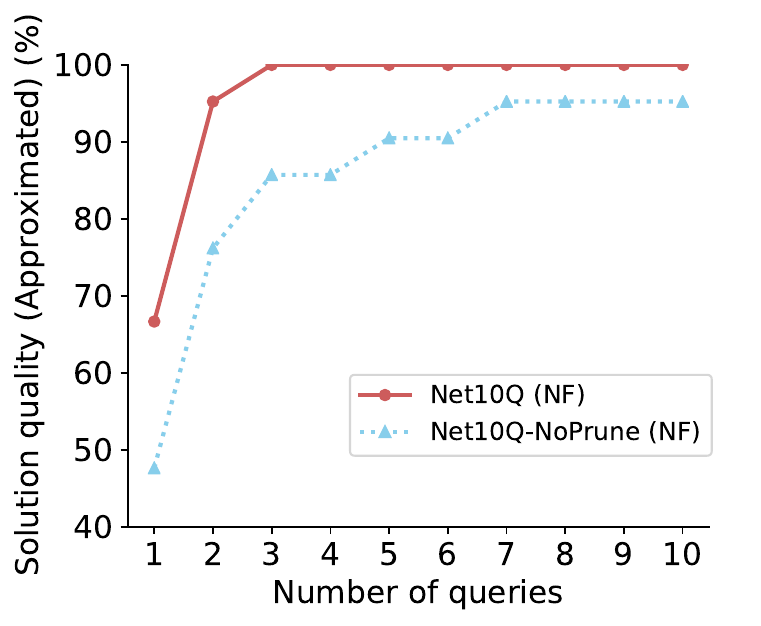}
\caption{CWIX.}
\end{subfigure}
\begin{subfigure}[c]{0.24\textwidth}
\centering
\includegraphics[width=\textwidth]{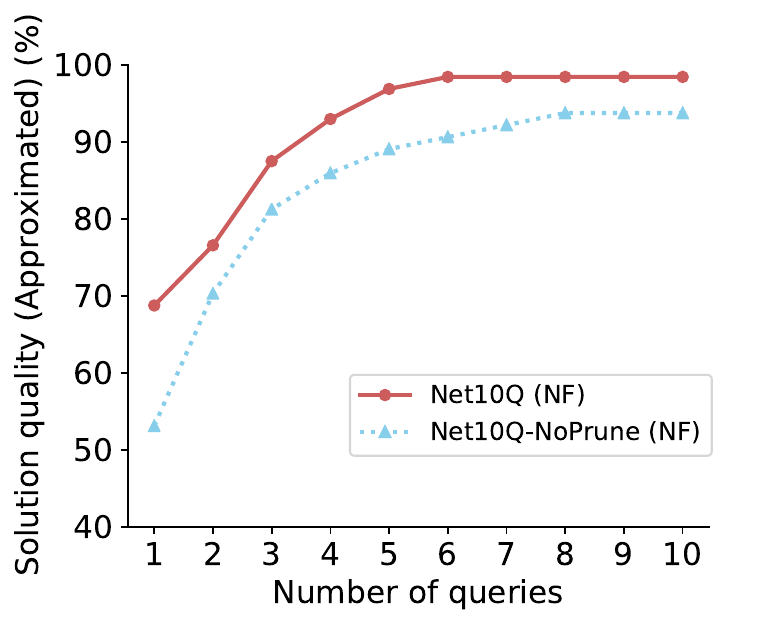}
\caption{BTNorthAmerica.}
\end{subfigure}
\begin{subfigure}[c]{0.33\textwidth}
\centering
\includegraphics[width=\textwidth]{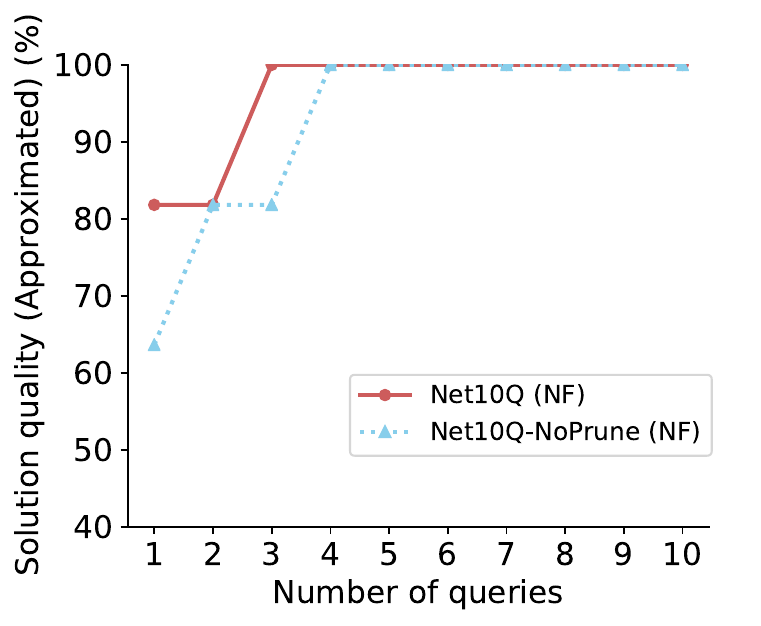}
\caption{Tinet.}
\end{subfigure}
\begin{subfigure}[c]{0.33\textwidth}
\centering
\includegraphics[width=\textwidth]{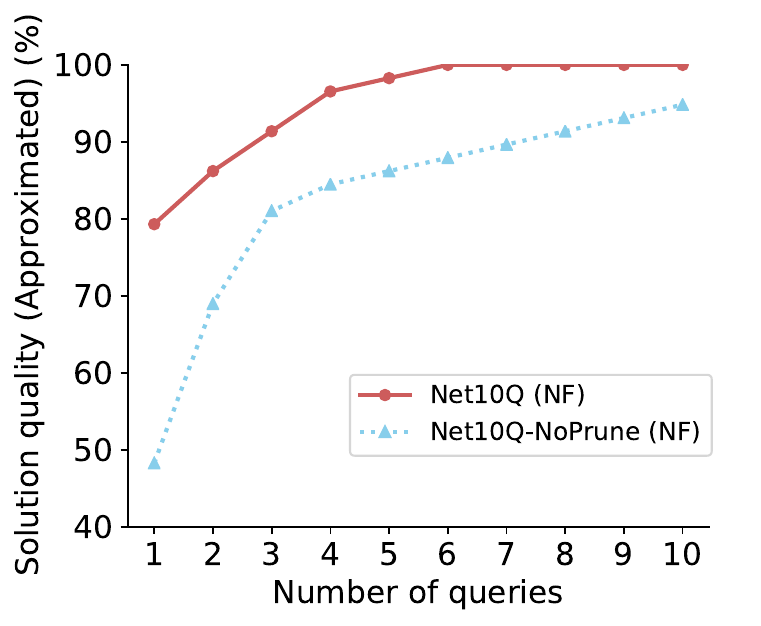}
\caption{Deltacom.}
\end{subfigure}
\begin{subfigure}[c]{0.32\textwidth}
\centering
\includegraphics[width=\textwidth]{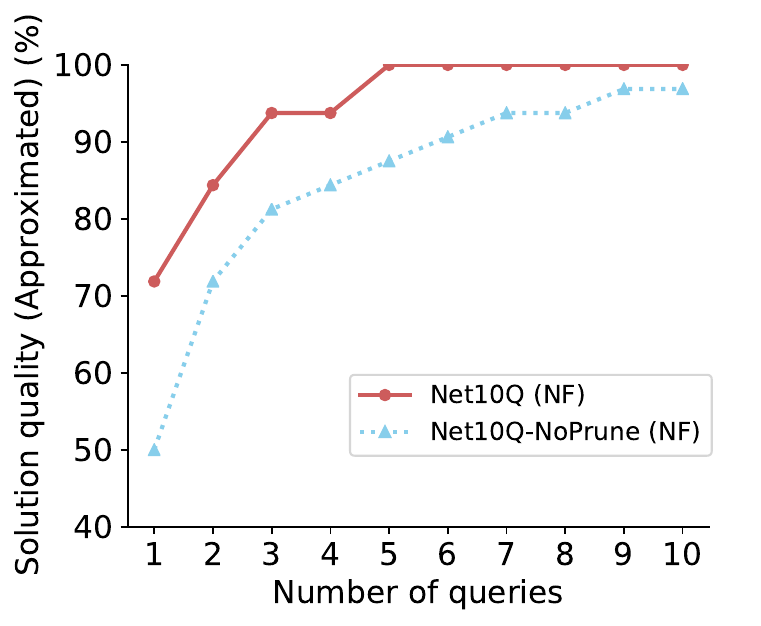}
\caption{Ion.}
\end{subfigure}
\caption{Comparing \name{} and \baseline{} with perfect oracle for \textbf{NF} (each subfigure for a topology). Curves to the left are better. \name{} outperforms in all topologies.}
\label{fig:fc_individual}
\end{figure*}

\subsection{Evaluation on Imperfect Oracle}
\label{app:exp}

\begin{figure*}[t]
\begin{subfigure}[c]{0.33\textwidth}
\centering
\includegraphics[width=\textwidth]{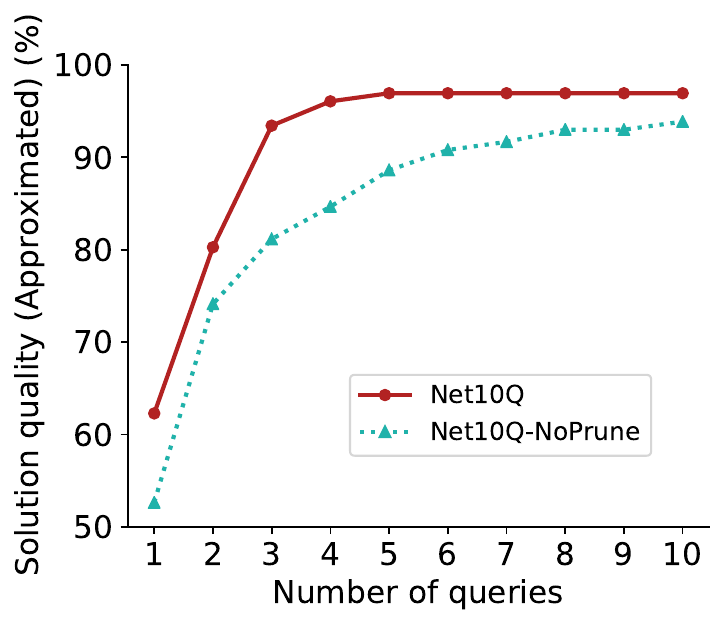}
\caption{MCF.}
\end{subfigure}
\begin{subfigure}[c]{0.33\textwidth}
\centering
\includegraphics[width=\textwidth]{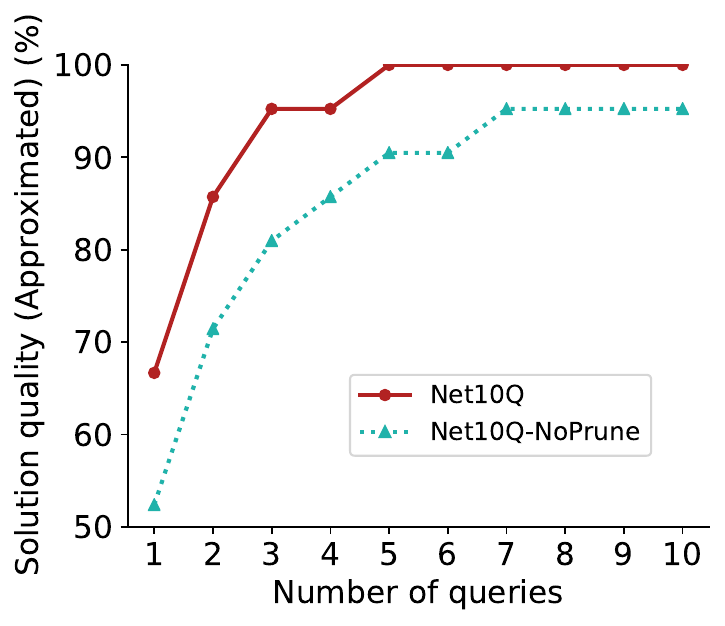}
\caption{NF.}
\end{subfigure}
\begin{subfigure}[c]{0.32\textwidth}
\centering
\includegraphics[width=\textwidth]{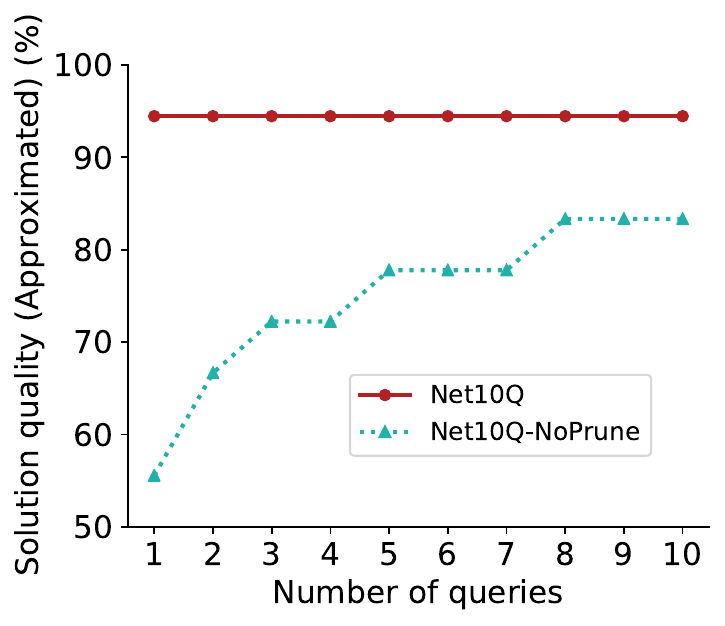}
\caption{OSPF.}
\end{subfigure}
\caption{Performance of \name{} with imperfect oracle ($p=10$) for \textbf{MCF}, \textbf{NF} and \textbf{OSPF} on CWIX.}
\label{fig:cwix_imperfect10}
\end{figure*}

\begin{figure*}[t]
\begin{subfigure}[c]{\textwidth}
\centering
\includegraphics[width=\textwidth]{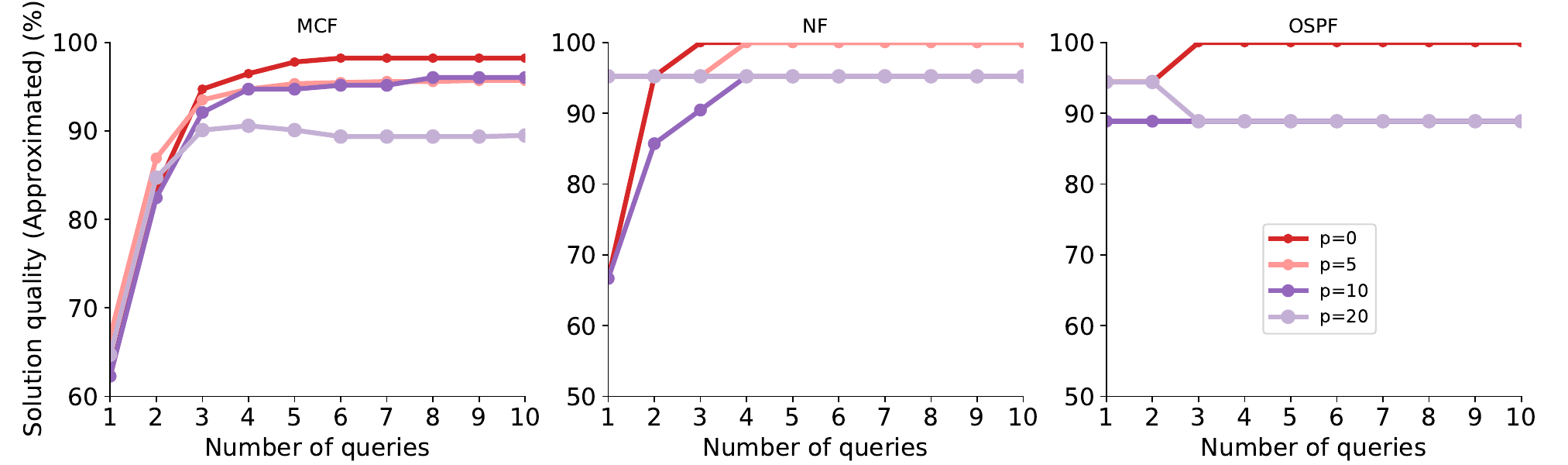}
\end{subfigure}
\caption{Performance of \name{} under different level of inconsistency ($p=0, 5, 10, 20$) on CWIX.}
\label{fig:imperfect10-1}
\end{figure*}

Fig~\ref{fig:robust} showed \name{} outperforms \baseline{} for \textbf{BW} on CWIX. Fig~\ref{fig:cwix_imperfect10} shows the performance of \name{} and \baseline{} when $p=10$ for the other three scenarios, namely, \textbf{MCF}, \textbf{NF} and \textbf{OSPF}. \name{} outperforms in these scenarios as well. Fig~\ref{fig:imperfect10-1} presents the performance of \name{} on CWIX when $p = 0, 5, 10, 20$ for \textbf{MCF}, \textbf{NF} and \textbf{OSPF}.  
While \name{} shows some performance degradation at higher inconsistency levels ($p=20$), it still achieved good solution quality.

\subsection{Sensitivity to Size of Pre-Computed Pool}
\label{app:pool}
To examine how sensitive \name{} is to the size of the pre-computed pool, we evaluated \name{} with pools of different size for \textbf{BW} on CWIX, sampled from a larger pool of size 5000. Fig~\ref{fig:pool} shows the solution quality achieved by \name{} after 10 queries. The solution quality in all cases was
based on the rank computed on the entire pool of size 5000.
The results indicate that performance does not substantially improve beyond a pool size of 300. 


\begin{figure}
\centering
  \includegraphics[width=0.5\columnwidth]{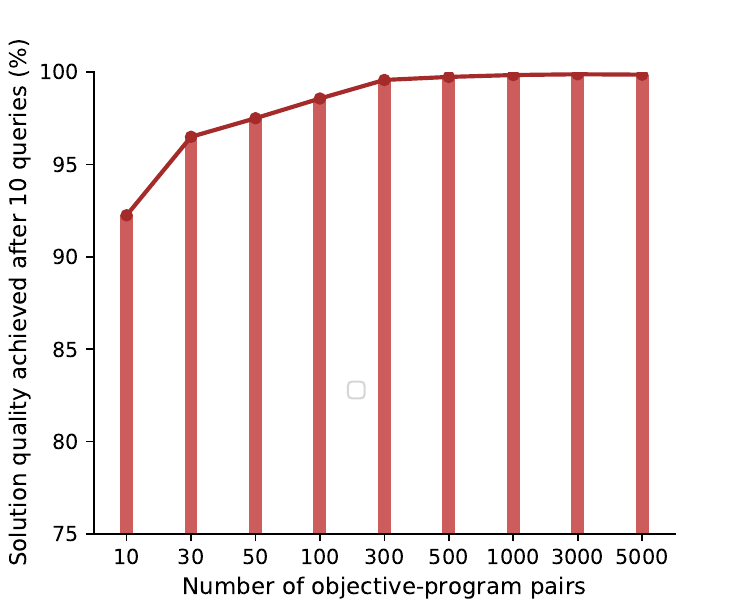}
  \caption{Performance of \name{} with different size of pre-computed pool for \textbf{BW} on CWIX.}
\label{fig:pool}
\end{figure}

\end{document}

%% file: macros.tex
\usepackage{xspace}
\usepackage{proof}
\usepackage[utf8]{inputenc}
\usepackage{graphicx}
\usepackage{pdfpages}
\usepackage{amsmath,amsthm}
\usepackage{amscd}
\usepackage{float}
\usepackage{epsfig}
\usepackage{xcolor}
\usepackage{color}
\usepackage{subcaption}
\usepackage{moresize}
\usepackage{stmaryrd}
\usepackage{mathrsfs}
\usepackage{mathtools}

\usepackage{mathpartir}

\usepackage{hhline}
\usepackage{changepage}
\usepackage{wrapfig}

\usepackage{xr}
\usepackage{threeparttable}
\usepackage{wrapfig}
\usepackage{epstopdf}

\newcommand{\throughput}{\textsf{throughput}}
\newcommand{\latency}{\textsf{latency}}
\newcommand{\throughputshort}{\textsf{thrpt}}
\newcommand{\latencyshort}{\textsf{ltncy}}
\newcommand{\utilization}{\textsf{utilization}}

\newcommand{\name}{{\sc Net10Q}}
\newcommand{\baseline}{{\sc Net10Q}{\sf -NoPrune}}

\newcommand{\Compare}{\textsc{Compare}}
\newcommand{\Validate}{\textsc{Propose}}

\newcommand{\Quality}{\textrm{Info}}

\newcommand{\defeq}{\stackrel{\textrm{def}}{=}}

\newcommand{\compareupdate}{{\color{blue} {\bf update}}}
\newcommand{\generatemore}{{\color{blue} {\bf generate-more}}}

\newcommand{\Rem}[3]{\big| #1 #2 #3 \big|}
\newcommand{\RemNEQ}[2]{\Rem{#1}{\lesssim_O}{#2}}
\newcommand{\RemNEQQ}[2]{\Rem{#2}{\gtrsim_O}{#1}}
\newcommand{\RemEQ}[2]{\Rem{#1}{\nsim_O}{#2}}
\newcommand{\RemOldBest}[1]{\Rem{#1}{>_O}{r\_best}}
\newcommand{\RemNewBest}[1]{\big| \textit{NewBest}(#1) \big| }

\font\xtiny=cmr10 scaled 600
\definecolor{siqrcolor}{gray}{0}

\def\tableentryraw#1#2#3{#1\parbox{17pt}{\raggedleft\color{siqrcolor}\xtiny#2#3\normalcolor}}
\def\tableentry#1#2#3{%
  \ifx\\#1\\%
    \multicolumn{1}{c|}{\scriptsize$\infty$}
  \else%
    \ifx\\#2\\%
      \tableentryraw{#1}{\centering\tiny$\infty$}{}%
    \else%
      \ifx\\#3\\%
        \tableentryraw{#1}{#2}{}%
       \else%
        \tableentryraw{#1}{#2}{(#3)}%
       \fi
    \fi%
  \fi%
}

\def\timeout#1{\tableentry{}{}{}}
\definecolor{darkgray}{gray}{.6}
\definecolor{lightgray}{gray}{.8}

\newcommand{\hide}[1]{} 

\numberwithin{equation}{section}

\newcommand{\codecolor}[1]{\textcolor{blue}{\textbf{\textsf{\small #1}}}}

\definecolor{dkgreen}{rgb}{0,0.3,0}
\definecolor{gray}{rgb}{0.5,0.5,0.5}
\definecolor{mauve}{rgb}{0.58,0,0.82}
\definecolor{light-gray}{gray}{0.80}

\usepackage{listings}
\lstdefinelanguage{spmd}{
  morekeywords = {
       loc, bit, bool, true, false
     , implements, harness
     , null
     , assert, assume
     , else
     , find, fix, fold, for, forall, function
     , generator, gen
     , if, while, int, float, bool, string
     , loop, simple, cond, val
     , fork, join
     , nil, null, none, new, malloc
     , option, or
     , ref, return
     , spmdfork, nprocs, spmdtransfer
     , void
     , concrete, sym
     , requires, ensures
     , invariant, decreases 
     , conj, exp
     , init, stmt},
  literate=
    {-}{--}1,
  morecomment=[l]{//}
}
\renewcommand{\scriptsize}{\fontsize{8.5}{9}\selectfont}
\makeatletter
\lst@AddToHook{TextStyle}{\let\lst@basicstyle\scriptsize\fontfamily{lmss}\selectfont}
\makeatother

\usepackage{url}

\usepackage{rotating}

\usepackage{enumitem}

\usepackage{algorithm2e}

\SetCommentSty{commentfont}
\usepackage{pstricks}
\usepackage{multirow}

\newcommand{\code}[1]{\textsf{#1}}

\newcommand{\conf}[1]{}

\newcommand{\ctr}{{\code{ctr}}}



%% file: aAbstract.tex
\begin{abstract}
When managing wide-area networks, network architects must decide how to balance multiple conflicting metrics, and ensure fair allocations to competing traffic while prioritizing critical traffic.
The state of practice poses challenges since architects must precisely encode their intent into formal optimization models using abstract notions such as utility functions, and ad-hoc manually tuned knobs. In this paper, we present the first effort to synthesize optimal network designs with indeterminate objectives using an interactive program-synthesis-based approach.
We make three contributions. 
First, we present comparative synthesis, an interactive synthesis framework which produces near-optimal programs (network designs) through two kinds of queries (\Validate{} and \Compare{}), without an objective explicitly given. 
%
Second, we develop 
the first learning algorithm for comparative synthesis
in which a voting-guided learner picks the most informative query in each iteration. We present theoretical analysis of the convergence rate of the algorithm. 
%
Third, we implemented \name{}, a system based on our approach, and demonstrate its effectiveness 
on four real-world network case studies using black-box oracles and simulation experiments, as well
as a pilot user study comprising network researchers and practitioners.
Both theoretical and experimental results show the promise of our approach.

\end{abstract}

%% file: Aintro.tex
\section{Introduction}
\label{sec:intro}
Synthesizing wide-area computer network designs typically involves solving \textit{multi-objective optimization} problems. For instance, consider the
task of managing the traffic of a wide-area network --- deciding the best routes and allocating bandwidth for them ---
the architect must consider myriad considerations.
She must choose from different routing approaches --- e.g., shortest path routing~\cite{OSPFWeights}, and routing along pre-specified paths~\cite{swan, ffc_sigcomm14}. The traffic may correspond to different classes of applications --- e.g., latency-sensitive applications such as Web search and video conferencing, and elastic applications such as video streaming, and file transfer applications~\cite{b4,swan,hierarchicalBW}. The architect may need to decide how much traffic to admit for each class of applications. It is desirable to make decisions that can ensure high throughput, low latency, and fairness across different applications, yet not all these goals may be simultaneously achievable. Likewise, a network must not only perform acceptably under normal conditions, but also under failures --- however, providing guaranteed performance under failures may require being sacrificing normal performance
~\cite{ffc_sigcomm14,pcf_sigcomm20}.


Traffic engineering formulates network design problems as optimization problems~\cite{OSPFWeights,r3:sigcomm10,b4,hierarchicalBW,pcf_sigcomm20}, e.g., minimizing a weighted sum of link utilization and latency subject to constraints. In this context, architect must provide the objectives as well-defined mathematical functions (which we henceforth refer
to as \emph{target functions}), which is a challenging task in the first place.
Even the simplest target functions may involve several knobs to capture the relative importance of different criteria (e.g., throughput, latency, and fairness, performance under normal conditions vs. failures). These knobs must be manually tuned by the architect in a ``trial and error'' fashion to result in a desired design. Further, many optimization problems (e.g., ~\cite{hierarchicalBW}) require architects to use abstract functions that capture the utility an application sees if 
a given design is deployed.
Utility functions are often non-linear (e.g., logarithmic) 
and may involve weights, which are not intuitive for a designer to specify in practice~\cite{srikant04}.
Finally, objectives are often chosen in a manner to ensure tractability, rather than necessarily reflecting the true intent of the architect.

This paper presents one of the first attempts to \emph{learn near-optimal network designs with indeterminate objectives}. Our work adopts an interactive, program-synthesis-based approach based on the key insight that when a user has difficulty in providing a concrete objective, it is relatively easy and natural to give preferences between pairs of concrete candidates. The approach may be viewed as a new variant of programming-by-example (PBE), where preference pairs are used as ``examples''  instead of input-output pairs in traditional PBE systems. 
%
%

In this paper, we make the following contributions:

$\bullet$ \textbf{A novel user-interaction paradigm (\S\ref{sec:foundation}).}
We present a rigorous formulation of an interactive synthesis framework which we refer to as \textbf{comparative synthesis}.
As Fig~\ref{fig:overview} shows, the framework consists of two major components: a comparative learner and a teacher 
(a user or a black-box oracle). The learner takes as input a clearly defined qualitative synthesis problem (including a parameterized program and a specification), a metric group and a target function space, and is tasked to find a near-optimal program w.r.t. the teacher's quantitative intent through two kinds of queries --- \Validate~and \Compare. 

The notion of comparative synthesis stems from a recent position paper~\cite{comparative-synthesis}. The preliminary work 
lacks formal foundation and query selection guidance,
and may involve
impractically many rounds of user interaction (see \S\ref{sec:motivation}). In contrast, the formalism of our framework 
enables the design and analysis of learning algorithms that 
strive to minimize the number of queries, and are amenable for real user interaction.

\begin{figure}[t!]
\begin{center}
\includegraphics[scale=0.33]{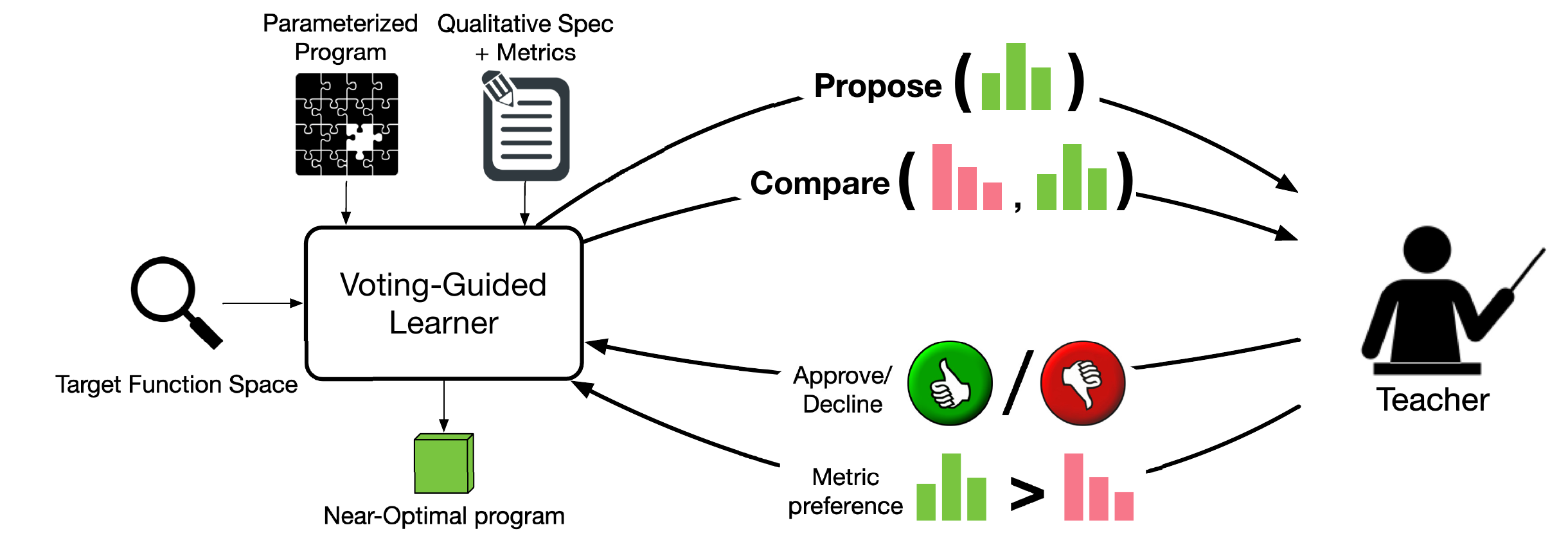}
\end{center}
\caption{Overview of comparative synthesis.}
\label{fig:overview}
\end{figure}

$\bullet$ \textbf{The first algorithm for comparative synthesis (\S\ref{sec:voting}).}
We develop 
the first, voting-guided learning algorithm for comparative synthesis, which provides a provable guarantee on the quality of the found program. 
The key insight behind the algorithm is that objective learning and program search are mutually beneficial and should be done in tandem.
The idea of the algorithm is to search over a special, unified search space we call \emph{Pareto candidate set}, and to pick the \emph{most informative} query in each iteration using a voting-guided estimation. 

We analyze the \emph{convergence} of voting-guided algorithm, i.e., how fast the solution approaches the real optimal as more queries are made. We prove that the algorithm guarantees the median quality of solutions to converge logarithmically to the optimal. When the target function space is sortable, which covers a commonly seen class of problems, a better convergence rate can be achieved --- the median quality of solutions converges linearly to the optimal.



$\bullet$ \textbf{Evaluations on network case studies and pilot user study (\S\ref{sec:evaluation}).}
We developed \name{}, an \emph{interactive network optimization system} based on our approach. We evaluated \name{} on four
real-world scenarios using oracle-based evaluation. Experiments show that \name{} only makes half or less queries than the baseline system to achieve the same solution quality, and robustly produces high-quality solutions with inconsistent teachers. We conducted a pilot study with \name{} among networking researchers and practitioners. Our study shows that user policies are diverse, and \name{} is effective in finding allocations meeting the diverse policy goals in an interactive fashion.


\vspace{.1in}
\textbf{A Lookahead:}
While our motivation and evaluation are from the context of network design, 
the challenge of indeterminate objectives is commonly seen in many quantitative synthesis problems beyond the networking domain. For example, the default ranker for the FlashFill synthesizer is manually designed and highly tuned by experts~\cite{gulwani2019quantitative}. In quantitative syntax-guided synthesis (qSyGuS)~\cite{qsygus}, the objective should be provided as a weighted grammar, which is nontrivial for average programmers.
Therefore, the problem we address in this paper can be viewed as an instance of \emph{specification mining}~\cite{specification-mining}, a long-standing problem in the formal methods community which recognizes that a precise specification may not always be available.
The key contributions of the paper, including comparison-based interaction  (\S\ref{sec:foundation}) and the voting-guided algorithm (\S\ref{sec:voting}), are
thus potentially applicable in 
other more traditional program synthesis domains in the future.

%% file: Amotivation.tex
\section{Motivation}
\label{sec:motivation}
In this section, we present background on network design, how it may be formulated as a program synthesis problem, and discuss challenges that we propose to tackle.

\textbf{Network design background.}
In designing Wide-Area Networks (WANs), Internet Service Providers (ISPs) and cloud providers must decide how to provision their networks, and route traffic so their traffic engineering goals are met. Typically WANs carry multiple classes of traffic (e.g., higher priority latency sensitive traffic, and lower priority elastic traffic). Traffic is usually specified as a matrix with cell $(i,j)$ indicating the total traffic which enters the network at router $i$ and that exits the network at router $j$. We refer to each pair ($i,j$) as a flow, or a source-destination pair. It is typical to pre-decide a set of tunnels (paths) for each flow, with traffic split across these tunnels in a manner decided by the architect, though traffic may also be routed along 
a routing algorithm that determines shortest paths (\S\ref{sec:scenarios}).

Given constraints on link capacities, it may not be feasible to meet the requirements of all traffic of all flows. An architect must decide how to allocate bandwidth to different flows of different classes and how to route traffic (split each flow's traffic across it paths) so desired objectives are met. In doing so, an architect must reconcile multiple metrics including throughput, latency, and link utilizations
~\cite{qarc,semioblivious-nsdi18,swan,b4}, ensure fairness 
across flows~\cite{srikant04,B4InfocomMaxMinFairness,hierarchicalBW}, and consider performance under failures~\cite{pcf_sigcomm20,ffc_sigcomm14,r3:sigcomm10,nsdiValidation17,cvarSigcomm19}. 

\textbf{Network design as program synthesis problems.}
Consider a variant of the classical multi-commodity flow
problem used in Microsoft's Software Defined Networking Controller SWAN~\cite{swan}, which we refer to as MCF.
MCF allocates traffic to tunnels optimally trading off
the total throughput seen by all flows with the weighted average flow latency~\cite{swan}. We consider a single class (see \S\ref{sec:scenarios}
for multiple classes).

Fig~\ref{fig:mcf} shows how the demand-capacity constraints may be described as a \emph{sketch-based synthesis} problem, in which the programmer specifies a sketch --- a program that contains unknowns to be solved for, and assertions to constrain the choice of unknowns. The \code{Topology} struct represents 
the network topology (we use
the Abilene topology~\cite{abilene} with 11 nodes, 14 links and 220 flows
as a running example). 
The \code{allocate} function should determine the bandwidth allocation (denoted by \codecolor{??}), which is missing and should be generated by the synthesizer. 
The function also serves as a test harness and checks that the synthesized allocation is valid, satisfying all demand and capacity constraints (lines~\ref{line:assert-start}--\ref{line:assert-end}). Finally, the \code{main} function takes the generated allocation, and computes and returns the total throughput and weighted latency.

\begin{figure}[t!]
\begin{lstlisting}[language=C,morekeywords={bit,new,assert,assume},breaklines=true,mathescape,xleftmargin=.2in,commentstyle=\color{gray}\it,numbers=left,backgroundcolor = \color{white}]
struct Topology {
  int n_nodes; int n_links; int n_flows;
  bit[n_nodes][n_nodes] links; 
  /* every link has a capacity and a weight, every flow consists of multiple links and has a demand */
  float[n_links] capacity; int[n_links] wght;
  bit[n_links][n_flows] l_in_f; float[n_flows] demand; ... 
}
Topology abilene = new Topology(n_nodes=11, n_links=220, $\dots$);
float[] allocate(Topology T) {
  float[T.n_flows] bw = $\color{blue} {\textbf{??}}$; // generate bandwidth allocation /**\label{line:hole}*/ 
  /* assert that every flow's demand is satisfied and every link's bandwidth is not exceeded */
  assert $\displaystyle \bigwedge_i$ bw[$i$] <= T.demand[$i$]; /**\label{line:assert-start}*/ 
  assert $\displaystyle \bigwedge_j \Big(\sum_i$ l_in_f[j][i] ? bw[i] : 0 $\Big) \leq$ T.capacity[j]; /**\label{line:assert-end}*/
  return bw; }
float[] main() {
  float[] alloc = allocate(abilene);
  /*compute the throughput and weighted latency*/
  float throughput = $\displaystyle \sum_{i} \textsf{alloc}[i]$, $~~$ latency = $\displaystyle \sum_{i} \Big( \textsf{alloc}[i] \cdot \sum_j \textsf{l\_in\_f}[j][i] \textsf{? wght[j] : 0} \Big)$;
  return {throughput, latency}; }
\end{lstlisting}
\vspace{-.1in}
\caption{MCF allocation encoded as a program sketch.}
\vspace{-.15in}
\label{fig:mcf}
\end{figure}

Now Fig~\ref{fig:mcf} has encoded all \emph{hard constraints} and represented a \emph{qualitative synthesis} problem, which can be solved by Sketch~\cite{sketch} easily. 
The bandwidth allocations generated by the synthesizer (the values of \code{bw}) is just a network design solving the MCF problem.\footnote{We will use network design and network program interchangeably in the paper, as network design can always be extracted from the synthesized network program.}
However, there are many different ways to fill the \codecolor{??}, corresponding to many different ways of assigning paths and leading to different throughput-latency combinations as computed in \code{main}. Which solution is the most desirable one?
Traditionally, the architect has to explicitly provide a \emph{target function} 
which maps each possible solution to a numerical value indicating the preference. Given a well-specified target function, the bandwidth allocation problem becomes a \emph{quantitative synthesis} problem and can be solved using existing techniques from both synthesis and optimization communities. E.g., in Fig~\ref{fig:mcf}, one can explicitly add a target function $O_{real}$ and use the \code{minimize} construct (cf. Sketch manual~\cite{sketch:manual}) to find the optimal solution. 


\textbf{Why synthesis with indeterminate objectives?}
Unfortunately, in practice, it is hard for network architects to precisely express their true intentions using target functions. For example, to capture the intuition that once the throughput (resp. latency) reaches a certain level, the marginal benefit (resp. penalty) may be smaller (resp. larger), an architect may need to use a target function like below:
\begin{multline}
\label{eqn:real}
{\small
O_{real}(\throughput, \latency) \defeq 2 \cdot \throughput - 9 \cdot \latency } \\
{\small \!\!\!\! - \max(\throughput - 350, 0) - 10 \cdot \max(\latency-28, 0)}
\end{multline}
%
More generally, the marginal reward in obtaining a higher bandwidth allocation is smaller capturing which may require a target function of the form 
$O(\throughput, \latency) \defeq
    \displaystyle 1 \cdot \log_n(\throughput/tmax + 1) + 5 \cdot \log_n(lmin / \latency + 1)
$
where $tmax$ is the maximum possible throughput and $lmin$ is the minimum possible latency.
Expressing such abstract target functions is not trivial, let alone the parameters associated with the functions. We present several other examples in \S\ref{sec:scenarios}.

\begin{wrapfigure}{r}{0.5\linewidth}
    \centering
    \includegraphics[scale=0.5]{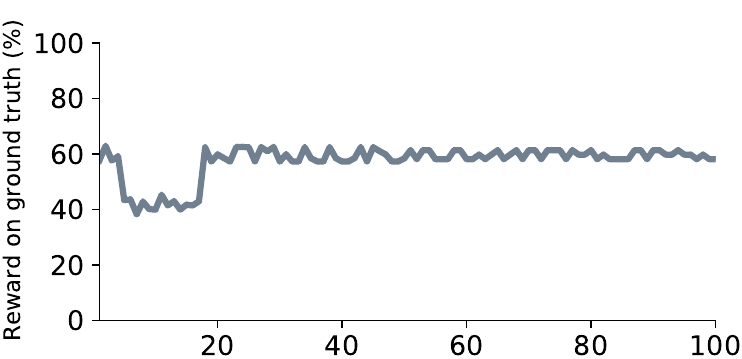}
    \vspace{-.1in}
    \caption{Na\"{i}ve objective synthesis~\cite{comparative-synthesis} of $O_{real}$.}
    \label{fig:obj_first}
    \vspace{-0.1in}
\end{wrapfigure}

\textbf{Na\"{i}ve objective synthesis is not enough.}
A preliminary effort \cite{comparative-synthesis} argued for synthesizing target functions 
by having the learner (synthesizer) iteratively query the teacher (user) on its preference between two concrete network designs. In each iteration, any pair of designs may be compared as long as there exist two target functions that (i) disagree on how they rate these designs, and (ii) both satisfy preferences expressed by the teacher in prior iterations. The process continues until no disagreeing target functions are found.
However, this work only discusses objective learning and does not explicitly consider design synthesis. Moreover, it does not address how to generate good preference pairs to minimize queries.
These limitations make this na\"{i}ve approach not amenable for real user interaction.
Fig~\ref{fig:obj_first} shows the performance of a design optimal for a target function synthesized using this procedure if it were terminated after a given number of queries. The resulting designs achieve a reward only 60\% of the true optimal design under the ground truth (Eq~\ref{eqn:real}), and there is hardly any performance improvement in the first 100 queries.

%% file: A-Sanjay-formalDefinition.tex
\section{Comparative Synthesis, Formally}
\label{sec:foundation}

In this section, we provide a formal foundation for the comparative synthesis framework, based on which we design and analyze learning algorithms.
The key novelty of our framework compared to past work on quantitative synthesis~\cite{Cerny2011,qsygus,synapse,superoptimization,superoptimization2,Chaudhuri14} is that comparative synthesis does not 
require the user to explicitly specify the objective. Instead, we approach synthesis via \emph{interaction through comparative queries}
where queries simply involve the users comparing two alternative programs and indicating which is more preferable. 
Since a user will only be willing to answer a small number of queries and may choose to stop at any point of the interactions, finding a perfect quantitative objective can be unrealistic. Therefore, our goal is to \emph{generate a near-optimal program within a budgeted number of queries}. As the real ground truth optimal is not accessible, we also introduce a natural notion, called \emph{quality of solution}, to estimate how close a solution is to optimal.

\textbf{Roadmap.}
In \S\ref{subsec:QuantSynthDef}, we formally define quantitative synthesis, a necessary first step for us to formally treat
comparative synthesis. Rather than restrict quantitative synthesis to objectives which are closed-form mathematical functions,
we formulate quantitative synthesis more generally as we motivate and discuss in \S\ref{subsec:QuantSynthDef}.
We formally define comparative queries, and solution quality in \S\ref{subsec:queryDef}. We conclude with a formal definition of comparative 
synthesis in \S\ref{subsec:CSDef}.



\subsection{Quantitative Synthesis with Metric Ranking}
\label{subsec:QuantSynthDef}
In this section, we present a formal definition of quantitative synthesis, as a first step towards
defining comparative synthesis more precisely. Quantitative synthesis may be viewed as a 
goal of synthesizing a program that meets a set of "hard constraints", while performing well on
a quantitative objective. Rather than restrict quantitative synthesis to objectives which are closed-form 
mathematical functions, we formulate quantitative synthesis when given a ranking over all possible programs 
(which we refer to as a \emph{metric ranking}). This more general formulation is motivated by the fact that we 
wish to capture a rich set of user policies in terms of relative preferences across programs, and not restrict 
the user to objectives that are closed-form mathematical functions. 


We start by defining qualitative synthesis (which captures the ``hard constraints'' that any acceptable program
must meet), and then discuss quantitative synthesis with metric ranking.

\begin{definition}[Qualitative synthesis problem]
A qualitative synthesis problem is represented as a tuple $(\mathcal{P}, C, \Phi)$ where $\mathcal{P}$ is a parameterized program, $C$ is the space of parameters for $\mathcal{P}$, and $\Phi$ is a verification condition with a single free variable ranging among $C$. The synthesis problem is to find a value $\ctr \in C$ such that $\Phi(\ctr)$ is valid.
\end{definition}

\begin{example}
\label{ex:mcf-qual}
Our running example can be formally described as a qualitative synthesis problem $\mathscr{L}_{MCF} = (\mathcal{P}_{\textsf{abilene}}, \mathbb{R}^{220}, \Phi_{\textsf{abilene}})$, where $\mathcal{P}_{\textsf{abilene}}$ is the program sketch presented in Fig~\ref{fig:mcf}, $\mathbb{R}^{220}$ is the search space of unknown hole (line~\ref{line:hole}) which includes 220 bandwidth values of the Abilene network, $\Phi_{\textsf{abilene}}$ is the verification condition, taking a candidate solution $c \in \mathbb{R}^{220}$ as input and checking whether $\mathcal{P}_{\textsf{abilene}}[c]$ satisfies all assertions in Fig~\ref{fig:mcf}. Any solution that satisfies assertions in Fig~\ref{fig:mcf} is a feasible program to the qualitative synthesis problem $\mathscr{L}_{MCF}$.
\end{example}

While a qualitative synthesis problem captures all hard constraints, there are potentially infinitely many solutions. Which one is the most desirable? \emph{Quantitative synthesis} concerns itself with this question and
extends a qualitative synthesis problem with a quantitative goal, which is evaluated using a metric group, as defined below.

\begin{definition}[Metric]
Given a parameterized program $\mathcal{P}[c]$ where $c$ ranges from a search space $C$, a metric with respect to $\mathcal{P}$ is a computable function $m_{\mathcal{P}}: C \rightarrow \mathbb{R}$. In other words, $m_{\mathcal{P}}$ takes as input a concrete program in the search space and computes a real value.
\end{definition}

\begin{definition}[Metric group]
Given a parameterized program $\mathcal{P}$, a $d$-dimensional metric group $M$ w.r.t. $\mathcal{P}$ is a sequence of $d$ metrics w.r.t. $\mathcal{P}$. We write $M_i$ for the $i$-th metric in the group and $M(c)$ for the value vector $\big(M_1(c), \dots, M_d(c)\big)$.
\end{definition}


\begin{example}
A metric can be computed from the syntactical aspects of the program. For example, a metric $\code{size}_{\mathcal{P}}(c)$ can be defined as the size of the parse tree for $\mathcal{P}[c]$.
\end{example}

\begin{example}
\label{ex:mcf-metric}
A metric can 
simply be
the value of a variable on a particular input (or with no input). In our running example in Fig~\ref{fig:mcf}, the two variables \code{throughput} and \code{latency} of the \code{main} function can be used to define two metrics. As the latency as a metric is not beneficial, i.e., smaller latency is better, we can simply use its inverse \code{-latency} as a beneficial metric. The two metrics form a metric group $M_{MCF} = (\code{throughput}, \code{-latency})$.
\end{example}


Now given a metric group, the quantitative intent of a user can be captured either \emph{syntactically} (using a target function) or \emph{semantically} (using a metric ranking). We formally define them below and discuss their relationship.

\begin{definition}[Target function]
Given a metric group $M$, a \emph{target function} with respect to $M$ is a computable function $\mathbb{R}^{|M|} \rightarrow \mathbb{R}$.
\end{definition}

\begin{definition}[Metric ranking]
\label{def:ranking}
Given a $d$-dimensional metric group $M$, a \emph{metric ranking} for $M$ is a total preorder $\lesssim_M$ over $\mathbb{R}^{d}$.
In other words, $\lesssim_M$ satisfies the following properties:
 for any $u \in \mathbb{R}^{d}$, $u \lesssim_M u$ (\emph{reflexivity});
 for any three vectors $u, v, w \in \mathbb{R}^{d}$, if $u \lesssim_M v$ and $v \lesssim_M w$, then $u \lesssim_M w$ (\emph{transitivity});
 for any two vectors $u, v \in \mathbb{R}^{d}$, $u \lesssim_M v$ or $v \lesssim_M u$ (\emph{connexivity}).
\end{definition}


We write $u \sim_M v$ if $u \lesssim_M v$ and $v \lesssim_M u$. In this paper, we also flexibly write $u \gtrsim_M v$, $u <_M v$, $u >_M v$ and $u \nsim_M v$ with the expected meaning. Moreover, we also abuse $\lesssim_M$ and other derived symbols we just described as relations between programs: when the metric group $M$ is clear from the context, for any two program parameters $c_1, c_2 \in \textrm{dom}(M)$, we write $c_1 \lesssim_M c_2$ to indicate that $M(c_1) \lesssim_M M(c_2)$.

\paragraph{\textbf{Why metric ranking?}} Target function and metric ranking are closely related, but metric ranking is a more general and unique representation for quantitative intent, and can capture a richer set of user policies in terms of which of two feasible programs is preferable. First, 
every target function implicitly determines a metric ranking (see Def~\ref{def:of-mc} below).
Second, multiple target functions may have the same metric ranking. For example, any target functions $O$ and $O'$ such that $O'(v) = 2 \cdot O(v)$ for any $v \in \mathbb{R}^{|M|}$ have the same metric ranking. Third, some metric ranking does not correspond to any target function, e.g., one can define a metric ranking $\lesssim$ between integer metric values such that $n_1 \lesssim n_2$ if and only if the $n_1$-th digit of $\Omega$ is less than or equal to the $n_2$-th digit of $\Omega$, where $\Omega$ is a Chaitin's constant~\cite{chaitin} representing the probability that a randomly generated program halts.
To this end, we define quantitative synthesis problem below using metric ranking.

\begin{definition}
\label{def:of-mc}
Given a target function $O$ w.r.t. a $d$-dimensional metric group $M$, the corresponding metric ranking $\lesssim_O \subset \mathbb{R}^{d} \times \mathbb{R}^{d}$ is defined as follows: for any two program parameters $c_1, c_2 \in \textrm{dom}(M)$, $c_1 \lesssim_O c_2$ if and only if $O(M(c_1)) \leq O(M(c_2))$. It can be easily verified that $\lesssim_O$ is indeed a metric ranking.
\end{definition}


\begin{definition}[Quantitative synthesis problem]
\label{def:quantitative}
A quantitative synthesis problem is represented as a tuple $\mathscr{Q} = (\mathcal{P}, C, \Phi, M, \lesssim_M)$ where $(\mathcal{P}, C, \Phi)$ forms a qualitative synthesis problem $\mathscr{Q}^{\textrm{qual}}$, $M$ is a metric group w.r.t. $\mathcal{P}$ and $\lesssim_M$ is a metric ranking for $M$. The synthesis problem is to find a solution $\code{ctr}$ to $\mathscr{Q}^{\textrm{qual}}$ such that for any other solution $\code{ctr'}$, $\code{ctr'} \lesssim_{M} \code{ctr}$.
\end{definition}

\begin{example}
\label{ex:mcf-quan}
With the metric group $M_{MCF}$ defined in Example~\ref{ex:mcf-metric}, the function $O_{real}$ defined in Equation~\ref{eqn:real} is a 2-dimensional target function with the corresponding metric ranking $\lesssim_{{real}}$. Then the qualitative synthesis problem $\mathscr{L}_{MCF} = (\mathcal{P}_{\textsf{abilene}}, \mathbb{R}^{220}, \Phi_{\textsf{abilene}})$ can be extended to a quantitative synthesis problem 
$
\mathscr{Q}_{MCF} \defeq (\mathcal{P}_{\textsf{abilene}},~ \mathbb{R}^{220},~ \Phi_{\textsf{abilene}},~ M_{MCF},~ \lesssim_{{real}})
$.
\end{example}

\subsection{Interaction Through Comparative Queries}
\label{subsec:queryDef}


Quantitative synthesis problem as defined in Def~\ref{def:quantitative} expects a metric ranking explicitly or implicitly (e.g., through a target function). Comparative synthesis is more challenging as it seeks to synthesize a program that is near-optimal in terms of the objective, but without being explicitly given the objective. To achieve the goal, our comparative synthesis framework is interactive between a {\bf learner} and a {\bf teacher} (see Fig~\ref{fig:overview}). As the teacher may choose to stop at any point of the interactions, our comparative synthesis framework maintains the best candidate solution found through the synthesis process and recommends the best solution confirmed by the teacher when terminated.

As the quantitative objective is assumed to be very complex and not direct accessible from the teacher, the comparative learner can only make several types of queries to the teacher, whose responses provide indirect access to the specifications. The query types serve as an interface between the learner and the teacher, and different query types lead to different synthesis power (e.g., see the query types discussed in~\cite{ogis}). 
What makes our framework special is that the learner can make queries about the \emph{metric ranking} --- queries that 
compare two program (based on their corresponding metric value vectors).

Let us fix a parameterized program $\mathcal{P}$ and a metric group $M$. The learner and the teacher interact using two types of queries:~\footnote{We believe our framework can be extended in the future to support more types of queries.}

$\bullet$ $\Compare(c_1, c_2)$ {\bf query:} The learner provides two concrete programs $\mathcal{P}[c_1]$ and $\mathcal{P}[c_2]$, and asks ``Which one is better under the target metric ranking $\lesssim_M$?'' 
The teacher responds with $<$ or $>$ if one is strictly better than the other, or $=$ if $\mathcal{P}[c_1]$ and $\mathcal{P}[c_2]$ are considered equally good.

$\bullet$ $\Validate(c)$ {\bf query:}~\footnote{Note that $\Validate(c)$ can be viewed as a special $\Compare$ query between $c$ and $r\_best$, the goal of the query is slightly different: $\Validate(c)$ intends to beat the running best using $c$, while $\Compare(c_1, c_2)$ intends to distinguish very close solutions $c_1$ and $c_2$ to learn the teacher's intent. $\Compare(c_1, c_2)$ will \textbf{not} update the running best.} The learner proposes a candidate program $\mathcal{P}[c]$ and asks ``Is $\mathcal{P}[c]$ better than the running best candidate $r\_best$?'' 
If the teacher finds that $\mathcal{P}[c]$ is not better than the running best, 
she can respond with $\bot$.
Otherwise, the teacher considers that $\mathcal{P}[c]$ is better 
and responds with $\top$; in that case, the running best will be updated to $\mathcal{P}[c]$.



\vspace{.05in}
Now in comparative synthesis, the specification (a metric ranking $R$) is hidden to the learner. Instead, the learner can approximate/guess the specifications by making queries to the teacher. Ideally, the teacher should be \emph{perfect}, i.e., the responses she makes to queries are always consistent and satisfiable. Formally,
\begin{definition}[Perfect teacher]
A teacher $\mathcal{T}$ is \emph{perfect} if there exists a metric ranking $\lesssim_M$ such that: 1) for any query $\Compare(v_1, v_2)$, the response is ``$<$'' if $v_1 \lesssim_M v_2$ and $v_2 \not \lesssim_M v_1$, or ``$>$'' if $v_1 \not \lesssim_M v_2$ and $v_2 \lesssim_M v_1$, or ``$=$'' if $v_1 \lesssim_M v_2$ and $v_2 \lesssim_M v_1$;
2) for any query $\Validate(c)$ with the current running best $r\_best$, the response is ``$\top$'' if $c >_M r\_best$; or ``$\bot$'' otherwise. 
\end{definition}
We denote the perfect teacher w.r.t. $\lesssim_M$ as $\mathcal{T}_{\lesssim_M}$. We also denote the metric ranking $\lesssim_M$ represented by a perfect teacher $\mathcal{T}$ as $\lesssim_{\mathcal{T}}$.
A perfect teacher guarantees that an optimal solution exists among all candidates.
For now, let us assume that the teacher is perfect, i.e., consistent and able to answer all queries; but in the real world, a human teacher may be inconsistent and responds incorrectly. We do consider imperfect teachers in our evaluation (see \S\ref{sec:oracle}).


\paragraph{\textbf{Why budgeted number of queries?}} Ideally, the goal of the learner is to find an objective target (in the form of target function or metric ranking) that matches the teacher's mind and the corresponding optimal program that optimizes the objective. However, finding the target function can be impossible as the objective target may have no closed-form representation and not in the target function space. As the teacher is free to terminate the synthesis process at any point, pinpointing the target function in a potentially infinitely large space can also be impossible, even if the target function is in the target function space.

Therefore, the goal of the learner is \emph{to spend a budgeted number of queries and to produce a near-optimal program from the perspective of the teacher}. Note that the learner may use a conjectured objective to guide the search process, finding a perfect target function is not a goal. This is also a key insight for our algorithm design (cf. \S\ref{sec:voting}).

Now to determine how close a solution is to the ground truth optimal, we introduce a natural notion called \emph{quality of solution} which is intuitively the ``relative rank'' of the solution among all solutions. E.g., a solution of quality 0.9 is better than or equal to 90\% of possible solutions. From a probability theory point of view, the quality is just the cumulative distribution function (CDF).
Below we formally define the quality of solutions.

\begin{definition}[Quality of solution]
\label{def:quality}
Let $\mathscr{Q} = (\mathcal{P}, C, \Phi, M, \lesssim_{M})$ be a quantitative synthesis problem and let $\ctr$ be a solution to $\mathscr{Q}^{\textrm{qual}}$. The quality of $\ctr$ is defined as
\[\textrm{Quality}_{\mathscr{Q}}(\ctr) \defeq \displaystyle P(\ctr \geq_{\mathcal{T}} X_{\mathscr{Q}}^{M})\]
where $X_{\mathscr{Q}}^{M}$ is a variable randomly sampled from the uniform distribution for $\textrm{Solutions}(\mathscr{Q}^{\textrm{qual}})$.
\end{definition}
In particular, when $\textrm{Quality}_{\mathscr{Q}}(c) = 1$, $c$ is better than or equal to all other possible solutions, i.e., is the optimal solution under the teacher's preference.
Note that computing the exact quality can be very expensive, if not impossible. However we can \emph{estimate} the quality by sampling, as we do in evaluation (see \S\ref{sec:oracle}).

\begin{example}
\label{ex:mcf-compare}
The quantitative synthesis problem $\mathscr{Q}_{MCF}$ in Example~\ref{ex:mcf-quan} involves a metric ranking $\lesssim_{real}$. 
Let $\mathcal{T}_{real}$ be a perfect teacher w.r.t. $\lesssim_{real}$. Table~\ref{tab:example-run} illustrates how a voting-guided learning algorithm (which we present later in \S\ref{sec:voting}) serves as the learner and learns a near-optimal solution to $\mathscr{Q}_{MCF}$ through queries to $\lesssim_{real}$. 
In the first iteration, the learner solves the synthesis problem in Fig~\ref{fig:mcf} and gets a first mediocre allocation $P_0$
and
presents it to the architect, using query $\Validate(P_0)$. The teacher accepts the proposal as this is the first running best. 
In the sixth iteration,
the learner presents two programs $P_6$ and $P_7$ to the teacher and asks her to compare them. Based on the feedback that the architect prefers $P_6$ to $P_7$, the learner proposes $P_8$ which is confirmed by the teacher to be the best program so far. After seven queries, the running best is already very close to the optimal under the real objective (Quality of this solution has already achieved $97.8\%$). If the teacher wishes to answer more queries, the solution quality can be further improved.
\end{example}

\begin{table}
\caption{A Comparative Synthesis run for Example~\ref{ex:mcf-compare}}
\vspace{-.1in}
\centering
\scriptsize
\begin{tabular}{|c|c|c|c|c|c|}
\hline
Iter & Candidate Allocation & Query & Response & Running Best & Quality \\ \toprule
$1$ & $P_0(\throughputshort=205.2, \latencyshort=10.3)$ & $\Validate(P_0)$ & $\top$ & $P_0$ & 32.8\% \\ \hline
$2$ & $P_1(\throughputshort=470.2, \latencyshort=33.0)$ & $\Validate(P_1)$ & $\top$ & $P_1$ & 73.6\%  \\ \hline
$3$ & $P_2(\throughputshort=385.2, \latencyshort=24.5)$ & $\Validate(P_2)$ & $\top$ & $P_2$ & 92.8\%  \\ \hline
\dots & \dots & \dots & \dots & \dots & \dots  \\ \hline
\multirow{2}{*}{$6$} & $P_6(\throughputshort=405.4, \latencyshort=26.5)$ & \multirow{2}{*}{$\Compare(P_6, P_7)$} & \multirow{2}{*}{$P_6$} & \multirow{2}{*}{$P_2$} & \multirow{2}{*}{92.8\%}  \\ \cline{2-2}
& $P_7(\throughputshort=377.8, \latencyshort=23.8)$ &  &  &  &  \\ \hline
$7$ & $P_8(\throughputshort=392.9,\latencyshort=25.3)$ & $\Validate(P_8)$ & $\top$ & $P_8$ & 97.8\% \\ \hline
\end{tabular}
\label{tab:example-run}
\vspace{-.1in}
\end{table}

\subsection{The Comparative Synthesis Problem}
\label{subsec:CSDef}

Given the approximation nature of query-based interaction and the quality of solution defined above, the learner is tasked to solve what we call \emph{comparative synthesis problem}, which is formally defined below.
\begin{definition}[Comparative synthesis problem]
\label{def:comparative}
A comparative synthesis problem is represented as a tuple $\mathscr{C} = (\mathcal{P}, C, \Phi, M, \mathcal{T})$ where $\mathcal{P}$ is a parameterized program, $C$ is the space of parameters for $\mathcal{P}$, $\Phi$ is a verification condition for $\mathcal{P}$, $M$ is a metric group w.r.t. $\mathcal{P}$ and $\mathcal{T}$ is a perfect teacher,
such that $(\mathcal{P}, C, \Phi, M, \lesssim_{\mathcal{T}})$ forms a quantitative synthesis problem, which is denoted as $\mathscr{Q}$.
The synthesis problem is to find, by making a sequence of $\Compare$ and $\Validate$ queries to the teacher $\mathcal{T}$, a near-optimal solution $\ctr$ to $\mathscr{Q}$ with a provable guarantee on $\textrm{Quality}_{\mathscr{Q}}(\ctr)$.
\end{definition}

%% file: Aevaluation.tex
We have prototyped the comparative synthesis framework and the voting-guided learning algorithm as \name{} --- an interactive system that produces near-optimal network design by asking 10 questions to the user --- through which we evaluate the effectiveness and efficiency of our approach. We selected four real-world network design scenarios and conducted experiments with both oracles and human users. Our evaluations were conducted
on seven real-world, large-scale internet backbone \textbf{topologies} obtained from~\cite{abilene, b4} (sizes summarized in Table~\ref{tab:topo}). Note that the size of our largest topologies, namely Deltacom and Ion, are already beyond the ones typically considered in the traffic engineering community.


\subsection{Network Optimization Problems}
\label{sec:scenarios}

We summarize the four optimization scenarios in Table~\ref{tab:scenarios}, including their metric groups, target function spaces and sortability. 
We present some details below.

\begin{table}
\centering
\scriptsize
\caption{Summary of topologies.}
\vspace{-.1in}
\begin{tabular}{|c||c|c|c|c|c|c|c|}
    \hline
    Topology & Abilene & B4 & CWIX & BTNorthAmerica & Tinet & Deltacom & Ion \\
    \hline
    \#nodes & 11 & 12 & 21 & 36 & 48 & 103 & 114 \\
    \hline
    \#links & 14 & 19 & 26 & 76 & 84 & 151 & 135 \\
    \hline
    \end{tabular}
    \label{tab:topo}
\end{table}

\begin{table}
\centering
\caption{Summary of optimization scenarios.}
 \vspace{-.1in}
\centering
\resizebox{\textwidth}{!}{
\begin{tabular}{|c|c|c|c|}
\hline
Scenario        &Metric group        &Target function space   &Sortable? \\ \toprule
\textbf{MCF}                                                 &( \throughput, -\latency )      
&\begin{tabular}{@{}c@{}} $\throughput~*~?? ~-~ \max(\throughput - ??,~ 0) ~*~ ??$ \\ $- \latency ~*~ ?? ~-~ \max(\latency~-~??,~ 0) ~*~ ??$ \end{tabular}  &Yes \\ \hline
\textbf{BW}      &\begin{tabular}{@{}c@{}} $\big($ $avg_k$: average allocation to \\ the flows in the $k$-th class $\big)$  \end{tabular}              
&$\displaystyle \sum_{1 \leq k \leq K} w_k \log \big(avg_k\big) ~~~~~~(w_k > 0)$  &No \\ \hline
\textbf{NF}                                                
&\begin{tabular}{@{}c@{}} $\big($ $zn_{i}, zf_{i}$: guaranteed fraction of the \\ traffic demand of group $i$ under normal \\ conditions and failures respectively $\big)$ \end{tabular}     
&$\displaystyle \sum_{i} wn_{i} * zn_{i} + wf_{i} * zf_{i} ~~~~~~(wn_{i}, wf_{i} > 0)$  &No \\ \hline
\textbf{OSPF}                                                & $( -\latency, -\utilization )$      
&\begin{tabular}{@{}c@{}c@{}}
$\begin{cases}
\utilization & \utilization > ?? \\
\latency ~*~ ?? ~+~ \utilization ~*~ ?? & ?? < \utilization < ?? \\
\latency & otherwise
\end{cases}$
\end{tabular}  
&Yes \\ \hline
\end{tabular}
}
\label{tab:scenarios}
\end{table}

\textbf{Balancing throughput and latency (MCF).}
This is our running example based on~\cite{swan} described throughout the paper. This bandwidth allocation problem focuses on a single traffic class and considers balancing the throughput and latency in the network. 

\textbf{Utility maximization with multiple traffic classes (BW).}
A well-studied optimization problem is maximizing utility when allocating bandwidth
to traffic of different classes~\cite{srikant04, hierarchicalBW, jrex:JSAC13}.
Many applications such as file transfer have concave utility functions
which indicate that as more traffic is received, the marginal utility in obtaining a higher allocation is smaller. A common concave utility function which is widely used is a logarithmic utility function, where a flow that receives a bandwidth allocation of $x$ gets a utility of $\log x$.
Consider $N$ flows, and $K$ classes. Each flow belongs to one of the classes with $F^k$ denoting the set of flows belonging to class $k$. The weight of class $k$ is denoted by $w_k$ and is a knob manually tuned today to control the priority of the class, which we treat as an unknown in our framework.


\textbf{Performance with and without failures (NF).}
Resilient routing mechanisms guarantee the network does not experience congestion on failures~\cite{pcf_sigcomm20,ffc_sigcomm14,r3:sigcomm10,nsdiValidation17,cvarSigcomm19} by solving optimization problems that conservatively allocate bandwidth to flows while planning for a desired set of failures. We consider the model used in ~\cite{ffc_sigcomm14} to determine how to allocating bandwidth to flows while being robust to single link failure scenarios. We consider an objective with (unknown) knobs $w_{ni}$ and $w_{fi}$ that trade off performance under normal conditions and failures tuned differently for each group of flows $i$. 

%

\textbf{Balancing latency and link utilization (OSPF).}
Open Shortest Path First (OSPF) is a widely used link-state routing protocol for intra-domain internet and the traffic flows are routed on shortest paths~\cite{OSPFWeights}. A variant of OSPF routing protocol assigns a weight to each link in the network topology and traffic is sent on paths with the shortest weight and equally split if multiple shortest paths with same weight exist. By configuring the link weights, network architect can tune the traffic routes to meet network demands and optimize the network on different metrics~\cite{OSPFWeights}. We consider a version of the OSPF problem where link weights must be tuned to ensure link utilizations are small while still ensuring low latency  paths~\cite{LatencyUtilizationSigcomm18}. 
%
Intuitively, when utilization is higher than a threshold, it becomes the primary metric to optimize, and when lower than a threshold, 
minimizing latency is the primary goal. In between the thresholds, both latency and utilization are important, and can be scaled  in a manner chosen by the network architect. We treat the thresholds and the scale factors as unknowns in the objective.




\subsection{Implementation} 
\label{sec:implementation}
Note that in \name{}, once the scenario and the topology are fixed, we can pre-compute a large pool of objective-program pairs, from which the PCS is generated.
For each scenario-topology combination, we used the templates shown in Table~\ref{tab:scenarios} to generate a pool of random target functions. Then for all scenarios except for \textbf{OSPF}, we generate their corresponding optimal allocations using Gurobi~\cite{gurobi}, a state-of-the-art solver for linear and mixed-integer programming problems.
For \textbf{OSPF}, as we are not aware of any existing tools that can symbolically solve the optimization problem, we used traditional synthesis approaches (cf. Fig~\ref{fig:mcf}) to generate numerous feasible link weight assignments.
The pre-computed target functions and allocations are paired to form a large PCS serving as the candidate pool. We provide details of pool creation in~\ref{subsec:scaling}.


When the teacher is an imperfect oracle or a human user, inconsistent answers may potentially result in the algorithm unable to determine objectives that meet all user preferences. To ensure \name{} robust to an imperfect oracle, inspired by the ensemble methods~\cite{ensemble-method},
we implemented \name{} as a multi-threaded application where a primary thread accepts all inputs and the backup threads run the same algorithm but randomly discard some user inputs. In case no objective could satisfy all user preferences, a backup thread with the largest satisfiable subset of user inputs would take over.



\subsection{Oracle-Based Evaluation}
\label{sec:oracle}

We used \name{} to solve all scenario-topology combinations described above, through interaction with (both perfect and imperfect) oracles who answer queries based on their internal objectives.
As a first-of-its-kind system, \name{} does not have any similar systems to compare with. Therefore, we developed a variant of \name{} which adopts a simple but aggressive strategy: repeatedly proposing optimal candidates generated from randomly picked target functions. We call this baseline algorithm \baseline{}, as the teacher's preference is not used to prune extra candidates from the search space. 
As a solution's real quality (per Def~\ref{def:quality}) is not practically computable, we approximate its quality using its rank in our pre-computed candidate pool.\footnote{The quality is computed among Pareto optimal solutions only. In other words, the solution quality as per Definition~\ref{def:quality} should be higher than what we report here.} Moreover, as \name{} involves random sampling, we ran each synthesis task 301 times and reported the median of the (approximated) solution quality achieved after every query.


\subsubsection*{Evaluation on perfect teacher}

We built an oracle to play the role of a perfect teacher who answers all queries correctly based on a ground truth objective. 
For each scenario, we as experts manually crafted a target function that fits the template and reflects practical domain knowledge.\footnote{The ground truth does not have to match the template; see \S\ref{sec:userstudy} for human teacher who is oblivious to the template.}

We presents the performance of \name{} and \baseline{} on solving four network optimization problems (cf. Table~\ref{tab:scenarios}) on seven different topologies (cf. Table~\ref{tab:topo}).
Our key observation is that \emph{\name{} performed constantly better than \baseline{} in every scenario-topology combination.}
In the interest of space, we collected the quality of solutions achieved over all seven topologies for each optimization scenario and presented the median (shown as dots) and the range from max to min (shown as bars).
As Fig~\ref{fig:normal} shows, our voting-guided algorithm is very effective. \name{} always only needs 5 or fewer queries to obtain a solution quality achieved by \baseline{} in 10 queries.
We note that 
although the all-topology range for \name{} sometimes overlaps with the corresponding range for \baseline{} (primarily for the \textbf{NF} scenario), 
\name{} still outperformed \baseline{} for every topology. We leave the topology-wise results for \textbf{NF} in Appendix~\ref{app:fc}. 
Further, in all the cases where we could compute the optimal under the ground truth objective, we confirmed that programs recommended by \name{} achieved at least 99\% of the optimal.

\begin{figure}
    \centering
    \begin{minipage}{0.65\textwidth}
    \begin{subfigure}[c]{0.49\textwidth}
    \centering
    \includegraphics[width=\textwidth]{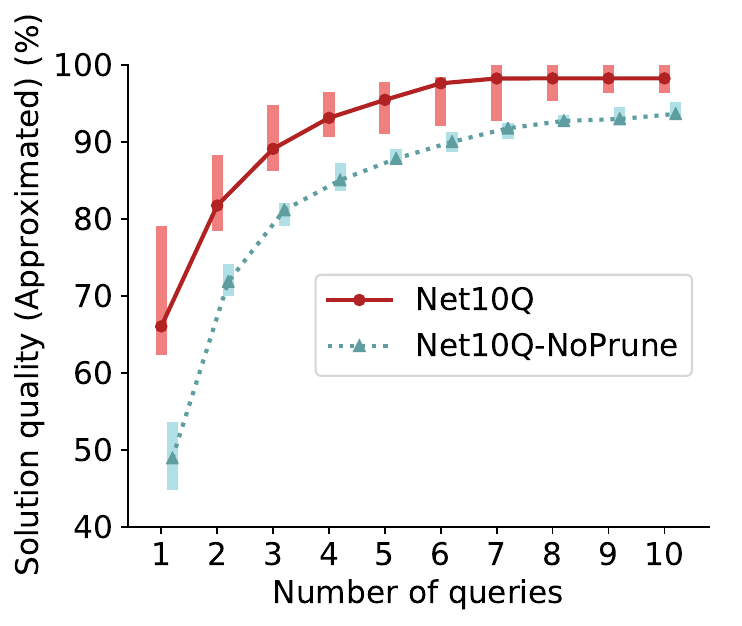}
    \caption{MCF.}
    \label{fig:mcf-grouped}
    \end{subfigure}
    \begin{subfigure}[c]{0.49\textwidth}
    \centering
    \includegraphics[width=\textwidth]{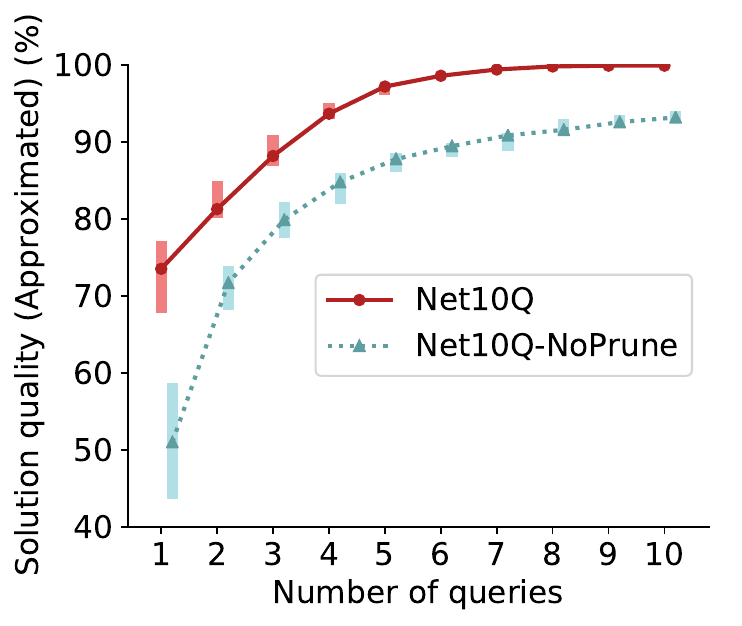}
    \caption{BW.}
    \label{fig:bw-grouped}
    \end{subfigure}
    \begin{subfigure}[c]{0.49\textwidth}
    \centering
    \includegraphics[width=\textwidth]{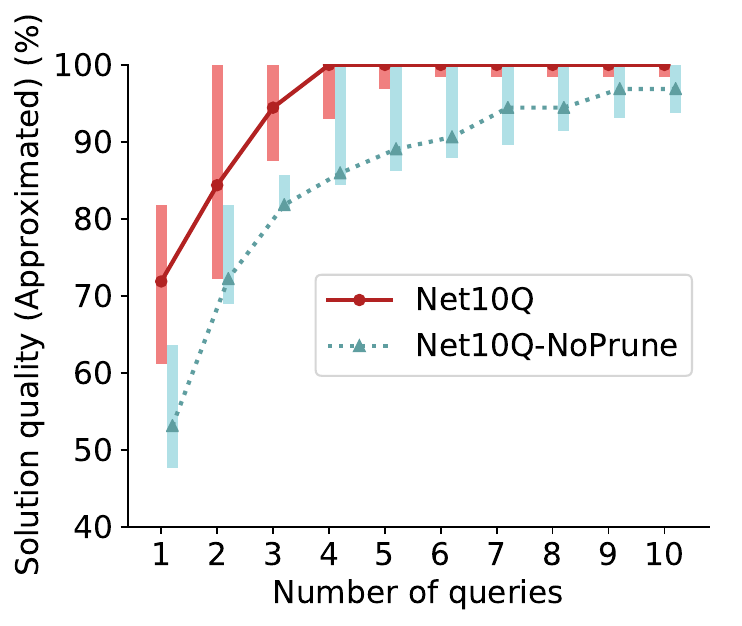}
    \caption{NF.}
    \label{fig:fc-grouped}
    \end{subfigure}
    \begin{subfigure}[c]{0.49\textwidth}
    \centering
    \includegraphics[width=\textwidth]{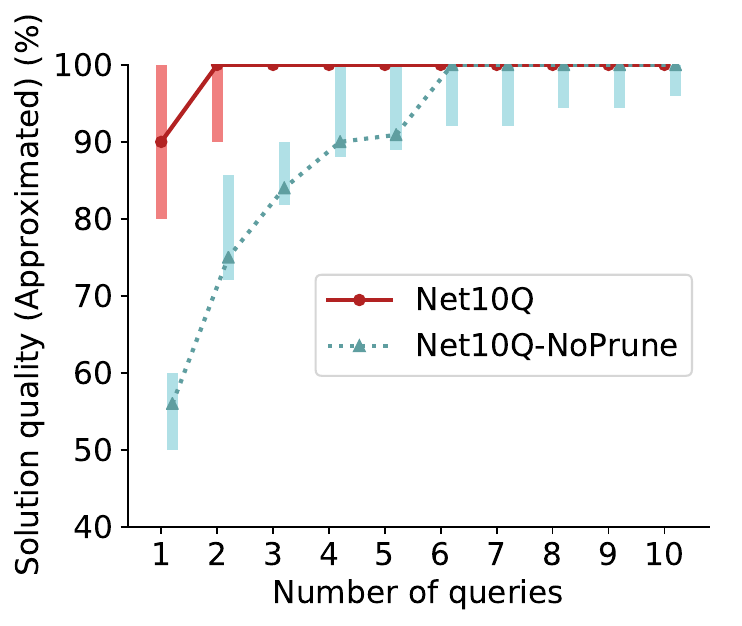}
    \caption{OSPF.}
    \label{fig:ospf-grouped}
    \end{subfigure}
    \caption{Comparing \name{} and \baseline{} with perfect oracle (across all seven topologies). Curves to the left are better. 
    (More detailed, per-topology results for \textbf{NF} is available in Appendix~\ref{app:fc})}
    \label{fig:normal}
    \end{minipage}\hfill
    \begin{minipage}{0.33\textwidth}
    \vspace{-.1in}
    \begin{subfigure}[c]{\columnwidth}
    \centering
    \includegraphics[width=.9\textwidth]{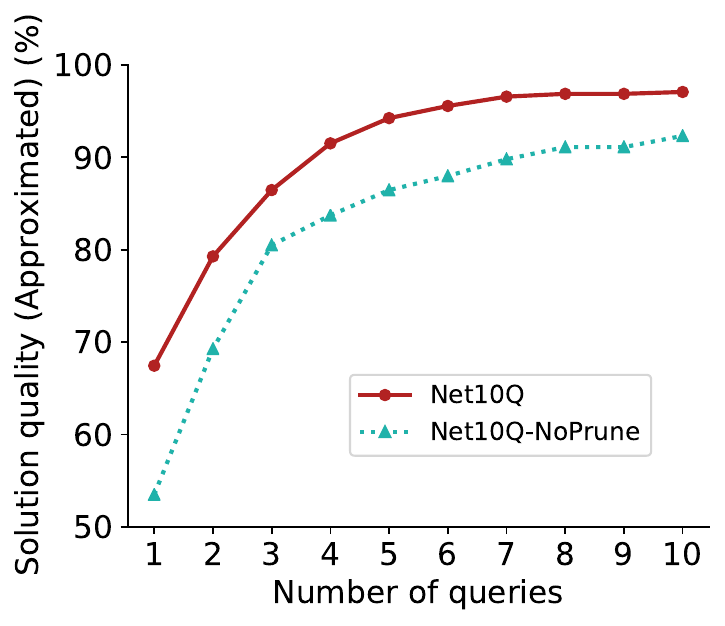}
    \caption{\name{} vs. \baseline{} ($p=10$).}
    \label{fig:robust}
    \end{subfigure}
    \hfill%
    \begin{subfigure}[c]{\columnwidth}
    \centering
    \includegraphics[width=.8\textwidth]{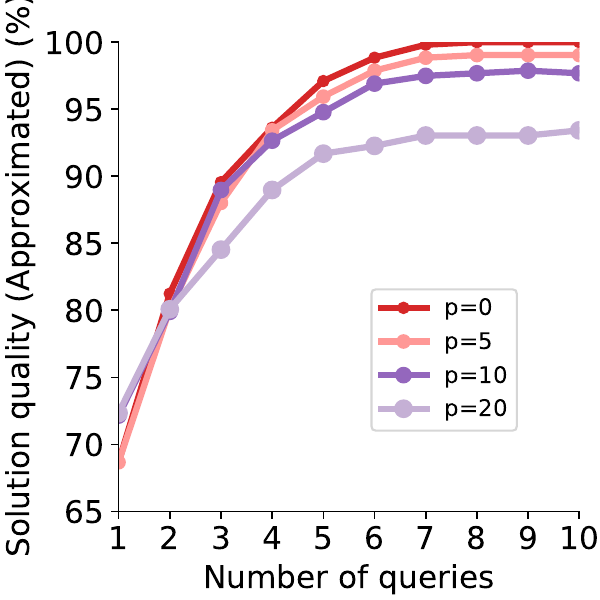}
    \caption{Sensitivity to $p$.}
    \label{fig:sensitivity}
    \end{subfigure}
    \vspace{-.1in}
    \caption{Performance of \name{} with imperfect oracle (\textbf{BW} on CWIX).
    \label{fig:imperfect}
    }
    \end{minipage}
\end{figure}

\subsubsection*{Evaluation on imperfect teacher}
We also adapt the oracle to simulate imperfect teachers
whose responses are potentially inconsistent, based on an error model described below.
When an allocation candidate is presented, the imperfect oracle assigns a random reward that is sampled from a normal distribution, whose expectation is the true reward under corresponding ground truth objective. The standard deviation is $p$ percentage of the distance between average reward and optimal reward under ground truth objective. 


Fig~\ref{fig:imperfect} shows our experimental results with imperfect teacher on \textbf{BW} with the CWIX topology. In the interest of space, we defer results of other scenarios to Appendix~\ref{app:exp}, from which we see similar trend. Fig~\ref{fig:robust} compares \name{} with \baseline{} under the inconsistency level $p = 10$. 
\name{} continues to outperform \baseline{}.
Fig~\ref{fig:sensitivity} presents the sensitivity of \name{} on the inconsistency level ($p = 0, 5, 10, 20$).
%
%
Although the solution quality degrades with higher
inconsistency $p$, \name{} achieves relatively high solution quality even when $p$ is as high as 20. 
The results show that \name{} tends to be able to handle moderate feedback inconsistency from an imperfect teacher, although investigating ways to achieve even higher robustness is an interesting area for future work.

\subsubsection*{Runtime and Scaling}
\label{subsec:scaling}
We first discuss the online query time experienced by users.
For every synthesis task mentioned above,
and across all topologies, the average running time spent by \name{} for each interactive user query is less than 0.15 seconds.
The approach scales well with topology size since a pool of objective-program pairs is created offline.
When creating a pool, the solving time for a single optimization problem is under a second for most topologies on all scenarios on a 2.6 GHz 6-Core Intel Core i7 laptop with 16 GB memory, and we used a pool size of 1000 objectives. The only exception was the \textbf{NF} on the two largest topologies, Deltacom and Ion, which took 11.8 and 15.5 seconds respectively, and we used a smaller pool size in these cases
to limit the pool generation time.
%
%
Note that the pool creation occurs offline. Further, it involves solving multiple distinct optimization problems,
and is trivially parallelizable.

To examine sensitivity to pool size, we first generated 5000 objective-program pairs and then randomly sampled a given number of objective-program pairs to form a candidate pool. Evaluating on candidate pools of size ranging from 10 to 5000, 
we found that pools with 300 objective-program pairs are sufficient for \name{} to achieve over 99\% optimal  after 10 iterations. Please find details in Appendix~\ref{app:pool}.


%% file: Auser_study.tex


\begin{figure}[b]
\vspace{-.1in}
\begin{subfigure}[c]{.25\columnwidth}
  \centering
  \includegraphics[width=.85\columnwidth,trim={0 3cm 0 2cm},clip]{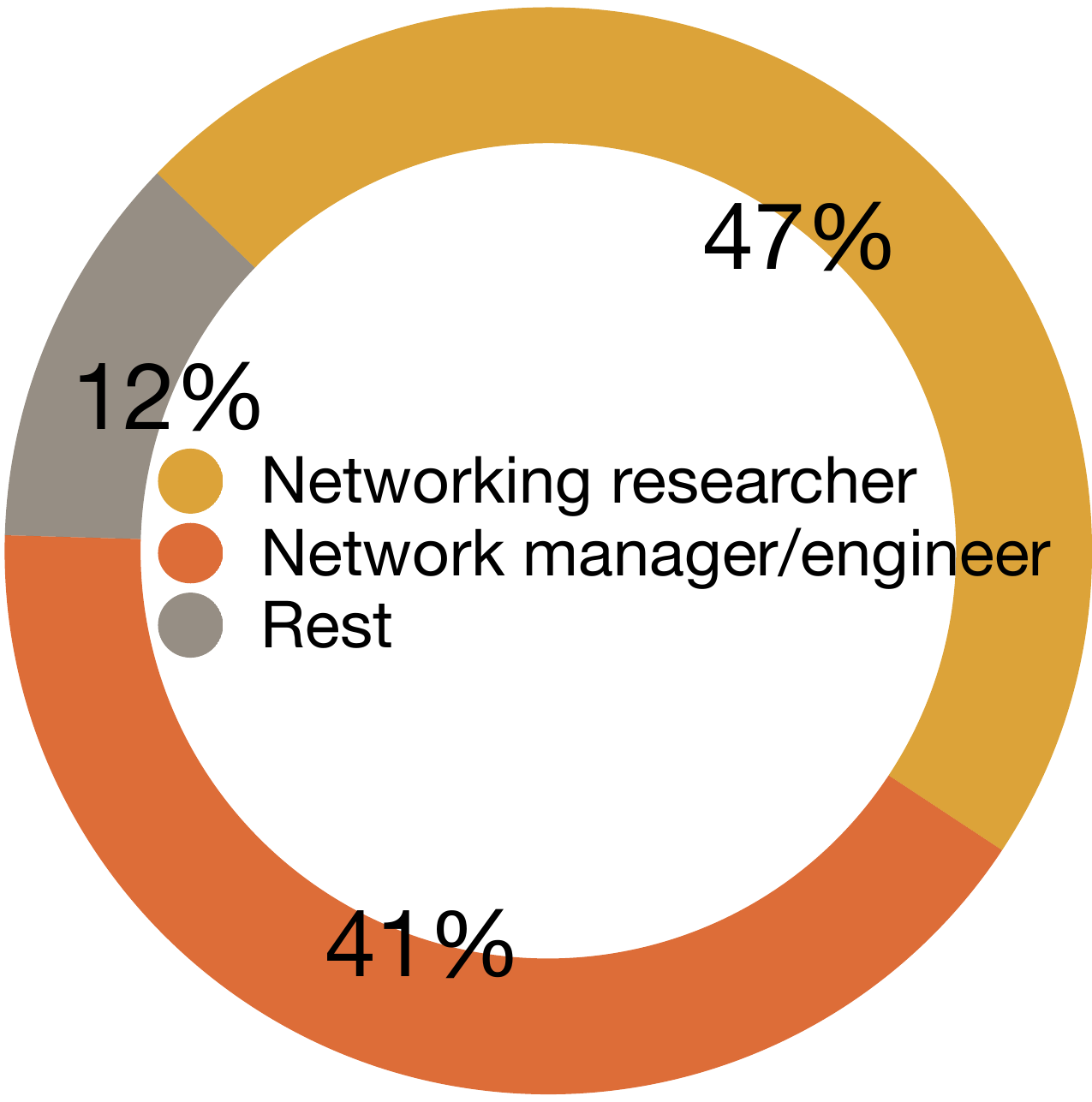}  
  \caption{User background.}
  \label{fig:pre_1}
\end{subfigure}
\begin{subfigure}[c]{.25\columnwidth}
  \centering
  \includegraphics[width=.85\columnwidth,trim={0 3cm 0 2cm},clip]{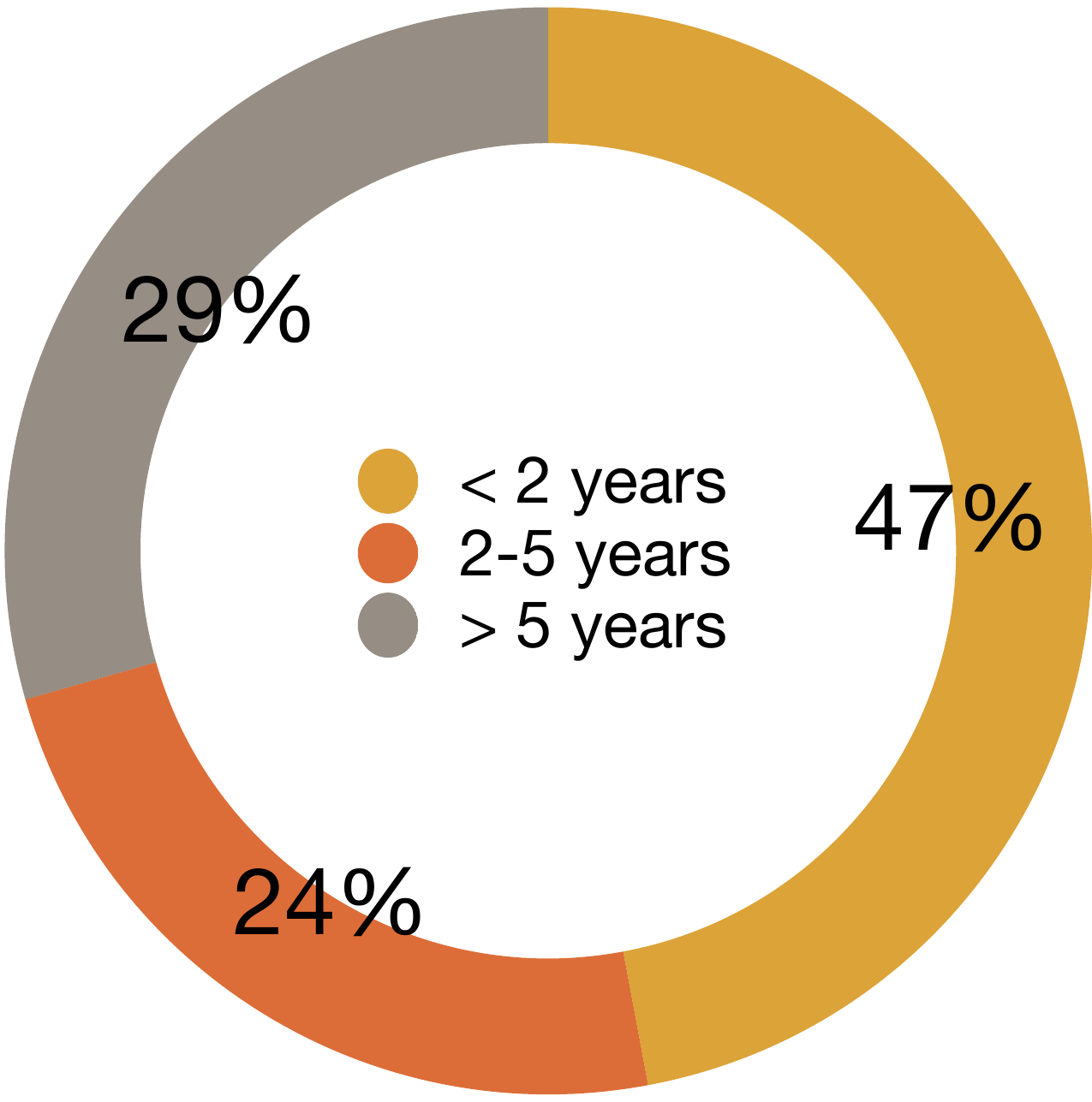}  
  \caption{User experience.}
  \label{fig:pre_2}
\end{subfigure}
\begin{subfigure}[c]{.47\textwidth}
\centering
  \includegraphics[width=\textwidth,trim={0 9cm 0 8cm},clip]{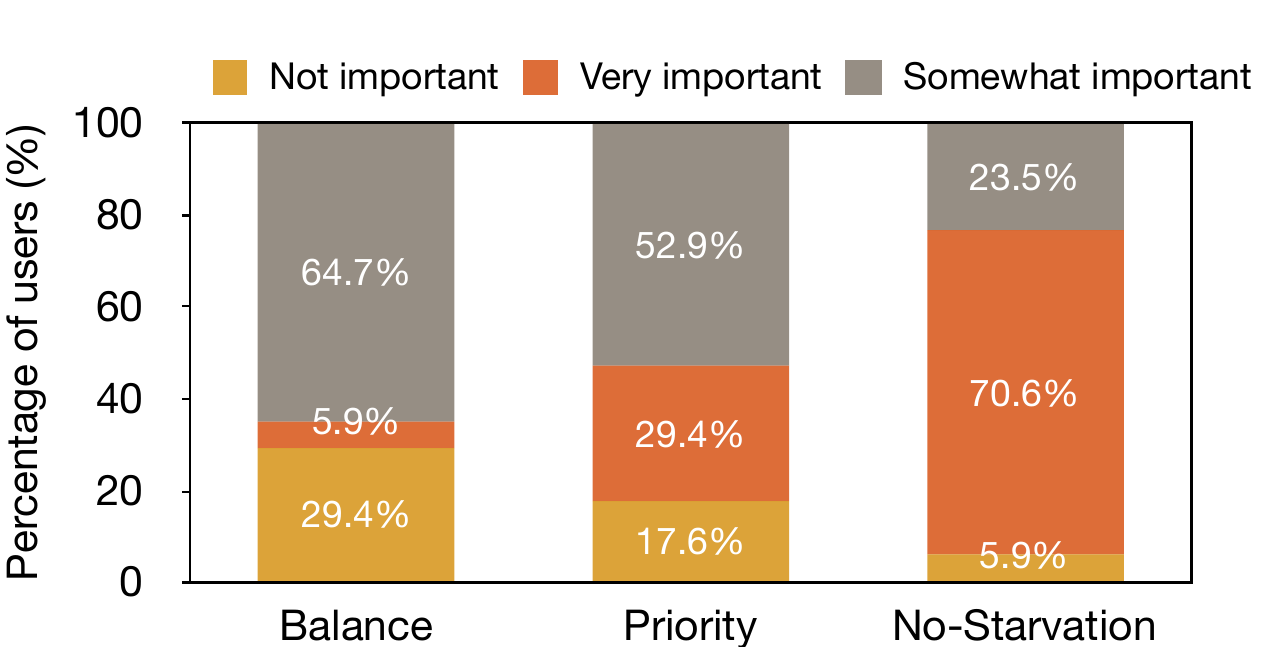}  
  \caption{Diversity of user policies.}
  \label{fig:pre_3}
\end{subfigure}
\vspace{-.15in}
\caption{User background and diversity of chosen policies.}
\label{fig:usr_profile}
\end{figure}

\begin{figure}[t]

\begin{minipage}{.62\textwidth}
  \centering
\begin{subfigure}[b]{.48\columnwidth}
  \centering
  \captionsetup{width=1\textwidth}
  \includegraphics[width=.75\columnwidth,trim={0 3cm 0 2cm},clip]{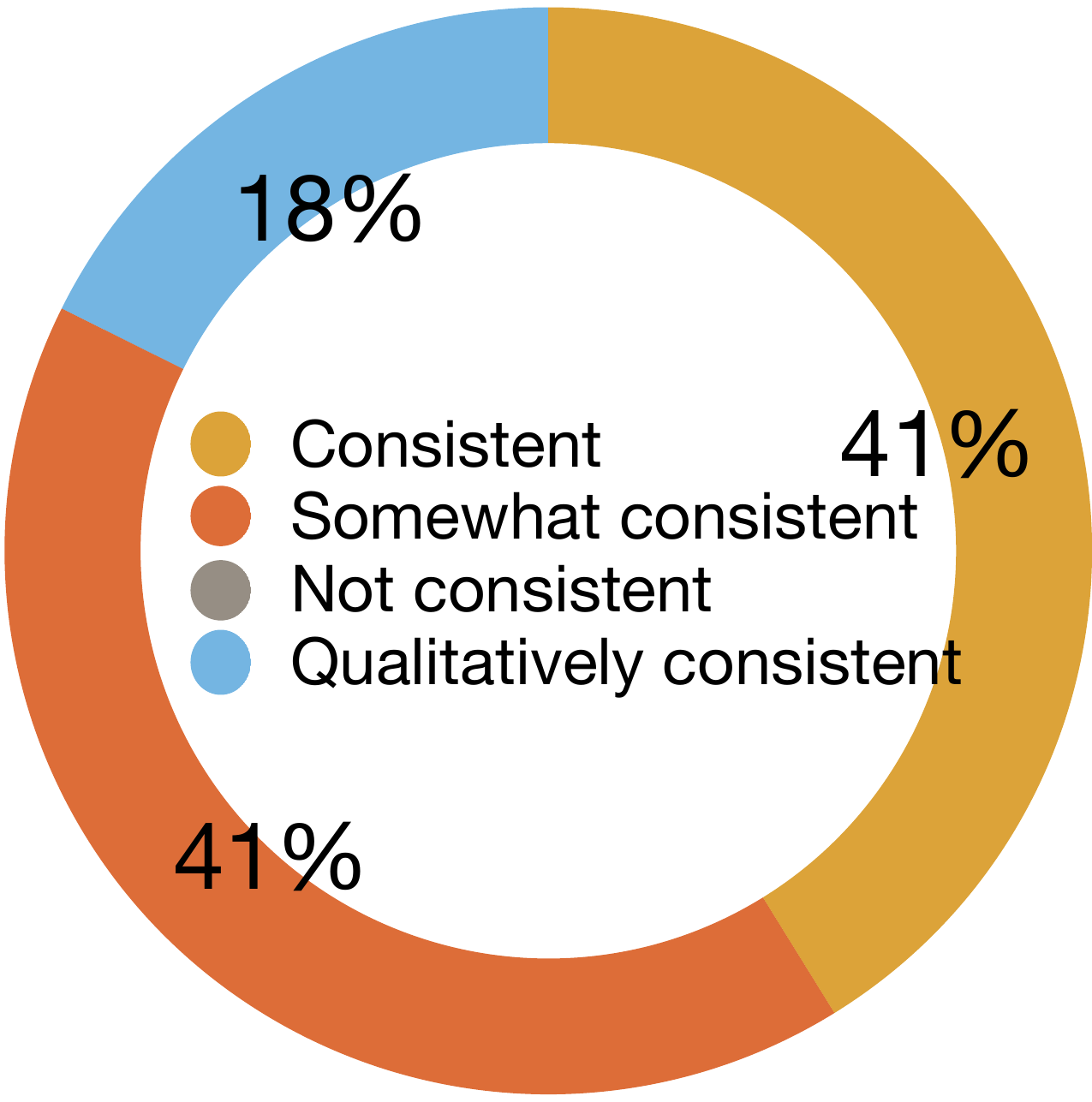}  
  \caption{Quality of recommendations.}
  \label{fig:post_1}
\end{subfigure}
\begin{subfigure}[b]{.48\columnwidth}
\centering
    \captionsetup{width=\textwidth}
  \includegraphics[width=.75\columnwidth,trim={0 3cm 0 2cm},clip]{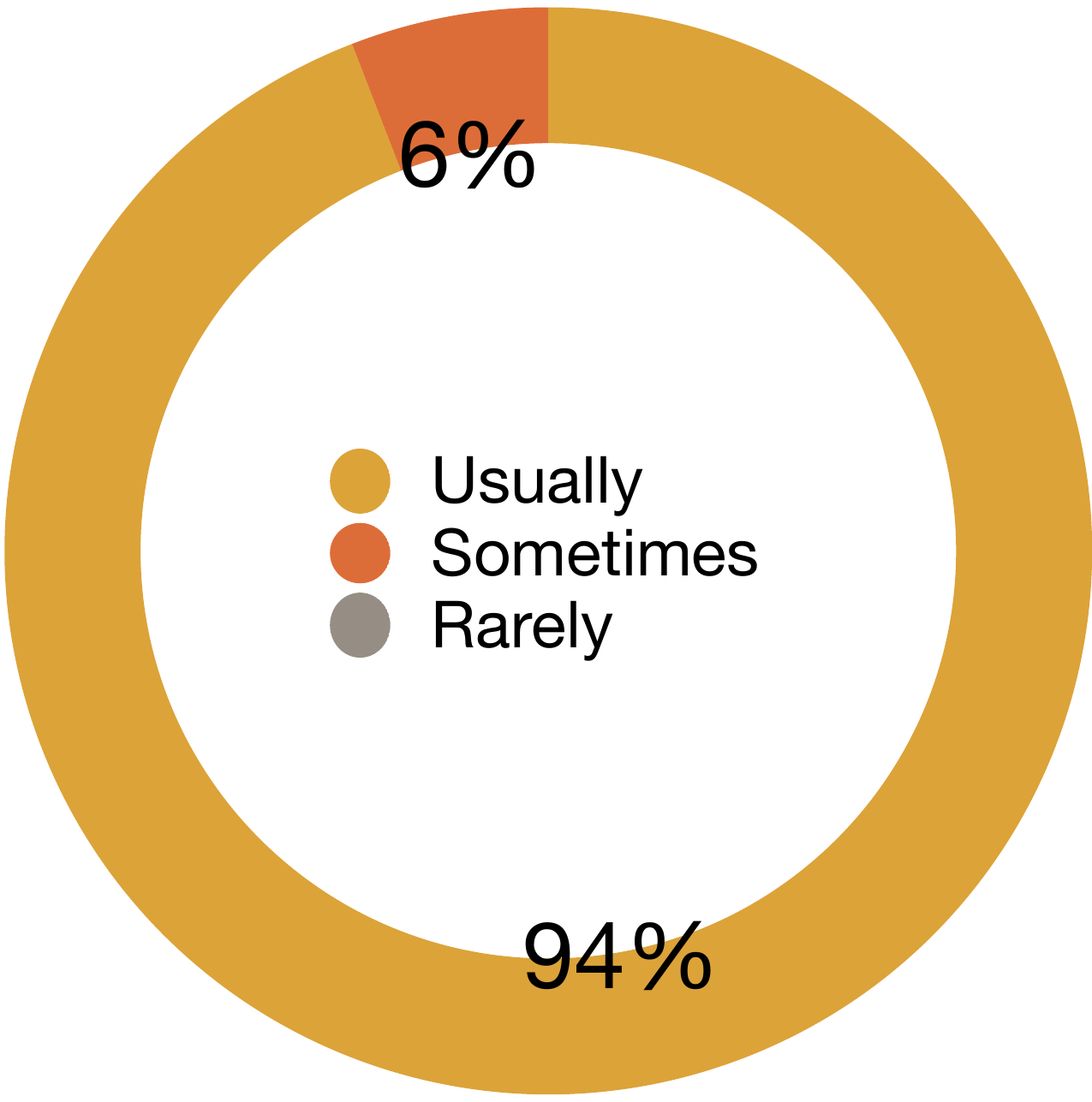}  
  \caption{Was time taken acceptable?}
  \label{fig:post_2}
\end{subfigure}
\vspace{-.1in}
\caption{Feedback on \name{} from real users.}
\label{fig:usr_fb}
\end{minipage}%
\hspace{.1in}
\begin{minipage}{.35\textwidth}
  \centering
  \includegraphics[width=.9\columnwidth]{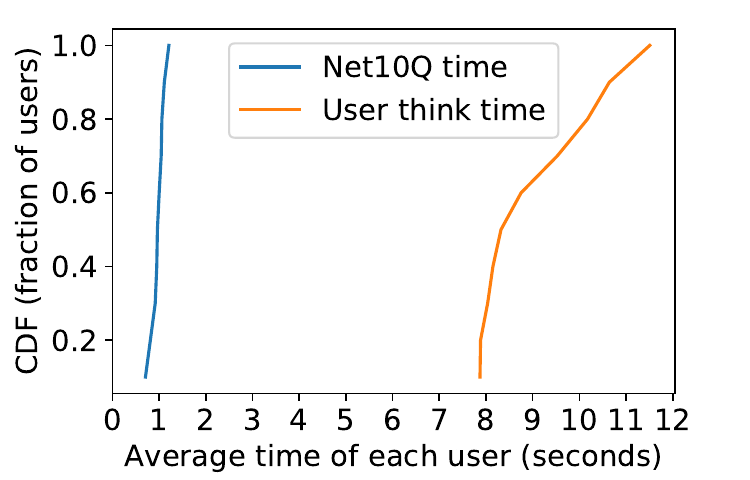}
  \caption{Avg. time per query across users.}
  \label{fig:post_3}
\end{minipage}
\end{figure}


\subsection{Pilot User Study}
\label{sec:userstudy}
We next report on a small-scale study involving 17 users. The primary goal of the user study is to evaluate \name{}  when the real objective is arbitrarily chosen by the user, and even the actual shape is unknown to \name{}. This is in contrast to the oracle experiments where the ground truth objectives are drawn from a template (with only the parameters unknown to \name{}). Further, like with imperfect oracles, users may not always correctly express relative preferences.

The user study was conducted online using an IRB-approved protocol. Participants were recruited with a minimal qualifying requirement being they have taken a university course in computer networking. Figs~\ref{fig:pre_1} and~\ref{fig:pre_2} show the background of users. 88\% of them are computer networking researchers or practitioners. 53\% of users have more than 2 years of experience managing networks.

Our user study used an earlier version of Algorithm~\ref{alg:voting} implemented as an online web application. Specifically, the PCSes were generated on the fly, rather than pre-generated.
To ensure responsiveness, the deployed algorithm set the threshold {\sc Thresh} = 2.
We note that the cloud application for our user study was developed and tested over multiple months, and in parallel to refinements we developed to the algorithm. We were conservative in deploying the latest version given the need for a robust user-facing system, and to ensure all participants saw the same version of the algorithm. 

The study focused on the \textbf{BW} scenario (with four classes of traffic) and the Abilene topology. The user was free to choose any policy on how bandwidth allocations were to be made, and answer queries based on their policy. In each iteration, the user was asked to choose between two different bandwidth allocations generated by \name{}. The user could either pick which allocation was better, or indicate it was too hard to call if the decisions were close. The study terminated after the user answered 10 questions, or when the user indicated she was ready to terminate the study.


In the post-study questionnaire, users were asked to characterize their policy by choosing how important it was to achieve each of three criteria below: (i) \textit{Balance}, indicating allocation across classes is balanced; (ii) \textit{Priority}, indicating how important it is to achieve a solution with more allocation to a higher priority class if a lower priority class does poorly; and (iii) \textit{No-Starvation}, indicating how important it is to ensure lower priority classes get at least some allocation.
Fig~\ref{fig:pre_3} presents a breakdown of user
policies. 70.6\% of users indicated \textit{Balance} were somewhat or very important. 82.3\% of users indicated \textit{Priority} were somewhat or very important, while 94.1\% of users indicated \textit{No-Starvation} was somewhat or very important. The results were consistent with the qualitative description each user provided regarding his/her policy.
Overall, almost all users were seeking to achieve \textit{Balance} and \textit{Priority} avoiding the extremes (starvation) - however, they differed considerably in terms of where they lay in the spectrum based on their qualitative comments.

\subsubsection*{Results}
Fig~\ref{fig:post_1} summarizes how well the recommended allocations generated by \name{} met user policy goals. 82\% of the users indicated that the final recommendation is consistent, or somewhat consistent with their policies. The remaining 18\% of the users took the study before we explicitly added the objective question to ask users to rate how well the recommended policy met their goals. However, the qualitative feedback provided by these users indicated \name{} produced allocations consistent with user goals. For instance, one expert user said: \emph{``The study was well done in my opinion. It put the engineer/architect in a position to make a qualified decision to try and chose the most reasonable outcome.''} 

Fig \ref{fig:post_2} shows that 94\% of users indicated the 
response time with \name{} was usually acceptable, while 6\% indicated the time was sometimes acceptable.
Fig \ref{fig:post_3} shows a cumulative distribution across users of the average \name{} time (i.e., the average time taken by \name{} between receiving the user's choice and presenting the next set of allocations). For comparison, the figure also shows a distribution of the average user think time (i.e., the average time taken by a user between when \name{} presents the options and when the user submits her/his choice).
The time taken by \name{} was hardly a second, and much smaller than the user think time which varied from 8 to 12 seconds, indicating \name{} can be used interactively.

Across all users, \name{} was always able to find a satisfiable objective that met all of the user's preferences (it never needed to
invoke the fallback approach (\S\ref{sec:implementation}) of only considering a subset of user preferences). This indicates users express their preferences in a relatively consistent fashion in practice.
Inconsistent responses may still allow for satisfiable objectives,
however we are unable to characterize this in the absence of the exact ground truth objective.

Overall, the results show the promise of a comparative synthesis approach even when dealing with complex user chosen objectives of unknown shape.  We believe there is potential for further improvements with all the optimizations in Algorithm~\ref{alg:voting}, and other future enhancements.

%% file: Aconclusion.tex
\section{Conclusions}
In this paper, we have presented comparative synthesis for learning near-optimal programs with indeterminate objectives, and applied it to network design.
First, we have developed a formal framework for comparative synthesis through queries with users.
Second, we have developed the first learning algorithm for our framework that combines program search and objective learning, and
seeks to achieve high solution quality with relatively few queries. We proved that the algorithm guarantees the median quality of solutions converges logarithmically to the optimal, or even linearly when the target function space is sortable (a property satisfied by two of our case studies).
Third, we have developed \name{}, a system implementing our approach. 
Experiments show that \name{} only makes half or less queries than the baseline system to achieve the same solution quality, and robustly produces high-quality solutions with inconsistent teachers.
A pilot user study with network practitioners and researchers shows \name{} is effective in finding allocations that meet diverse user policy goals in an interactive fashion.
Overall, the results show the promise of our framework, which we believe can 
help in domains beyond networking in the future.

%% file: main.bbl

\begin{thebibliography}{62}


\ifx \showCODEN    \undefined \def \showCODEN     #1{\unskip}     \fi
\ifx \showDOI      \undefined \def \showDOI       #1{#1}\fi
\ifx \showISBNx    \undefined \def \showISBNx     #1{\unskip}     \fi
\ifx \showISBNxiii \undefined \def \showISBNxiii  #1{\unskip}     \fi
\ifx \showISSN     \undefined \def \showISSN      #1{\unskip}     \fi
\ifx \showLCCN     \undefined \def \showLCCN      #1{\unskip}     \fi
\ifx \shownote     \undefined \def \shownote      #1{#1}          \fi
\ifx \showarticletitle \undefined \def \showarticletitle #1{#1}   \fi
\ifx \showURL      \undefined \def \showURL       {\relax}        \fi
\providecommand\bibfield[2]{#2}
\providecommand\bibinfo[2]{#2}
\providecommand\natexlab[1]{#1}
\providecommand\showeprint[2][]{arXiv:#2}

\bibitem[Ammons et~al\mbox{.}(2002)]%
        {specification-mining}
\bibfield{author}{\bibinfo{person}{Glenn Ammons}, \bibinfo{person}{Rastislav
  Bod\'{\i}k}, {and} \bibinfo{person}{James~R. Larus}.}
  \bibinfo{year}{2002}\natexlab{}.
\newblock \showarticletitle{Mining Specifications}. In
  \bibinfo{booktitle}{\emph{Proceedings of the 29th ACM SIGPLAN-SIGACT
  Symposium on Principles of Programming Languages}} (Portland, Oregon)
  \emph{(\bibinfo{series}{POPL '02})}. \bibinfo{publisher}{Association for
  Computing Machinery}, \bibinfo{address}{New York, NY, USA},
  \bibinfo{pages}{4–16}.
\newblock
\showISBNx{1581134509}
\urldef\tempurl%
\url{https://doi.org/10.1145/503272.503275}
\showDOI{\tempurl}


\bibitem[Angluin(2004)]%
        {angluin2004}
\bibfield{author}{\bibinfo{person}{Dana Angluin}.}
  \bibinfo{year}{2004}\natexlab{}.
\newblock \showarticletitle{Queries revisited}.
\newblock \bibinfo{journal}{\emph{Theoretical Computer Science}}
  \bibinfo{volume}{313}, \bibinfo{number}{2} (\bibinfo{year}{2004}),
  \bibinfo{pages}{175 -- 194}.
\newblock
\showISSN{0304-3975}
\urldef\tempurl%
\url{https://doi.org/10.1016/j.tcs.2003.11.004}
\showDOI{\tempurl}
\newblock
\shownote{Algorithmic Learning Theory}.


\bibitem[Beckett et~al\mbox{.}(2019)]%
        {Beckett2020}
\bibfield{author}{\bibinfo{person}{Ryan Beckett}, \bibinfo{person}{Aarti
  Gupta}, \bibinfo{person}{Ratul Mahajan}, {and} \bibinfo{person}{David
  Walker}.} \bibinfo{year}{2019}\natexlab{}.
\newblock \showarticletitle{Abstract Interpretation of Distributed Network
  Control Planes}.
\newblock \bibinfo{journal}{\emph{Proc. ACM Program. Lang.}}
  \bibinfo{volume}{4}, \bibinfo{number}{POPL}, Article \bibinfo{articleno}{42}
  (\bibinfo{date}{Dec.} \bibinfo{year}{2019}), \bibinfo{numpages}{27}~pages.
\newblock
\urldef\tempurl%
\url{https://doi.org/10.1145/3371110}
\showDOI{\tempurl}


\bibitem[Birkner et~al\mbox{.}(2020)]%
        {config2spec}
\bibfield{author}{\bibinfo{person}{R\"{u}diger Birkner}, \bibinfo{person}{Dana
  Drachsler-Cohen}, \bibinfo{person}{Laurent Vanbever}, {and}
  \bibinfo{person}{Martin Vechev}.} \bibinfo{year}{2020}\natexlab{}.
\newblock \showarticletitle{Config2Spec: Mining Network Specifications from
  Network Configurations}. In \bibinfo{booktitle}{\emph{Proceedings of 17th
  USENIX Symposium on Networked Systems Design and Implementation (NSDI '20)}}.
\newblock


\bibitem[Bogle et~al\mbox{.}(2019)]%
        {cvarSigcomm19}
\bibfield{author}{\bibinfo{person}{Jeremy Bogle}, \bibinfo{person}{Nikhil
  Bhatia}, \bibinfo{person}{Manya Ghobadi}, \bibinfo{person}{Ishai Menache},
  \bibinfo{person}{Nikolaj Bj\o{}rner}, \bibinfo{person}{Asaf Valadarsky},
  {and} \bibinfo{person}{Michael Schapira}.} \bibinfo{year}{2019}\natexlab{}.
\newblock \showarticletitle{TEAVAR: Striking the Right Utilization-Availability
  Balance in WAN Traffic Engineering}. In \bibinfo{booktitle}{\emph{Proceedings
  of the ACM Special Interest Group on Data Communication}} (Beijing, China)
  \emph{(\bibinfo{series}{SIGCOMM '19})}. \bibinfo{publisher}{Association for
  Computing Machinery}, \bibinfo{address}{New York, NY, USA},
  \bibinfo{pages}{29–43}.
\newblock
\showISBNx{9781450359566}
\urldef\tempurl%
\url{https://doi.org/10.1145/3341302.3342069}
\showDOI{\tempurl}


\bibitem[Bornholt et~al\mbox{.}(2016)]%
        {synapse}
\bibfield{author}{\bibinfo{person}{James Bornholt}, \bibinfo{person}{Emina
  Torlak}, \bibinfo{person}{Dan Grossman}, {and} \bibinfo{person}{Luis Ceze}.}
  \bibinfo{year}{2016}\natexlab{}.
\newblock \showarticletitle{Optimizing Synthesis with Metasketches}. In
  \bibinfo{booktitle}{\emph{Proceedings of the 43rd Annual ACM SIGPLAN-SIGACT
  Symposium on Principles of Programming Languages}} (St. Petersburg, FL, USA)
  \emph{(\bibinfo{series}{POPL '16})}. \bibinfo{publisher}{Association for
  Computing Machinery}, \bibinfo{address}{New York, NY, USA},
  \bibinfo{pages}{775--788}.
\newblock
\showISBNx{9781450335492}
\urldef\tempurl%
\url{https://doi.org/10.1145/2837614.2837666}
\showDOI{\tempurl}


\bibitem[Boyd and Vandenberghe(2004)]%
        {convex-optimization}
\bibfield{author}{\bibinfo{person}{Stephen Boyd} {and} \bibinfo{person}{Lieven
  Vandenberghe}.} \bibinfo{year}{2004}\natexlab{}.
\newblock \bibinfo{booktitle}{\emph{Convex Optimization}}.
\newblock \bibinfo{publisher}{Cambridge University Press},
  \bibinfo{address}{USA}.
\newblock
\showISBNx{0521833787}


\bibitem[Chaitin(1975)]%
        {chaitin}
\bibfield{author}{\bibinfo{person}{Gregory~J. Chaitin}.}
  \bibinfo{year}{1975}\natexlab{}.
\newblock \showarticletitle{A Theory of Program Size Formally Identical to
  Information Theory}.
\newblock \bibinfo{journal}{\emph{J. ACM}} \bibinfo{volume}{22},
  \bibinfo{number}{3} (\bibinfo{date}{July} \bibinfo{year}{1975}),
  \bibinfo{pages}{329–340}.
\newblock
\showISSN{0004-5411}
\urldef\tempurl%
\url{https://doi.org/10.1145/321892.321894}
\showDOI{\tempurl}


\bibitem[Chang et~al\mbox{.}(2017)]%
        {nsdiValidation17}
\bibfield{author}{\bibinfo{person}{Yiyang Chang}, \bibinfo{person}{Sanjay Rao},
  {and} \bibinfo{person}{Mohit Tawarmalani}.} \bibinfo{year}{2017}\natexlab{}.
\newblock \showarticletitle{Robust Validation of Network Designs under
  Uncertain Demands and Failures}. In \bibinfo{booktitle}{\emph{14$^{th}$
  {USENIX} Symposium on Networked Systems Design and Implementation ({NSDI})}}.
  \bibinfo{pages}{347--362}.
\newblock


\bibitem[Chaudhuri et~al\mbox{.}(2014)]%
        {Chaudhuri14}
\bibfield{author}{\bibinfo{person}{Swarat Chaudhuri}, \bibinfo{person}{Martin
  Clochard}, {and} \bibinfo{person}{Armando Solar-Lezama}.}
  \bibinfo{year}{2014}\natexlab{}.
\newblock \showarticletitle{Bridging Boolean and Quantitative Synthesis Using
  Smoothed Proof Search}. In \bibinfo{booktitle}{\emph{Proceedings of the 41st
  ACM SIGPLAN-SIGACT Symposium on Principles of Programming Languages}} (San
  Diego, California, USA) \emph{(\bibinfo{series}{POPL '14})}.
  \bibinfo{publisher}{Association for Computing Machinery},
  \bibinfo{address}{New York, NY, USA}, \bibinfo{pages}{207--220}.
\newblock
\showISBNx{9781450325448}
\urldef\tempurl%
\url{https://doi.org/10.1145/2535838.2535859}
\showDOI{\tempurl}


\bibitem[Danna et~al\mbox{.}(2012)]%
        {B4InfocomMaxMinFairness}
\bibfield{author}{\bibinfo{person}{Emilie Danna}, \bibinfo{person}{Subhasree
  Mandal}, {and} \bibinfo{person}{Arjun Singh}.}
  \bibinfo{year}{2012}\natexlab{}.
\newblock \showarticletitle{A practical algorithm for balancing the max-min
  fairness and throughput objectives in traffic engineering}. In
  \bibinfo{booktitle}{\emph{2012 Proceedings IEEE INFOCOM}}.
  \bibinfo{pages}{846--854}.
\newblock
\urldef\tempurl%
\url{https://doi.org/10.1109/INFCOM.2012.6195833}
\showDOI{\tempurl}


\bibitem[de~Moura and Bj{\o}rner(2008)]%
        {Z3}
\bibfield{author}{\bibinfo{person}{Leonardo~Mendon{\c{c}}a de Moura} {and}
  \bibinfo{person}{Nikolaj Bj{\o}rner}.} \bibinfo{year}{2008}\natexlab{}.
\newblock \showarticletitle{{Z3:} An Efficient {SMT} Solver}. In
  \bibinfo{booktitle}{\emph{TACAS'08}} \emph{(\bibinfo{series}{LNCS},
  Vol.~\bibinfo{volume}{4963})}. \bibinfo{publisher}{Springer},
  \bibinfo{pages}{337--340}.
\newblock
\showISBNx{978-3-540-78799-0}
\urldef\tempurl%
\url{https://doi.org/10.1007/978-3-540-78800-3_24}
\showDOI{\tempurl}


\bibitem[Dietterich(2000)]%
        {ensemble-method}
\bibfield{author}{\bibinfo{person}{Thomas~G. Dietterich}.}
  \bibinfo{year}{2000}\natexlab{}.
\newblock \showarticletitle{Ensemble Methods in Machine Learning}.
\newblock \bibinfo{journal}{\emph{Lecture Notes in Computer Science}}
  (\bibinfo{year}{2000}), \bibinfo{pages}{1--15}.
\newblock
\showISBNx{9783540450146}
\showISSN{0302-9743}
\urldef\tempurl%
\url{https://doi.org/10.1007/3-540-45014-9_1}
\showDOI{\tempurl}


\bibitem[Drachsler-Cohen et~al\mbox{.}(2017)]%
        {Drachsler-Cohen17}
\bibfield{author}{\bibinfo{person}{Dana Drachsler-Cohen},
  \bibinfo{person}{Sharon Shoham}, {and} \bibinfo{person}{Eran Yahav}.}
  \bibinfo{year}{2017}\natexlab{}.
\newblock \showarticletitle{Synthesis with Abstract Examples}. In
  \bibinfo{booktitle}{\emph{Computer Aided Verification}},
  \bibfield{editor}{\bibinfo{person}{Rupak Majumdar} {and}
  \bibinfo{person}{Viktor Kun{\v{c}}ak}} (Eds.). \bibinfo{publisher}{Springer
  International Publishing}, \bibinfo{address}{Cham},
  \bibinfo{pages}{254--278}.
\newblock
\showISBNx{978-3-319-63387-9}


\bibitem[Drosos et~al\mbox{.}(2020)]%
        {wrex}
\bibfield{author}{\bibinfo{person}{Ian Drosos}, \bibinfo{person}{Titus Barik},
  \bibinfo{person}{Philip~J. Guo}, \bibinfo{person}{Robert DeLine}, {and}
  \bibinfo{person}{Sumit Gulwani}.} \bibinfo{year}{2020}\natexlab{}.
\newblock \showarticletitle{Wrex: A Unified Programming-by-Example Interaction
  for Synthesizing Readable Code for Data Scientists}. In
  \bibinfo{booktitle}{\emph{Proceedings of the 2020 CHI Conference on Human
  Factors in Computing Systems}} (Honolulu, HI, USA)
  \emph{(\bibinfo{series}{CHI ’20})}. \bibinfo{publisher}{Association for
  Computing Machinery}, \bibinfo{address}{New York, NY, USA},
  \bibinfo{pages}{1–12}.
\newblock
\showISBNx{9781450367080}
\urldef\tempurl%
\url{https://doi.org/10.1145/3313831.3376442}
\showDOI{\tempurl}


\bibitem[El-Hassany et~al\mbox{.}(2017)]%
        {synet}
\bibfield{author}{\bibinfo{person}{Ahmed El-Hassany}, \bibinfo{person}{Petar
  Tsankov}, \bibinfo{person}{Laurent Vanbever}, {and} \bibinfo{person}{Martin
  Vechev}.} \bibinfo{year}{2017}\natexlab{}.
\newblock \showarticletitle{Network-Wide Configuration Synthesis}. In
  \bibinfo{booktitle}{\emph{Computer Aided Verification}},
  \bibfield{editor}{\bibinfo{person}{Rupak Majumdar} {and}
  \bibinfo{person}{Viktor Kun{\v{c}}ak}} (Eds.). \bibinfo{publisher}{Springer
  International Publishing}, \bibinfo{address}{Cham},
  \bibinfo{pages}{261--281}.
\newblock
\showISBNx{978-3-319-63390-9}


\bibitem[El-Hassany et~al\mbox{.}(2018)]%
        {netcomplete}
\bibfield{author}{\bibinfo{person}{Ahmed El-Hassany}, \bibinfo{person}{Petar
  Tsankov}, \bibinfo{person}{Laurent Vanbever}, {and} \bibinfo{person}{Martin
  Vechev}.} \bibinfo{year}{2018}\natexlab{}.
\newblock \showarticletitle{NetComplete: Practical Network-Wide Configuration
  Synthesis with Autocompletion}. In \bibinfo{booktitle}{\emph{15th {USENIX}
  Symposium on Networked Systems Design and Implementation ({NSDI} 18)}}.
  \bibinfo{publisher}{{USENIX} Association}, \bibinfo{address}{Renton, WA},
  \bibinfo{pages}{579--594}.
\newblock
\showISBNx{978-1-931971-43-0}
\urldef\tempurl%
\url{https://www.usenix.org/conference/nsdi18/presentation/el-hassany}
\showURL{%
\tempurl}


\bibitem[Fortz and Thorup(2000)]%
        {OSPFWeights}
\bibfield{author}{\bibinfo{person}{B. Fortz} {and} \bibinfo{person}{M.
  Thorup}.} \bibinfo{year}{2000}\natexlab{}.
\newblock \showarticletitle{Internet traffic engineering by optimizing OSPF
  weights}. In \bibinfo{booktitle}{\emph{INFOCOM 2000. Nineteenth Annual Joint
  Conference of the IEEE Computer and Communications Societies. Proceedings.
  IEEE}}. \bibinfo{pages}{519--528}.
\newblock


\bibitem[Gao et~al\mbox{.}(2019)]%
        {chipmunk}
\bibfield{author}{\bibinfo{person}{Xiangyu Gao}, \bibinfo{person}{Taegyun Kim},
  \bibinfo{person}{Aatish~Kishan Varma}, \bibinfo{person}{Anirudh Sivaraman},
  {and} \bibinfo{person}{Srinivas Narayana}.} \bibinfo{year}{2019}\natexlab{}.
\newblock \showarticletitle{Autogenerating Fast Packet-Processing Code Using
  Program Synthesis}. In \bibinfo{booktitle}{\emph{Proceedings of the 18th ACM
  Workshop on Hot Topics in Networks}} (Princeton, NJ, USA)
  \emph{(\bibinfo{series}{HotNets '19})}. \bibinfo{publisher}{Association for
  Computing Machinery}, \bibinfo{address}{New York, NY, USA},
  \bibinfo{pages}{150–160}.
\newblock
\showISBNx{9781450370202}
\urldef\tempurl%
\url{https://doi.org/10.1145/3365609.3365858}
\showDOI{\tempurl}


\bibitem[Garg et~al\mbox{.}(2014)]%
        {ice-cs}
\bibfield{author}{\bibinfo{person}{Pranav Garg}, \bibinfo{person}{Christof
  L{\"{o}}ding}, \bibinfo{person}{P. Madhusudan}, {and} \bibinfo{person}{Daniel
  Neider}.} \bibinfo{year}{2014}\natexlab{}.
\newblock \showarticletitle{{ICE:} {A} Robust Framework for Learning
  Invariants}. In \bibinfo{booktitle}{\emph{CAV'14}}. \bibinfo{pages}{69--87}.
\newblock
\urldef\tempurl%
\url{https://doi.org/10.1007/978-3-319-08867-9_5}
\showDOI{\tempurl}


\bibitem[Gautschi(1997)]%
        {convergence}
\bibfield{author}{\bibinfo{person}{Walter Gautschi}.}
  \bibinfo{year}{1997}\natexlab{}.
\newblock \bibinfo{booktitle}{\emph{Numerical analysis: an introduction}}.
\newblock \bibinfo{publisher}{Birkh{\"a}user}, Chapter~4, \bibinfo{pages}{215}.
\newblock


\bibitem[Ghosh et~al\mbox{.}(2013)]%
        {jrex:JSAC13}
\bibfield{author}{\bibinfo{person}{A. Ghosh}, \bibinfo{person}{Sangtae Ha},
  \bibinfo{person}{E. Crabbe}, {and} \bibinfo{person}{J. Rexford}.}
  \bibinfo{year}{2013}\natexlab{}.
\newblock \showarticletitle{Scalable Multi-Class Traffic Management in Data
  Center Backbone Networks}.
\newblock \bibinfo{journal}{\emph{IEEE Journal on Selected Areas in
  Communications}}  \bibinfo{volume}{31} (\bibinfo{year}{2013}),
  \bibinfo{pages}{2673--2684}.
\newblock


\bibitem[Gulwani et~al\mbox{.}(2019)]%
        {gulwani2019quantitative}
\bibfield{author}{\bibinfo{person}{Sumit Gulwani}, \bibinfo{person}{Kunal
  Pathak}, \bibinfo{person}{Arjun Radhakrishna}, \bibinfo{person}{Ashish
  Tiwari}, {and} \bibinfo{person}{Abhishek Udupa}.}
  \bibinfo{year}{2019}\natexlab{}.
\newblock \bibinfo{title}{Quantitative Programming by Examples}.
\newblock
\newblock
\showeprint[arxiv]{1909.05964}~[cs.PL]


\bibitem[Gurobi~Optimization(2020)]%
        {gurobi}
\bibfield{author}{\bibinfo{person}{LLC Gurobi~Optimization}.}
  \bibinfo{year}{2020}\natexlab{}.
\newblock \bibinfo{title}{Gurobi Optimizer Reference Manual}.
\newblock
\newblock
\urldef\tempurl%
\url{http://www.gurobi.com}
\showURL{%
\tempurl}


\bibitem[Gvozdiev et~al\mbox{.}(2018)]%
        {LatencyUtilizationSigcomm18}
\bibfield{author}{\bibinfo{person}{Nikola Gvozdiev}, \bibinfo{person}{Stefano
  Vissicchio}, \bibinfo{person}{Brad Karp}, {and} \bibinfo{person}{Mark
  Handley}.} \bibinfo{year}{2018}\natexlab{}.
\newblock \showarticletitle{On Low-Latency-Capable Topologies, and Their Impact
  on the Design of Intra-Domain Routing}. In
  \bibinfo{booktitle}{\emph{Proceedings of the 2018 Conference of the ACM
  Special Interest Group on Data Communication}} (Budapest, Hungary)
  \emph{(\bibinfo{series}{SIGCOMM '18})}. \bibinfo{publisher}{Association for
  Computing Machinery}, \bibinfo{address}{New York, NY, USA},
  \bibinfo{pages}{88–102}.
\newblock
\showISBNx{9781450355674}
\urldef\tempurl%
\url{https://doi.org/10.1145/3230543.3230575}
\showDOI{\tempurl}


\bibitem[Hong et~al\mbox{.}(2013)]%
        {swan}
\bibfield{author}{\bibinfo{person}{Chi-Yao Hong}, \bibinfo{person}{Srikanth
  Kandula}, \bibinfo{person}{Ratul Mahajan}, \bibinfo{person}{Ming Zhang},
  \bibinfo{person}{Vijay Gill}, \bibinfo{person}{Mohan Nanduri}, {and}
  \bibinfo{person}{Roger Wattenhofer}.} \bibinfo{year}{2013}\natexlab{}.
\newblock \showarticletitle{Achieving High Utilization with Software-driven
  WAN}. In \bibinfo{booktitle}{\emph{Proceedings of the ACM SIGCOMM 2013
  Conference on SIGCOMM}} (Hong Kong, China) \emph{(\bibinfo{series}{SIGCOMM
  '13})}. \bibinfo{publisher}{ACM}, \bibinfo{address}{New York, NY, USA},
  \bibinfo{pages}{15--26}.
\newblock
\showISBNx{978-1-4503-2056-6}
\urldef\tempurl%
\url{https://doi.org/10.1145/2486001.2486012}
\showDOI{\tempurl}


\bibitem[Hu and D'Antoni(2018)]%
        {qsygus}
\bibfield{author}{\bibinfo{person}{Qinheping Hu} {and} \bibinfo{person}{Loris
  D'Antoni}.} \bibinfo{year}{2018}\natexlab{}.
\newblock \showarticletitle{Syntax-Guided Synthesis with Quantitative Syntactic
  Objectives}. In \bibinfo{booktitle}{\emph{Computer Aided Verification}},
  \bibfield{editor}{\bibinfo{person}{Hana Chockler} {and}
  \bibinfo{person}{Georg Weissenbacher}} (Eds.). \bibinfo{publisher}{Springer
  International Publishing}, \bibinfo{address}{Cham},
  \bibinfo{pages}{386--403}.
\newblock
\showISBNx{978-3-319-96145-3}


\bibitem[Jain et~al\mbox{.}(2013)]%
        {b4}
\bibfield{author}{\bibinfo{person}{Sushant Jain}, \bibinfo{person}{Alok Kumar},
  \bibinfo{person}{Subhasree Mandal}, \bibinfo{person}{Joon Ong},
  \bibinfo{person}{Leon Poutievski}, \bibinfo{person}{Arjun Singh},
  \bibinfo{person}{Subbaiah Venkata}, \bibinfo{person}{Jim Wanderer},
  \bibinfo{person}{Junlan Zhou}, \bibinfo{person}{Min Zhu}, {and}
  \bibinfo{person}{et al.}} \bibinfo{year}{2013}\natexlab{}.
\newblock \showarticletitle{B4: Experience with a Globally-Deployed Software
  Defined Wan}.
\newblock \bibinfo{journal}{\emph{SIGCOMM Comput. Commun. Rev.}}
  \bibinfo{volume}{43}, \bibinfo{number}{4} (\bibinfo{date}{Aug.}
  \bibinfo{year}{2013}), \bibinfo{pages}{3–14}.
\newblock
\showISSN{0146-4833}
\urldef\tempurl%
\url{https://doi.org/10.1145/2534169.2486019}
\showDOI{\tempurl}


\bibitem[Jha et~al\mbox{.}(2010)]%
        {jha10}
\bibfield{author}{\bibinfo{person}{Susmit Jha}, \bibinfo{person}{Sumit
  Gulwani}, \bibinfo{person}{Sanjit~A. Seshia}, {and} \bibinfo{person}{Ashish
  Tiwari}.} \bibinfo{year}{2010}\natexlab{}.
\newblock \showarticletitle{Oracle-Guided Component-Based Program Synthesis}.
  In \bibinfo{booktitle}{\emph{Proceedings of the 32nd ACM/IEEE International
  Conference on Software Engineering - Volume 1}} (Cape Town, South Africa)
  \emph{(\bibinfo{series}{ICSE '10})}. \bibinfo{publisher}{Association for
  Computing Machinery}, \bibinfo{address}{New York, NY, USA},
  \bibinfo{pages}{215--224}.
\newblock
\showISBNx{9781605587196}
\urldef\tempurl%
\url{https://doi.org/10.1145/1806799.1806833}
\showDOI{\tempurl}


\bibitem[Jha and Seshia(2017)]%
        {ogis}
\bibfield{author}{\bibinfo{person}{Susmit Jha} {and} \bibinfo{person}{Sanjit~A.
  Seshia}.} \bibinfo{year}{2017}\natexlab{}.
\newblock \showarticletitle{A theory of formal synthesis via inductive
  learning}.
\newblock \bibinfo{journal}{\emph{Acta Informatica}} \bibinfo{volume}{54},
  \bibinfo{number}{7} (\bibinfo{date}{Feb} \bibinfo{year}{2017}),
  \bibinfo{pages}{693--726}.
\newblock
\showISSN{1432-0525}
\urldef\tempurl%
\url{https://doi.org/10.1007/s00236-017-0294-5}
\showDOI{\tempurl}


\bibitem[Ji et~al\mbox{.}(2020)]%
        {Ji2020}
\bibfield{author}{\bibinfo{person}{Ruyi Ji}, \bibinfo{person}{Jingjing Liang},
  \bibinfo{person}{Yingfei Xiong}, \bibinfo{person}{Lu Zhang}, {and}
  \bibinfo{person}{Zhenjiang Hu}.} \bibinfo{year}{2020}\natexlab{}.
\newblock \showarticletitle{Question Selection for Interactive Program
  Synthesis}. In \bibinfo{booktitle}{\emph{Proceedings of the 41st ACM SIGPLAN
  Conference on Programming Language Design and Implementation}} (London, UK)
  \emph{(\bibinfo{series}{PLDI 2020})}. \bibinfo{publisher}{Association for
  Computing Machinery}, \bibinfo{address}{New York, NY, USA},
  \bibinfo{pages}{1143–1158}.
\newblock
\showISBNx{9781450376136}
\urldef\tempurl%
\url{https://doi.org/10.1145/3385412.3386025}
\showDOI{\tempurl}


\bibitem[Jiang et~al\mbox{.}(2020)]%
        {pcf_sigcomm20}
\bibfield{author}{\bibinfo{person}{Chuan Jiang}, \bibinfo{person}{Sanjay Rao},
  {and} \bibinfo{person}{Mohit Tawarmalani}.} \bibinfo{year}{2020}\natexlab{}.
\newblock \showarticletitle{PCF: Provably Resilient Flexible Routing}. In
  \bibinfo{booktitle}{\emph{Proceedings of the Annual Conference of the ACM
  Special Interest Group on Data Communication on the Applications,
  Technologies, Architectures, and Protocols for Computer Communication}}
  (Virtual Event, USA) \emph{(\bibinfo{series}{SIGCOMM '20})}.
  \bibinfo{publisher}{Association for Computing Machinery},
  \bibinfo{address}{New York, NY, USA}, \bibinfo{pages}{139–153}.
\newblock
\showISBNx{9781450379557}
\urldef\tempurl%
\url{https://doi.org/10.1145/3387514.3405858}
\showDOI{\tempurl}


\bibitem[{Knight} et~al\mbox{.}(2011)]%
        {abilene}
\bibfield{author}{\bibinfo{person}{S. {Knight}}, \bibinfo{person}{H.~X.
  {Nguyen}}, \bibinfo{person}{N. {Falkner}}, \bibinfo{person}{R. {Bowden}},
  {and} \bibinfo{person}{M. {Roughan}}.} \bibinfo{year}{2011}\natexlab{}.
\newblock \showarticletitle{The Internet Topology Zoo}.
\newblock \bibinfo{journal}{\emph{IEEE Journal on Selected Areas in
  Communications}} \bibinfo{volume}{29}, \bibinfo{number}{9}
  (\bibinfo{date}{October} \bibinfo{year}{2011}), \bibinfo{pages}{1765--1775}.
\newblock
\showISSN{1558-0008}
\urldef\tempurl%
\url{https://doi.org/10.1109/JSAC.2011.111002}
\showDOI{\tempurl}


\bibitem[Kumar et~al\mbox{.}(2015)]%
        {hierarchicalBW}
\bibfield{author}{\bibinfo{person}{Alok Kumar}, \bibinfo{person}{Sushant Jain},
  \bibinfo{person}{Uday Naik}, \bibinfo{person}{Anand Raghuraman},
  \bibinfo{person}{Nikhil Kasinadhuni}, \bibinfo{person}{Enrique~Cauich
  Zermeno}, \bibinfo{person}{C.~Stephen Gunn}, \bibinfo{person}{Jing Ai},
  \bibinfo{person}{Bj\"{o}rn Carlin}, \bibinfo{person}{Mihai
  Amarandei-Stavila}, {and} \bibinfo{person}{et al.}}
  \bibinfo{year}{2015}\natexlab{}.
\newblock \showarticletitle{BwE: Flexible, Hierarchical Bandwidth Allocation
  for WAN Distributed Computing}.
\newblock \bibinfo{journal}{\emph{SIGCOMM Comput. Commun. Rev.}}
  \bibinfo{volume}{45}, \bibinfo{number}{4} (\bibinfo{date}{Aug.}
  \bibinfo{year}{2015}), \bibinfo{pages}{1–14}.
\newblock
\showISSN{0146-4833}
\urldef\tempurl%
\url{https://doi.org/10.1145/2829988.2787478}
\showDOI{\tempurl}


\bibitem[Kumar et~al\mbox{.}(2018)]%
        {semioblivious-nsdi18}
\bibfield{author}{\bibinfo{person}{Praveen Kumar}, \bibinfo{person}{Yang Yuan},
  \bibinfo{person}{Chris Yu}, \bibinfo{person}{Nate Foster},
  \bibinfo{person}{Robert Kleinberg}, \bibinfo{person}{Petr Lapukhov},
  \bibinfo{person}{Chiun~Lin Lim}, {and} \bibinfo{person}{Robert Soul{\'e}}.}
  \bibinfo{year}{2018}\natexlab{}.
\newblock \showarticletitle{Semi-Oblivious Traffic Engineering: The Road Not
  Taken}. In \bibinfo{booktitle}{\emph{15th {USENIX} Symposium on Networked
  Systems Design and Implementation ({NSDI} 18)}}. \bibinfo{publisher}{{USENIX}
  Association}, \bibinfo{address}{Renton, WA}, \bibinfo{pages}{157--170}.
\newblock
\showISBNx{978-1-939133-01-4}
\urldef\tempurl%
\url{https://www.usenix.org/conference/nsdi18/presentation/kumar}
\showURL{%
\tempurl}


\bibitem[Li et~al\mbox{.}(2014)]%
        {symba}
\bibfield{author}{\bibinfo{person}{Yi Li}, \bibinfo{person}{Aws Albarghouthi},
  \bibinfo{person}{Zachary Kincaid}, \bibinfo{person}{Arie Gurfinkel}, {and}
  \bibinfo{person}{Marsha Chechik}.} \bibinfo{year}{2014}\natexlab{}.
\newblock \showarticletitle{Symbolic Optimization with SMT Solvers}. In
  \bibinfo{booktitle}{\emph{Proceedings of the 41st ACM SIGPLAN-SIGACT
  Symposium on Principles of Programming Languages}} (San Diego, California,
  USA) \emph{(\bibinfo{series}{POPL '14})}. \bibinfo{publisher}{Association for
  Computing Machinery}, \bibinfo{address}{New York, NY, USA},
  \bibinfo{pages}{607--618}.
\newblock
\showISBNx{9781450325448}
\urldef\tempurl%
\url{https://doi.org/10.1145/2535838.2535857}
\showDOI{\tempurl}


\bibitem[Liu et~al\mbox{.}(2014)]%
        {ffc_sigcomm14}
\bibfield{author}{\bibinfo{person}{Hongqiang~Harry Liu},
  \bibinfo{person}{Srikanth Kandula}, \bibinfo{person}{Ratul Mahajan},
  \bibinfo{person}{Ming Zhang}, {and} \bibinfo{person}{David Gelernter}.}
  \bibinfo{year}{2014}\natexlab{}.
\newblock \showarticletitle{Traffic Engineering with Forward Fault Correction}.
  In \bibinfo{booktitle}{\emph{Proceedings of the 2014 ACM Conference on
  SIGCOMM}} (Chicago, Illinois, USA) \emph{(\bibinfo{series}{SIGCOMM '14})}.
  \bibinfo{publisher}{Association for Computing Machinery},
  \bibinfo{address}{New York, NY, USA}, \bibinfo{pages}{527–538}.
\newblock
\showISBNx{9781450328364}
\urldef\tempurl%
\url{https://doi.org/10.1145/2619239.2626314}
\showDOI{\tempurl}


\bibitem[Mayer et~al\mbox{.}(2015)]%
        {mayer15}
\bibfield{author}{\bibinfo{person}{Mika\"{e}l Mayer}, \bibinfo{person}{Gustavo
  Soares}, \bibinfo{person}{Maxim Grechkin}, \bibinfo{person}{Vu Le},
  \bibinfo{person}{Mark Marron}, \bibinfo{person}{Oleksandr Polozov},
  \bibinfo{person}{Rishabh Singh}, \bibinfo{person}{Benjamin Zorn}, {and}
  \bibinfo{person}{Sumit Gulwani}.} \bibinfo{year}{2015}\natexlab{}.
\newblock \showarticletitle{User Interaction Models for Disambiguation in
  Programming by Example}. In \bibinfo{booktitle}{\emph{Proceedings of the 28th
  Annual ACM Symposium on User Interface Software \& Technology}} (Charlotte,
  NC, USA) \emph{(\bibinfo{series}{UIST '15})}. \bibinfo{publisher}{Association
  for Computing Machinery}, \bibinfo{address}{New York, NY, USA},
  \bibinfo{pages}{291--301}.
\newblock
\showISBNx{9781450337793}
\urldef\tempurl%
\url{https://doi.org/10.1145/2807442.2807459}
\showDOI{\tempurl}


\bibitem[McClurg et~al\mbox{.}(2016)]%
        {McClurgPLDI16}
\bibfield{author}{\bibinfo{person}{Jedidiah McClurg}, \bibinfo{person}{Hossein
  Hojjat}, \bibinfo{person}{Nate Foster}, {and} \bibinfo{person}{Pavol
  \v{C}ern\'{y}}.} \bibinfo{year}{2016}\natexlab{}.
\newblock \showarticletitle{Event-Driven Network Programming}. In
  \bibinfo{booktitle}{\emph{Proceedings of the 37th ACM SIGPLAN Conference on
  Programming Language Design and Implementation}} (Santa Barbara, CA, USA)
  \emph{(\bibinfo{series}{PLDI '16})}. \bibinfo{publisher}{Association for
  Computing Machinery}, \bibinfo{address}{New York, NY, USA},
  \bibinfo{pages}{369–385}.
\newblock
\showISBNx{9781450342612}
\urldef\tempurl%
\url{https://doi.org/10.1145/2908080.2908097}
\showDOI{\tempurl}


\bibitem[McClurg et~al\mbox{.}(2015)]%
        {McClurgPLDI15}
\bibfield{author}{\bibinfo{person}{Jedidiah McClurg}, \bibinfo{person}{Hossein
  Hojjat}, \bibinfo{person}{Pavol \v{C}ern\'{y}}, {and} \bibinfo{person}{Nate
  Foster}.} \bibinfo{year}{2015}\natexlab{}.
\newblock \showarticletitle{Efficient Synthesis of Network Updates}. In
  \bibinfo{booktitle}{\emph{Proceedings of the 36th ACM SIGPLAN Conference on
  Programming Language Design and Implementation}} (Portland, OR, USA)
  \emph{(\bibinfo{series}{PLDI '15})}. \bibinfo{publisher}{Association for
  Computing Machinery}, \bibinfo{address}{New York, NY, USA},
  \bibinfo{pages}{196–207}.
\newblock
\showISBNx{9781450334686}
\urldef\tempurl%
\url{https://doi.org/10.1145/2737924.2737980}
\showDOI{\tempurl}


\bibitem[Miettinen et~al\mbox{.}(2008)]%
        {multiobjective}
\bibfield{author}{\bibinfo{person}{Kaisa Miettinen}, \bibinfo{person}{Francisco
  Ruiz}, {and} \bibinfo{person}{Andrzej~P. Wierzbicki}.}
  \bibinfo{year}{2008}\natexlab{}.
\newblock \showarticletitle{Introduction to Multiobjective Optimization:
  Interactive Approaches}. In \bibinfo{booktitle}{\emph{Multiobjective
  Optimization: Interactive and Evolutionary Approaches}},
  \bibfield{editor}{\bibinfo{person}{J{\"u}rgen Branke},
  \bibinfo{person}{Kalyanmoy Deb}, \bibinfo{person}{Kaisa Miettinen}, {and}
  \bibinfo{person}{Roman S{\l}owi{\'{n}}ski}} (Eds.).
  \bibinfo{publisher}{Springer Berlin Heidelberg}, \bibinfo{address}{Berlin,
  Heidelberg}, \bibinfo{pages}{27--57}.
\newblock
\showISBNx{978-3-540-88908-3}
\urldef\tempurl%
\url{https://doi.org/10.1007/978-3-540-88908-3_2}
\showDOI{\tempurl}


\bibitem[Peleg et~al\mbox{.}(2018)]%
        {peleg18}
\bibfield{author}{\bibinfo{person}{Hila Peleg}, \bibinfo{person}{Sharon
  Shoham}, {and} \bibinfo{person}{Eran Yahav}.}
  \bibinfo{year}{2018}\natexlab{}.
\newblock \showarticletitle{Programming Not Only by Example}. In
  \bibinfo{booktitle}{\emph{Proceedings of the 40th International Conference on
  Software Engineering}} (Gothenburg, Sweden) \emph{(\bibinfo{series}{ICSE
  '18})}. \bibinfo{publisher}{Association for Computing Machinery},
  \bibinfo{address}{New York, NY, USA}, \bibinfo{pages}{1114--1124}.
\newblock
\showISBNx{9781450356381}
\urldef\tempurl%
\url{https://doi.org/10.1145/3180155.3180189}
\showDOI{\tempurl}


\bibitem[Ryzhyk et~al\mbox{.}(2017)]%
        {cocoon}
\bibfield{author}{\bibinfo{person}{Leonid Ryzhyk}, \bibinfo{person}{Nikolaj
  Bj{\o}rner}, \bibinfo{person}{Marco Canini}, \bibinfo{person}{Jean-Baptiste
  Jeannin}, \bibinfo{person}{Cole Schlesinger}, \bibinfo{person}{Douglas~B.
  Terry}, {and} \bibinfo{person}{George Varghese}.}
  \bibinfo{year}{2017}\natexlab{}.
\newblock \showarticletitle{Correct by Construction Networks Using Stepwise
  Refinement}. In \bibinfo{booktitle}{\emph{14th {USENIX} Symposium on
  Networked Systems Design and Implementation ({NSDI} 17)}}.
  \bibinfo{publisher}{{USENIX} Association}, \bibinfo{address}{Boston, MA},
  \bibinfo{pages}{683--698}.
\newblock
\showISBNx{978-1-931971-37-9}
\urldef\tempurl%
\url{https://www.usenix.org/conference/nsdi17/technical-sessions/presentation/ryzhyk}
\showURL{%
\tempurl}


\bibitem[Saha et~al\mbox{.}(2015)]%
        {netgen}
\bibfield{author}{\bibinfo{person}{Shambwaditya Saha},
  \bibinfo{person}{Santhosh Prabhu}, {and} \bibinfo{person}{P. Madhusudan}.}
  \bibinfo{year}{2015}\natexlab{}.
\newblock \showarticletitle{NetGen: Synthesizing Data-plane Configurations for
  Network Policies}. In \bibinfo{booktitle}{\emph{Proceedings of the 1st ACM
  SIGCOMM Symposium on Software Defined Networking Research}} (Santa Clara,
  California) \emph{(\bibinfo{series}{SOSR '15})}. \bibinfo{publisher}{ACM},
  \bibinfo{address}{New York, NY, USA}, Article \bibinfo{articleno}{17},
  \bibinfo{numpages}{6}~pages.
\newblock
\showISBNx{978-1-4503-3451-8}
\urldef\tempurl%
\url{https://doi.org/10.1145/2774993.2775006}
\showDOI{\tempurl}


\bibitem[Schkufza et~al\mbox{.}(2013)]%
        {superoptimization}
\bibfield{author}{\bibinfo{person}{Eric Schkufza}, \bibinfo{person}{Rahul
  Sharma}, {and} \bibinfo{person}{Alex Aiken}.}
  \bibinfo{year}{2013}\natexlab{}.
\newblock \showarticletitle{Stochastic Superoptimization}. In
  \bibinfo{booktitle}{\emph{Proceedings of the Eighteenth International
  Conference on Architectural Support for Programming Languages and Operating
  Systems}} (Houston, Texas, USA) \emph{(\bibinfo{series}{ASPLOS '13})}.
  \bibinfo{publisher}{ACM}, \bibinfo{address}{New York, NY, USA},
  \bibinfo{pages}{305--316}.
\newblock
\showISBNx{978-1-4503-1870-9}
\urldef\tempurl%
\url{https://doi.org/10.1145/2451116.2451150}
\showDOI{\tempurl}


\bibitem[Schkufza et~al\mbox{.}(2014)]%
        {superoptimization2}
\bibfield{author}{\bibinfo{person}{Eric Schkufza}, \bibinfo{person}{Rahul
  Sharma}, {and} \bibinfo{person}{Alex Aiken}.}
  \bibinfo{year}{2014}\natexlab{}.
\newblock \showarticletitle{Stochastic Optimization of Floating-point Programs
  with Tunable Precision}. In \bibinfo{booktitle}{\emph{Proceedings of the 35th
  ACM SIGPLAN Conference on Programming Language Design and Implementation}}
  (Edinburgh, United Kingdom) \emph{(\bibinfo{series}{PLDI '14})}.
  \bibinfo{publisher}{ACM}, \bibinfo{address}{New York, NY, USA},
  \bibinfo{pages}{53--64}.
\newblock
\showISBNx{978-1-4503-2784-8}
\urldef\tempurl%
\url{https://doi.org/10.1145/2594291.2594302}
\showDOI{\tempurl}


\bibitem[Settles(2012)]%
        {active-learning}
\bibfield{author}{\bibinfo{person}{Burr Settles}.}
  \bibinfo{year}{2012}\natexlab{}.
\newblock \showarticletitle{Active Learning}.
\newblock \bibinfo{journal}{\emph{Synthesis Lectures on Artificial Intelligence
  and Machine Learning}} \bibinfo{volume}{6}, \bibinfo{number}{1}
  (\bibinfo{year}{2012}), \bibinfo{pages}{1--114}.
\newblock
\urldef\tempurl%
\url{https://doi.org/10.2200/S00429ED1V01Y201207AIM018}
\showDOI{\tempurl}
\showeprint{https://doi.org/10.2200/S00429ED1V01Y201207AIM018}


\bibitem[Seung et~al\mbox{.}(1992)]%
        {query-by-committee}
\bibfield{author}{\bibinfo{person}{H.~S. Seung}, \bibinfo{person}{M. Opper},
  {and} \bibinfo{person}{H. Sompolinsky}.} \bibinfo{year}{1992}\natexlab{}.
\newblock \showarticletitle{Query by Committee}. In
  \bibinfo{booktitle}{\emph{Proceedings of the Fifth Annual Workshop on
  Computational Learning Theory}} (Pittsburgh, Pennsylvania, USA)
  \emph{(\bibinfo{series}{COLT '92})}. \bibinfo{publisher}{Association for
  Computing Machinery}, \bibinfo{address}{New York, NY, USA},
  \bibinfo{pages}{287–294}.
\newblock
\showISBNx{089791497X}
\urldef\tempurl%
\url{https://doi.org/10.1145/130385.130417}
\showDOI{\tempurl}


\bibitem[Shi et~al\mbox{.}(2021)]%
        {shi:tacas}
\bibfield{author}{\bibinfo{person}{Lei Shi}, \bibinfo{person}{Yahui Li},
  \bibinfo{person}{Boon~Thau Loo}, {and} \bibinfo{person}{Rajeev Alur}.}
  \bibinfo{year}{2021}\natexlab{}.
\newblock \showarticletitle{Network Traffic Classification by Program
  Synthesis}. In \bibinfo{booktitle}{\emph{Tools and Algorithms for the
  Construction and Analysis of Systems}},
  \bibfield{editor}{\bibinfo{person}{Jan~Friso Groote} {and}
  \bibinfo{person}{Kim~Guldstrand Larsen}} (Eds.). \bibinfo{publisher}{Springer
  International Publishing}, \bibinfo{address}{Cham},
  \bibinfo{pages}{430--448}.
\newblock
\showISBNx{978-3-030-72016-2}


\bibitem[Sivaraman et~al\mbox{.}(2016)]%
        {domino}
\bibfield{author}{\bibinfo{person}{Anirudh Sivaraman}, \bibinfo{person}{Alvin
  Cheung}, \bibinfo{person}{Mihai Budiu}, \bibinfo{person}{Changhoon Kim},
  \bibinfo{person}{Mohammad Alizadeh}, \bibinfo{person}{Hari Balakrishnan},
  \bibinfo{person}{George Varghese}, \bibinfo{person}{Nick McKeown}, {and}
  \bibinfo{person}{Steve Licking}.} \bibinfo{year}{2016}\natexlab{}.
\newblock \showarticletitle{Packet Transactions: High-Level Programming for
  Line-Rate Switches}. In \bibinfo{booktitle}{\emph{Proceedings of the 2016 ACM
  SIGCOMM Conference}} (Florianopolis, Brazil) \emph{(\bibinfo{series}{SIGCOMM
  '16})}. \bibinfo{publisher}{Association for Computing Machinery},
  \bibinfo{address}{New York, NY, USA}, \bibinfo{pages}{15–28}.
\newblock
\showISBNx{9781450341936}
\urldef\tempurl%
\url{https://doi.org/10.1145/2934872.2934900}
\showDOI{\tempurl}


\bibitem[Solar-Lezama(2016)]%
        {sketch:manual}
\bibfield{author}{\bibinfo{person}{Armando Solar-Lezama}.}
  \bibinfo{year}{2016}\natexlab{}.
\newblock \bibinfo{booktitle}{\emph{The Sketch Programmers Manual}}.
\newblock
\newblock
\shownote{Version 1.7.2}.


\bibitem[Solar-Lezama et~al\mbox{.}(2006)]%
        {sketch}
\bibfield{author}{\bibinfo{person}{Armando Solar-Lezama},
  \bibinfo{person}{Liviu Tancau}, \bibinfo{person}{Rastislav Bodik},
  \bibinfo{person}{Sanjit Seshia}, {and} \bibinfo{person}{Vijay Saraswat}.}
  \bibinfo{year}{2006}\natexlab{}.
\newblock \showarticletitle{Combinatorial sketching for finite programs}. In
  \bibinfo{booktitle}{\emph{ASPLOS'06}} (San Jose, California, USA).
  \bibinfo{publisher}{ACM}, \bibinfo{pages}{404--415}.
\newblock


\bibitem[Soul{\'e} et~al\mbox{.}(2014)]%
        {merlin}
\bibfield{author}{\bibinfo{person}{Robert Soul{\'e}},
  \bibinfo{person}{Shrutarshi Basu}, \bibinfo{person}{Parisa~Jalili Marandi},
  \bibinfo{person}{Fernando Pedone}, \bibinfo{person}{Robert Kleinberg},
  \bibinfo{person}{Emin~Gun Sirer}, {and} \bibinfo{person}{Nate Foster}.}
  \bibinfo{year}{2014}\natexlab{}.
\newblock \showarticletitle{Merlin: A Language for Provisioning Network
  Resources}. In \bibinfo{booktitle}{\emph{Proceedings of the 10th ACM
  International on Conference on Emerging Networking Experiments and
  Technologies}} (Sydney, Australia) \emph{(\bibinfo{series}{CoNEXT '14})}.
  \bibinfo{publisher}{ACM}, \bibinfo{address}{New York, NY, USA},
  \bibinfo{pages}{213--226}.
\newblock
\showISBNx{978-1-4503-3279-8}
\urldef\tempurl%
\url{https://doi.org/10.1145/2674005.2674989}
\showDOI{\tempurl}


\bibitem[Srikant(2004)]%
        {srikant04}
\bibfield{author}{\bibinfo{person}{Rayadurgam Srikant}.}
  \bibinfo{year}{2004}\natexlab{}.
\newblock \bibinfo{booktitle}{\emph{The Mathematics of Internet Congestion
  Control (Systems and Control: Foundations and Applications)}}.
\newblock \bibinfo{publisher}{SpringerVerlag}.
\newblock
\showISBNx{0817632271}
\urldef\tempurl%
\url{https://doi.org/10.1007/978-0-8176-8216-3}
\showURL{%
\tempurl}


\bibitem[Steffen et~al\mbox{.}(2020)]%
        {netdice}
\bibfield{author}{\bibinfo{person}{Samuel Steffen}, \bibinfo{person}{Timon
  Gehr}, \bibinfo{person}{Petar Tsankov}, \bibinfo{person}{Laurent Vanbever},
  {and} \bibinfo{person}{Martin Vechev}.} \bibinfo{year}{2020}\natexlab{}.
\newblock \showarticletitle{Probabilistic Verification of Network
  Configurations}. In \bibinfo{booktitle}{\emph{Proceedings of the Annual
  Conference of the ACM Special Interest Group on Data Communication on the
  Applications, Technologies, Architectures, and Protocols for Computer
  Communication}} (Virtual Event, USA) \emph{(\bibinfo{series}{SIGCOMM '20})}.
  \bibinfo{publisher}{Association for Computing Machinery},
  \bibinfo{address}{New York, NY, USA}, \bibinfo{pages}{750–764}.
\newblock
\showISBNx{9781450379557}
\urldef\tempurl%
\url{https://doi.org/10.1145/3387514.3405900}
\showDOI{\tempurl}


\bibitem[Subramanian et~al\mbox{.}(2020)]%
        {qarc}
\bibfield{author}{\bibinfo{person}{Kausik Subramanian},
  \bibinfo{person}{Anubhavnidhi Abhashkumar}, \bibinfo{person}{Loris D'Antoni},
  {and} \bibinfo{person}{Aditya Akella}.} \bibinfo{year}{2020}\natexlab{}.
\newblock \showarticletitle{Detecting Network Load Violations for Distributed
  Control Planes}. In \bibinfo{booktitle}{\emph{Proceedings of the 41st ACM
  SIGPLAN Conference on Programming Language Design and Implementation}}
  (London, UK) \emph{(\bibinfo{series}{PLDI 2020})}.
  \bibinfo{publisher}{Association for Computing Machinery},
  \bibinfo{address}{New York, NY, USA}, \bibinfo{pages}{974–988}.
\newblock
\showISBNx{9781450376136}
\urldef\tempurl%
\url{https://doi.org/10.1145/3385412.3385976}
\showDOI{\tempurl}


\bibitem[Subramanian et~al\mbox{.}(2017)]%
        {genesis}
\bibfield{author}{\bibinfo{person}{Kausik Subramanian}, \bibinfo{person}{Loris
  D'Antoni}, {and} \bibinfo{person}{Aditya Akella}.}
  \bibinfo{year}{2017}\natexlab{}.
\newblock \showarticletitle{Genesis: Synthesizing Forwarding Tables in
  Multi-tenant Networks}. In \bibinfo{booktitle}{\emph{Proceedings of the 44th
  ACM SIGPLAN Symposium on Principles of Programming Languages}} (Paris,
  France) \emph{(\bibinfo{series}{POPL 2017})}. \bibinfo{publisher}{ACM},
  \bibinfo{address}{New York, NY, USA}, \bibinfo{pages}{572--585}.
\newblock
\showISBNx{978-1-4503-4660-3}
\urldef\tempurl%
\url{https://doi.org/10.1145/3009837.3009845}
\showDOI{\tempurl}


\bibitem[\v{C}ern\'{y} and Henzinger(2011)]%
        {Cerny2011}
\bibfield{author}{\bibinfo{person}{Pavol \v{C}ern\'{y}} {and}
  \bibinfo{person}{Thomas~A. Henzinger}.} \bibinfo{year}{2011}\natexlab{}.
\newblock \showarticletitle{From Boolean to Quantitative Synthesis}. In
  \bibinfo{booktitle}{\emph{Proceedings of the Ninth ACM International
  Conference on Embedded Software}} (Taipei, Taiwan)
  \emph{(\bibinfo{series}{EMSOFT ’11})}. \bibinfo{publisher}{Association for
  Computing Machinery}, \bibinfo{address}{New York, NY, USA},
  \bibinfo{pages}{149–154}.
\newblock
\showISBNx{9781450307147}
\urldef\tempurl%
\url{https://doi.org/10.1145/2038642.2038666}
\showDOI{\tempurl}


\bibitem[Wang et~al\mbox{.}(2009)]%
        {metricComparison}
\bibfield{author}{\bibinfo{person}{Yi Wang}, \bibinfo{person}{Ioannis
  Avramopoulos}, {and} \bibinfo{person}{Jennifer Rexford}.}
  \bibinfo{year}{2009}\natexlab{}.
\newblock \showarticletitle{Design for configurability: rethinking interdomain
  routing policies from the ground up}.
\newblock \bibinfo{journal}{\emph{IEEE Journal on Selected Areas in
  Communications}} \bibinfo{volume}{27}, \bibinfo{number}{3}
  (\bibinfo{year}{2009}), \bibinfo{pages}{336--348}.
\newblock
\urldef\tempurl%
\url{https://doi.org/10.1109/JSAC.2009.090409}
\showDOI{\tempurl}


\bibitem[Wang et~al\mbox{.}(2019)]%
        {comparative-synthesis}
\bibfield{author}{\bibinfo{person}{Yanjun Wang}, \bibinfo{person}{Chuan Jiang},
  \bibinfo{person}{Xiaokang Qiu}, {and} \bibinfo{person}{Sanjay~G. Rao}.}
  \bibinfo{year}{2019}\natexlab{}.
\newblock \showarticletitle{Learning Network Design Objectives Using A Program
  Synthesis Approach}. In \bibinfo{booktitle}{\emph{Proceedings of the 18th ACM
  Workshop on Hot Topics in Networks}} (Princeton, NJ, USA)
  \emph{(\bibinfo{series}{HotNets '19})}. \bibinfo{publisher}{ACM},
  \bibinfo{address}{New York, NY, USA}, \bibinfo{pages}{69--76}.
\newblock
\showISBNx{978-1-4503-7020-2}
\urldef\tempurl%
\url{https://doi.org/10.1145/3365609.3365861}
\showDOI{\tempurl}


\bibitem[Wang et~al\mbox{.}(2010)]%
        {r3:sigcomm10}
\bibfield{author}{\bibinfo{person}{Ye Wang}, \bibinfo{person}{Hao Wang},
  \bibinfo{person}{Ajay Mahimkar}, \bibinfo{person}{Richard Alimi},
  \bibinfo{person}{Yin Zhang}, \bibinfo{person}{Lili Qiu}, {and}
  \bibinfo{person}{Yang~Richard Yang}.} \bibinfo{year}{2010}\natexlab{}.
\newblock \showarticletitle{R3: Resilient Routing Reconfiguration}. In
  \bibinfo{booktitle}{\emph{Proceedings of the ACM SIGCOMM 2010 Conference}}
  (New Delhi, India) \emph{(\bibinfo{series}{SIGCOMM '10})}.
  \bibinfo{publisher}{Association for Computing Machinery},
  \bibinfo{address}{New York, NY, USA}, \bibinfo{pages}{291–302}.
\newblock
\showISBNx{9781450302012}
\urldef\tempurl%
\url{https://doi.org/10.1145/1851182.1851218}
\showDOI{\tempurl}


\bibitem[Yuan et~al\mbox{.}(2015)]%
        {netegg}
\bibfield{author}{\bibinfo{person}{Yifei Yuan}, \bibinfo{person}{Dong Lin},
  \bibinfo{person}{Rajeev Alur}, {and} \bibinfo{person}{Boon~Thau Loo}.}
  \bibinfo{year}{2015}\natexlab{}.
\newblock \showarticletitle{Scenario-based Programming for SDN Policies}. In
  \bibinfo{booktitle}{\emph{Proceedings of the 11th ACM Conference on Emerging
  Networking Experiments and Technologies}} (Heidelberg, Germany)
  \emph{(\bibinfo{series}{CoNEXT '15})}. \bibinfo{publisher}{ACM},
  \bibinfo{address}{New York, NY, USA}, Article \bibinfo{articleno}{34},
  \bibinfo{numpages}{13}~pages.
\newblock
\showISBNx{978-1-4503-3412-9}
\urldef\tempurl%
\url{https://doi.org/10.1145/2716281.2836119}
\showDOI{\tempurl}


\end{thebibliography}
